\documentclass[12pt,a4paper]{article}
\usepackage[utf8]{inputenc}
\usepackage{color}  
\definecolor{coolblack}{rgb}{0.0, 0.18, 0.39}
\definecolor{midnightblue}{rgb}{0.1, 0.1, 0.44}
\definecolor{prussianblue}{rgb}{0.0, 0.19, 0.33}
\definecolor{oxfordblue}{rgb}{0.0, 0.13, 0.28}
\usepackage{hyperref,float}
\hypersetup{
    colorlinks=true,
    linkcolor=blue,
    filecolor=magenta,      
    urlcolor=blue,
    citecolor =blue,
}

\usepackage{amsmath,amsthm}
\usepackage{amsfonts}
\usepackage{amssymb}
\usepackage{graphicx}
\usepackage{lmodern}
\usepackage[round]{natbib}
\usepackage[onehalfspacing]{setspace}
\usepackage{booktabs}
\usepackage{wrapfig}
\usepackage{textcomp}
\usepackage{pdflscape}
\usepackage[small,bf]{caption}
\usepackage{breqn}
\usepackage{lscape}
\usepackage{pgfplots}
\pgfplotsset{compat=newest}
\usetikzlibrary{plotmarks}
\usepackage{grffile}
\usepgfplotslibrary{external}
\tikzsetnextfilename{\jobname}
\usepackage{adjustbox}
\usepackage{tikz-network}
\usepackage{tikz}
\usetikzlibrary{snakes}

\newcommand{\bX}{\ensuremath{\bm{X}}}
\newcommand{\IZ}{\ensuremath{\mathbb{Z}}}

\newcommand{\IN}{\ensuremath{\mathbb{N}}}

\newcommand{\bPhi}{\ensuremath{\boldsymbol\Phi}}
\newcommand{\bPsi}{\ensuremath{\boldsymbol\Psi}}
\newcommand{\btheta}{\ensuremath{\boldsymbol\theta}}
\newcommand{\bepsilon}{\ensuremath{\boldsymbol\epsilon}}
\newcommand{\bbeta}{\ensuremath{\boldsymbol\eta}}
\newcommand{\bSigma}{\ensuremath{\boldsymbol\Sigma}}
\newcommand{\bOmega}{\ensuremath{\boldsymbol\Omega}}
\newcommand{\bP}{\ensuremath{\boldsymbol P}}

\newcommand{\mC}{\ensuremath{\mathcal{C}}}

\newtheorem{prop}{Proposition}
\newenvironment{prop*}
  {\ex}
  {\endex}

\newtheorem{remark}{Remark}
\newenvironment{remark*}
  {\ex}
  {\endex}

\newenvironment{definition*}
  {\ex}
  {\endex}
  
\usepackage{cleveref}
\usepackage{multirow}
\usepackage{subcaption}
\usepackage[stable]{footmisc}
\usepackage{bm}
\usepackage{lineno}

\usepackage[left=1in,right=1in,top=1.5in,bottom=1.5in]{geometry}
\author{Jozef \textsc{Barun\'{i}k}$^{\ddag}$ and Michael \textsc{Ellington}$^{\dag}$}
\title{\textbf{Dynamic Network Risk}}
\date{\today}

\begin{document}
\begin{titlepage}
\maketitle 
%\begin{center}
%\begin{small}
%\textbf{Preliminary Draft}\\
%\ddag \textit{Institute of Economic Studies, Charles University,}\\ 
%\textit{Opletalova 26, 110 00, }\\
%\textit{and The Czech Academy of Sciences, IITA} \\
%\textit{Pod Vodárenskou Věží 4, 182 08, Prague, Czech Republic} \\
%\href{mailto: barunik@fsv.cuni.cz}{barunik@fsv.cuni.cz}\\
%\vspace{1em}
%\dag \textit{University of Liverpool Management School,}\\
%\textit{Chatham Building, Liverpool, L69 7ZH, UK}\\
%\href{mailto: m.ellington@liverpool.ac.uk}{m.ellington@liverpool.ac.uk}\\
%\end{small}
%\end{center}

\begin{abstract}
\noindent This paper examines the pricing of short-term and long-term dynamic network risk in the cross-section of stock returns. Stocks with high sensitivities to dynamic network risk earn lower returns. We rationalize our finding with economic theory that allows the stochastic discount factor to load on network risk through the precautionary savings channel. A one-standard deviation increase in long-term (short-term) network risk loadings associate with a 7.66\% (6.71\%) drop in annualized expected returns.
%We show how to track such dynamic network connections on a daily basis using time-varying parameter vector auto-regressions.
\end{abstract}



\begin{small}
\begin{singlespace}
\noindent\rule{10cm}{0.4pt}\\
\ddag \textit{Institute of Economic Studies, Charles University, Opletalova 26, 110 00, and The Czech Academy of Sciences, IITA, Pod Vodárenskou Věží 4, 182 08, Prague, Czech Republic}.\\
\href{mailto: barunik@fsv.cuni.cz}{barunik@fsv.cuni.cz}\\
\dag \textit{University of Liverpool Management School, Chatham Building, Liverpool, L69 7ZH, UK}.
\href{mailto: m.ellington@liverpool.ac.uk}{m.ellington@liverpool.ac.uk}\\

\noindent \textbf{Acknowledgements}\\
We thank Lubo\v{s} P\'{a}stor, Lucio Sarno, Oliver Linton, Wolfgang H\"{a}rdle, Luk\'{a}\v{s} V\'{a}cha, Ryland Thomas, Chris Florackis, Costas Milas, Charlie Cai, and Marcin Michalski for invaluable discussions and comments. We are grateful to Lubo\v{s} Hanus for help in furnishing and converting estimation codes. We acknowledge insightful comments from many seminar presentations, including: the Danish National Bank; the 2019 STAT of ML conference; the 13${\text{th}}$ International Conference on Computational and Financial Econometrics; and many more. Jozef Barun\'{i}k gratefully acknowledges support from the Czech Science Foundation under the 19-28231X (EXPRO) project. % For estimation of dynamic horizon specific networks, we provide packages \texttt{DynamicNets.jl} in \textsf{JULIA} and \texttt{DynamicNets} in \textsf{MATLAB}. The packages are available at \url{https://github.com/barunik/DynamicNets.jl} and \url{https://github.com/mte00/DynamicNets}. % Estimation of S\&P500 sectoral constituents were all done on a desktop PC with 64GB DDR4 2400MHz RAM and an Intel Core i7 Six Core Processor i7-8700k (3.7GHz) 12MB Cache.
\end{singlespace}
\end{small}
\end{titlepage}

\noindent \textbf{Disclosure Statement:} Jozef Barun\'{i}k and Michael Ellington have nothing to disclose.

\section{Introduction}\label{intro}

Idiosyncratic stock return volatilities vary over time and exhibit strong co-movement \citep{herskovic2018firm}. The synchronous behavior of return volatilities provides the foundation for a common component explaining cross-sectional variation in stock returns. Economic theory and empirical evidence show that stocks loading positively on volatility risk earn lower returns in equilibrium. This is because they act as inter-temporal hedging devices against future uncertainty \citep{cremers2015aggregate,herskovic2016common}. However, \cite{acemoglu2012network} present evidence that network structures forming from idiosyncratic shock propagation from volatilities can determine systematic fluctuations. The implication here is that network structures forming around volatility connections create a source of risk for investors.

In this paper, we examine the causal nature of dynamic horizon specific shock propagation that determines network structures among idiosyncratic return volatilities; and the premium investors demand for bearing such exposures. Our analysis provides robust empirical evidence that stocks with higher sensitivities to dynamic horizon specific network risk earn lower returns. We rationalize our findings with economic theory, that allows the stochastic discount factor to load on volatility risk through the precautionary savings channel \citep{ang2006cross,campbell2018intertemporal}. In doing so, we propose a mechanism that describes this type of investor behaviour by allowing the stochastic discount factor to load on horizon specific network connections among idiosyncratic return volatilities.

Viewing stocks as nodes within a network allows us to characterize the nature of shocks that determine such risk exposures. From a \cite{lucas1977understanding} perspective, investors are able to diversify away firm level shocks when the network is completely unconnected or completely and equally connected. Therefore, idiosyncratic fluctuations average out and result in negligible aggregate effects. Hence, an investor holding a large well-diversified portfolio has minimal exposure to such shocks, and thus only demands a premium for non-diversifiable systematic sources of risk. However, relaxing the assumption of symmetrical connections breaks down the ability of investors to diversify away firm-level shocks, permits the snow-balling of firm-level shocks \citep{elliott2014financial}, and induces differences in the strength of directional connections.

\begin{figure}[!t]
  \begin{subfigure}{8cm}
    \centering\includegraphics[width=5cm]{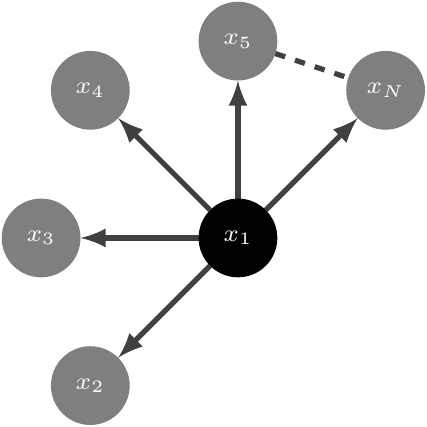}
    \subcaption{Asymmetric network}
  \end{subfigure}  
  \begin{subfigure}{8cm}
    \centering\includegraphics[width=5cm]{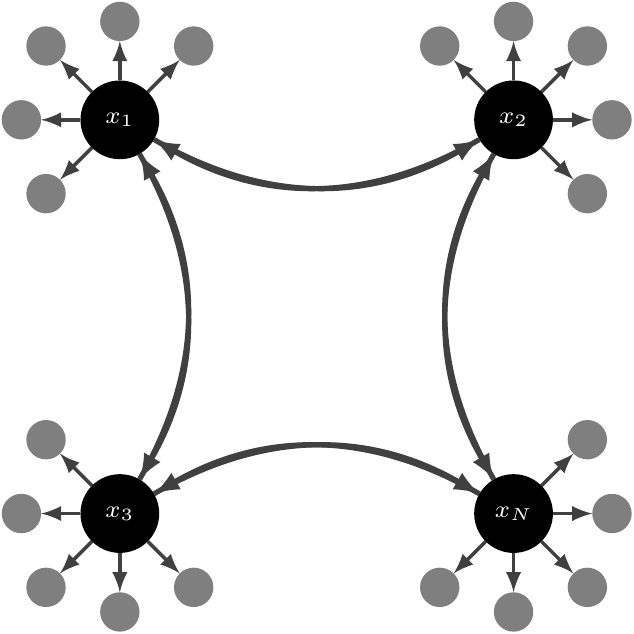}
    \subcaption{Weighted asymmetric network}
  \end{subfigure}
  \caption{\textbf{Emergence of Network Risk:} The sub-figure \small{(a) presents a network containing $x_1,\ldots,x_N$ stocks represented as nodes. Here node $x_1$ influences all other $x_2,\ldots,x_N$ nodes. Nodes $x_2,\ldots,x_N$ are not connected to each other and also do not influence node $x_1$. The sub-figure (b) presents an $N$-star network where nodes $x_1,\ldots,x_N$ are connected to a set of nodes exclusively, and also to one another. Arrows denote the direction of the connection and the density of the line denotes the strength of the connections.}}
  \label{Asymmetric_networks}
\end{figure}

To illustrate the above, Figure \ref{Asymmetric_networks} (a) shows a star network topology where stock 1 is central and its shocks propagate to the remaining N assets on the market. In this case, even if an investor holds a large number of stocks, they still face exposure to stock 1. This is a so called directed, or asymmetric network in which the nodes are connected with the same strength. Figure \ref{Asymmetric_networks} (b) is an N star network topology where the links among the N central firms have weighted strengths. In this case, an investor faces exposure to shock propagation and reception of the N central firms resulting in a complex risk structure from the network. 

While the centrality of these nodes is a characteristic for the determination of network risk premia, these networks are missing two key ingredients. First, these networks are static; and therefore do not capture evolving relationships among firms. Second, while distinguishing between short-term and long-term risks is a prominent issue in finance \citep{bansal2004risks}\footnote{Note that the assumption that sources of systematic risk are constant across horizons faces rare questioning within the literature; although studies are emerging. \cite{dew2016asset} formalize the notion of horizon specific risk pricing focusing on investor preferences. This allows the authors to quantify the meaning of long-run in the context of recursive preferences which they go on to show is significantly priced in equity markets. \cite{bandi2017business} take a different approach and extract the business cycle component of consumption to obtain factor loadings and show that they have similar pricing ability to those stemming from consumption growth aggregated over a two-year horizon. \cite{bandi2018measuring} argue that frequency is a source of systematic risk by decomposing betas into horizon specific components. They show a simple horizon specific CAPM may outperform popular multi-factor models whilst possessing economic interpretability.}, these networks convey no information on the persistence of such exposures. A network topology hence needs to reflect dynamic linkages due to shocks that create risks at different investment horizons.

Figure \ref{multilayer_networks} introduces a dynamic network with multiple layers that represent distinct connections among stocks. We interpret these links as dynamic horizon specific network connections where the layers distinguish short-term and long-term connections. Without loss of generality, the curves connecting the N central assets allow for the possibility of connections across layers. We represent time dynamics of network connections at periods $t=1,\ldots,T$ on a time line. For simplicity we assume that the N central assets are the same across each layer but will have connections of different strengths over time and investment horizon. Distinguishing between short-term and long-term layers that evolve dynamically allows one to relax the assumption that exposures to network risk are constant over horizons and time. 

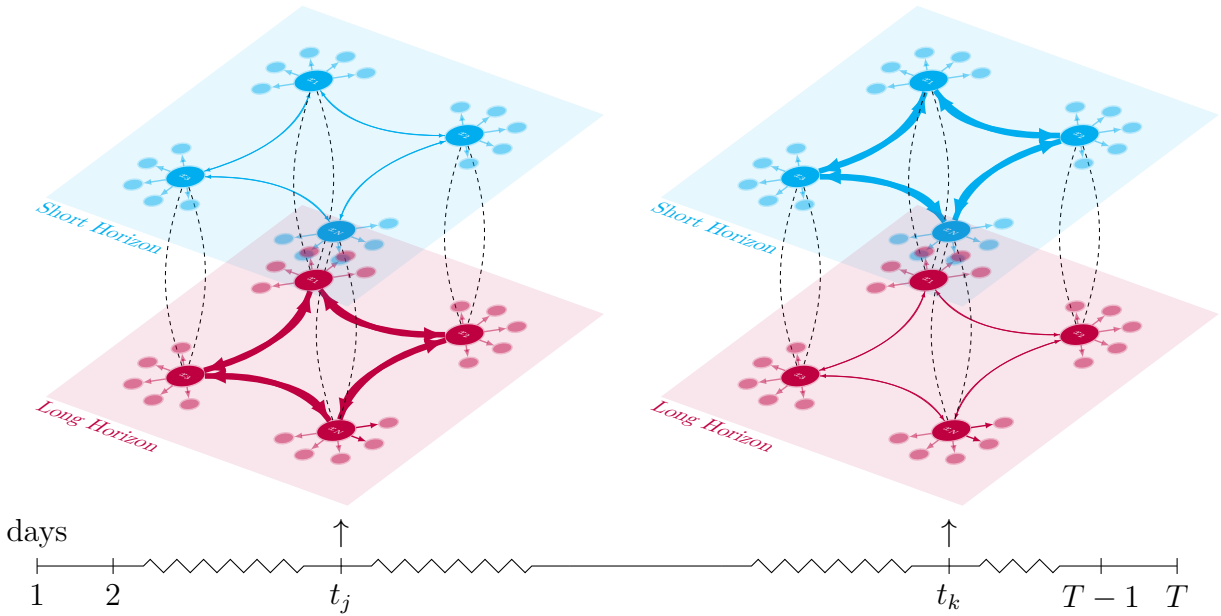
\begin{figure}[!hp]
  \begin{subfigure}{7.5cm}
    \centering\includegraphics[width=8cm]{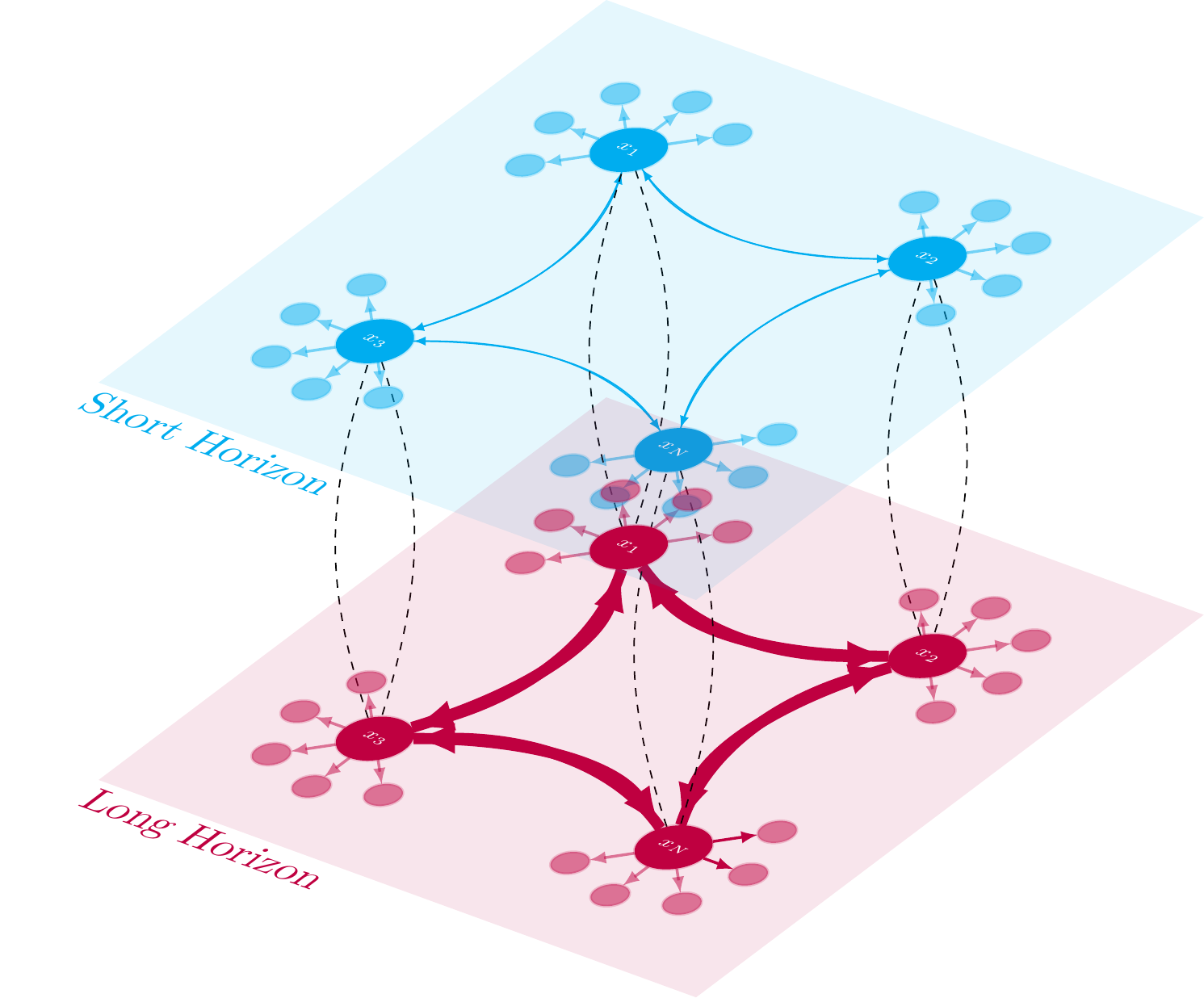}
  \end{subfigure}  
  \begin{subfigure}{9cm}
    \centering\includegraphics[width=8cm]{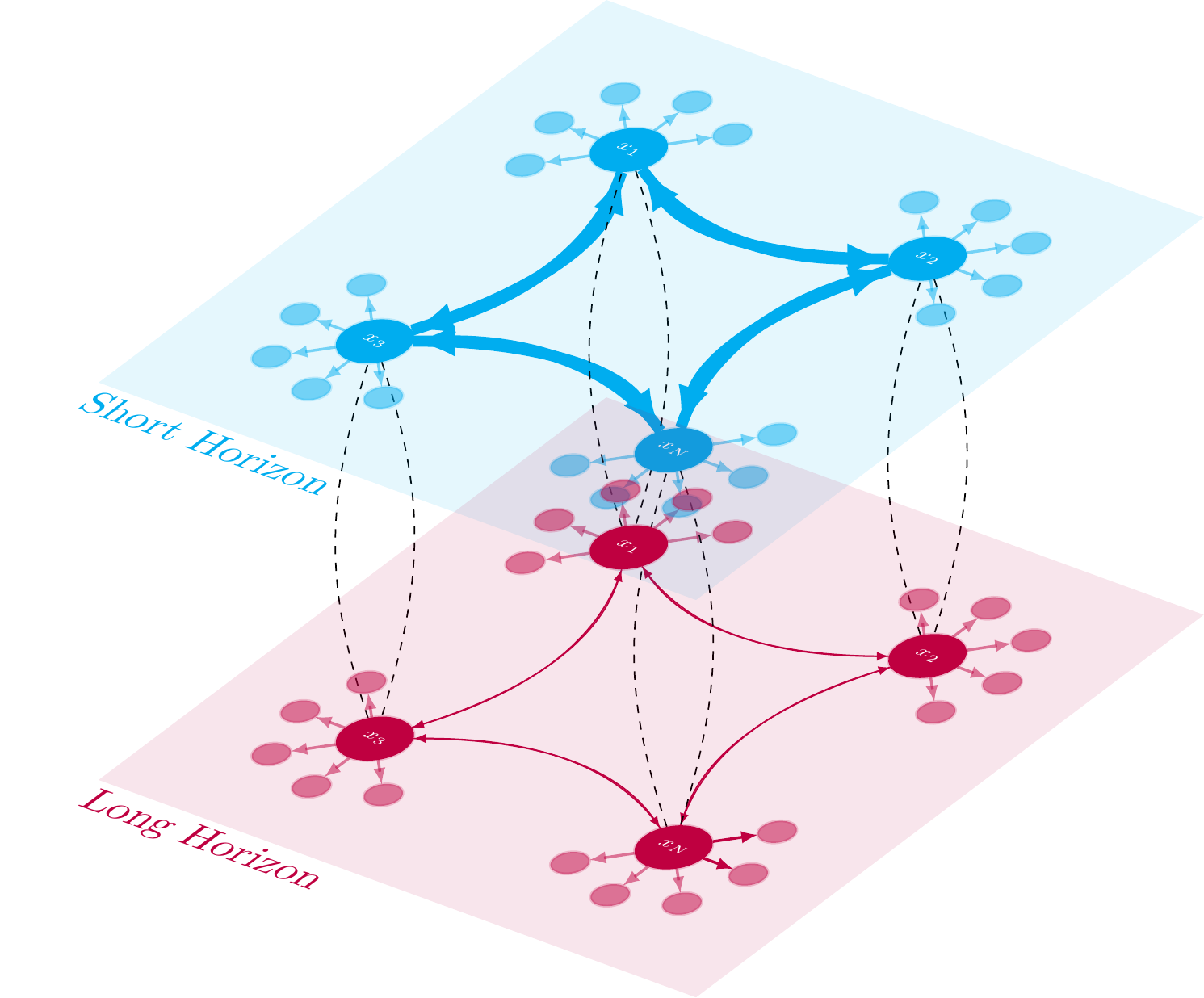}
  \end{subfigure}
  \begin{tikzpicture}[snake=zigzag, line before snake = 4mm, line after snake = 4mm]
    % draw horizontal line   
    \draw (0,0) -- (1,0);
    \draw[snake] (1,0) -- (4,0);
    \draw (4,0) -- (4,0);
    \draw[snake] (4,0) -- (7,0);
    \draw (7,0) -- (9,0);
    \draw[snake] (9,0) -- (12,0);
    \draw (12,0) -- (12,0);
    \draw[snake] (12,0) -- (14,0);
    \draw (14,0) -- (15,0);
    % draw vertical lines
    \foreach \x in {0,1,4,12,14,15}
      \draw (\x cm,3pt) -- (\x cm,-3pt);
    % draw nodes
    \draw (0,0) node[below=3pt] {$ 1 $} node[above=3pt] {days};
    \draw (1,0) node[below=3pt] {$ 2 $} ;
    \draw (4,0) node[below=3pt] {$ t_j $} node[above=3pt] {$\uparrow$};
    \draw (12,0) node[below=3pt] {$ t_k $} node[above=3pt] {$\uparrow$};
    \draw (14,0) node[below=3pt] {$ T-1 $} ;
    \draw (15,0) node[below=3pt] {$ T $} ;
  \end{tikzpicture}
  \caption{\textbf{A Dynamic Horizon Specific Network:} \small{This figure presents a multi-layer N-star network with snapshot of two time and two horizon specific network layers. Arrows denote directions of connections and the line density denotes strength. The curves from the $N$ central nodes allow for connections to spillover across layers. We interpret each layer as a network specific to a horizon of interest; for example short-term depicted by light blue color and long-term depicted by light red color.}}\label{multilayer_networks}
\end{figure}

Therefore, Figure \ref{multilayer_networks} describes an investment opportunity with dynamic horizon specific directional network risk exposures. Investors should require larger compensation, in absolute value, for risk in the long-run when investing during period $t_j$ in comparison to short-run investments, and vice versa when investing during period $t_k$. Understanding the dynamics of node centrality and shock propagation in such network is crucial for an investor because these dynamic connections create a source of systematic risk from network structures.

Our main objective is to investigate the pricing implications of dynamic horizon specific directional network risk for the cross-section of stock returns. In doing so, we decompose dynamic network connections among the daily realized volatilities of S\&P500 constituent stock returns into horizon specific components; namely a short-term and long-term component. We track dynamic horizon specific network connections among all S\&P500 constituents using measures that stem from a large scale time-varying parameter vector autoregressive approximating model. Our approach to tracking network connections permits one to examine risk at any horizon of interest and consequently hedge against it. We show that dynamic horizon specific, and aggregate, directional network risk constitute sources of risk that price stocks in the cross-section. To the best of our knowledge, we are the first to propose and explore the asset pricing implications of truly time-varying directional network risk; whilst also disentangling between short-term and long-term exposures. An important aspect of our work is the ability to characterize risks in large scale networks composed from hundreds of stocks.

To proxy dynamic horizon specific network risk, we construct tradable factors using directional connections among asset return volatilities. \cite{branger2018equilibrium} provide theoretical motivations for directional network connections in prices and the channels through which they influence equilibrium returns. We also control for the relative importance of nodes to the network, which is similar to concentration defined in \cite{herskovic2018networks}; a characteristic found to be important for explaining the cross-section of stock returns. Inherently our proxies for dynamic horizon specific network risk incorporate dynamic network-wide properties.

Our main result is that stocks with high sensitivities to dynamic network risk earn lower returns. These sources of risk are statistically significant and economically meaningful. Fama-MacBeth regressions indicate a one-standard-deviation increase across stocks in short-term and long-term directional network risk loadings implies a 6.76\% and 7.66\% drop in expected annual returns respectively. We also conduct portfolio sorts that assign stocks into quintile portfolios according to their short-term (long-term) directional network risk betas. The annual \cite{fama2015five} five-factor alphas of value weighted hedge portfolios of short-term and long-term directional network risk are -3.43\% and -4.05\% respectively. 

Our analysis is robust to using portfolio sorts and Fama-MacBeth regressions that include a battery of additional factors over and above the \cite{fama2015five} five-factor model. Namely, we include: changes in the VIX; momentum; conditional skewness; conditional kurtosis; and market illiquidity. We also control for the variance risk premium, tail risk, idiosyncratic volatility and idiosyncratic skewness. Even with these additional factors, dynamic horizon specific network risk maintains statistical and economic significance; whilst also showing it is a distinct source of systematic risk that is not attributable to volatility risk. We also provide an alternative specification for directional network risk and obtain statistically and economically significant results consistent with our main results. Moreover, we show that one is able to predict future directional network risk, both horizon specific and aggregate, by conducting portfolio sorts on past directional network risk betas. We obtain statistically significant raw and risk-adjusted returns across all specifications. This means that investors are able to implement strategies that hedge against these sources of risk in real time. 

%Finally, we propose a mechanism describing this type of investor's behavior theoretically. For this we borrow from the literature that allows the stochastic discount factor to load on volatility risk (e.g. \cite{ang2006downside,campbell2018intertemporal}). The economic mechanism in such models suggests that investors seeking to hedge against changes in investment opportunities will find assets that have positive covariance with market volatility attractive, and thus require lower expected returns. In tracking horizon specific network connections among return volatilities, we outline a model where the stochastic discount factor loads on mutually and self exciting jumps in horizon specific volatility processes. We combine the above rationale with the outline of an endowment economy tracking horizon specific connections in asset return volatilities to conjecture negative risk prices.

The remainder of this paper proceeds as follows. Section \ref{lit} outlines an economy where the stochastic discount factor loads on dynamic horizon specific directional network connections among asset return volatilities and discusses related literature. In Section \ref{empirical_measures}, we describe our methodology in tracking dynamic horizon specific network connections. In Section \ref{data_fact_constr}, we discuss data and how we track dynamic horizon specific network risk of all S\&P500 constituents. Our empirical results and extensions are in Sections \ref{main_results} and \ref{extensions}. Finally, conclusion are given in Section \ref{conclusion}.

\section{Theoretical Framework}\label{lit}

%The use of networks in finance and macroeconomics mainly concentrates on providing a description of financial contagion (e.g. \cite{diebold2014}, \cite{elliott2014financial}, and \cite{glasserman2016contagion}); documenting stylized facts (e.g. \cite{carvalho2013great}, \cite{billio2012econometric}); and building microfoundations for business cycles (e.g. \cite{acemoglu2012network}, \cite{acemoglu2017microeconomic}, \cite{carvalho2019large}).

Networks in asset pricing is an emerging literature. \cite{buraschi2012dynamic} establish a link between network structure and the cross-section of stock returns using dividends. Meanwhile, \cite{ahern2013network} focuses on a production-based framework using input-output data to show that industries in more central positions within the network earn higher returns. \cite{herskovic2018networks} builds on this and deduces an equilibrium asset pricing model where sectors connect through an input-output network. Defining concentration as the degree to which the network is dominated by a few sectors, and sparsity as the distribution of sectoral connections, the author shows that factors stemming from these network-wide properties price stocks in the cross-section. \cite{branger2018equilibrium} characterize a model where equilibrium expected returns depend on directional network connections of negative price jumps. Their model implies that the overall effect of network connections depends on whether shock propagation and reception dominates a hedging channel\footnote{If the hedging channel dominates the risk premium is negative. This is because less connected assets during periods of financial turbulence are relatively more favorable to connected assets within the network.}. \cite{herskovic2018firm} assess firm volatility in a network model where shocks to customers influence their suppliers. Their model allows for asymmetries between strength of customer-supplier linkages and produces distributions of firm volatility, size and customer concentration consistent with the data.

Our work also relates well with studies regarding volatility risk and time-varying volatility (e.g. \cite{ang2006cross}, and \cite{ang2006downside}). These studies show that investors seek to use assets with positive covariance with market volatility as hedging devices and will thus accept lower returns. \cite{campbell2018intertemporal} allow for stochastic volatility in an ICAPM framework. They show that asset returns that positively covary with a variable forecasting future market volatility have low expected returns in equilibrium. The economic mechanism is that investors reduce current consumption to increase precautionary savings in light of uncertainty around market returns. \cite{cremers2015aggregate} distinguish between jump and volatility risk showing that they bear different risk premia; both of which are negative. \cite{herskovic2016common} extract a common factor from firm level volatility and show that the highest quintile portfolios sorted on the common idiosyncratic volatility factor earn 5.4\% lower returns per year than the lowest quintile portfolios.

We contribute to this literature by focusing on the directional network connections among asset return volatilities. We combine theoretical justifications of negative risk prices for aggregate volatility with an economy where the stochastic discount loads on horizon specific directional network connections to motivate negative risk prices.

\subsection{An economy with horizon specific volatility connections}\label{economy}

Our two-tree endowment economy generates horizon specific connections in the volatility of asset returns and the volatility of consumption growth. Specifically, we extend on \cite{cochrane2007two} and \cite{lucas1978asset}. Similar to \cite{bansal2004risks} and \cite{backus2011disasters}, we model asset returns as claims on risk factors in the consumption process. The representative investor has the following general utility over the stream of consumption
\begin{align}
U_t = \mathbb{E}_t \int_{0}^{\infty} e^{- \delta \tau} u(c_{t+\tau}) d \tau . 
\end{align}
Each endowment dividend stream follows a geometric Brownian motion with stochastic volatility; whose respective drift and diffusion parameters differ.
\begin{align}
\label{eq:consumption01}\frac{dD_\iota}{D_\iota} &= \mu_\iota dt + \sqrt{v_{\iota,t}} d Z_\iota, \quad \iota =\{S,L\}\\
\label{eq:consumption02} dv_{\iota,t}       &= \kappa_{v_{\iota}}\left(\bar{v}_{\iota} - v_{\iota,t} \right)dt + \sigma_{v_{\iota}}\sqrt{v_{\iota,t}}d Z_{v_{\iota}} + \sum^{N}_{j=1} K_{\iota,j}d\mathcal{N}_{\iota,j,t}
\end{align}

where $dZ_\iota$ are standard Brownian motions that are possibly correlated, $\mathbb{C}\text{orr}\Big(dZ_{S},dZ_{L}\Big)=\rho_{S,L}dt$. We interpret the endowment trees as short-term ($S$) and long-term ($L$) risk factors within the economy\footnote{One may conjecture that the $L$ tree corresponds to the long-term component of consumption. This generates a persistent dividend stream and bears long-term risk in the economy. The $S$ tree generates a less persistent dividend stream bearing short-term risks in the economy; $\mu_L > \mu_{S} > 0$.}. The diffusion of each tree follows a mean-reverting square root process \citep{heston1993closed} with $N$ self and mutually exciting jumps $\mathcal{N}_{\iota,j,t}$, $\iota=\{S,L\}$ \citep{ait2015modeling}. $\bar{v}_{\iota}$ is the unconditional horizon specific component of conditional variance and $\kappa_{v_{\iota}}$ captures the speed of mean reversion\footnote{Again one may conjecture the speed of mean reversion is slower for the volatility process associated to the long-term part of consumption.}.  

The $N$ self and mutually exciting jumps introduce horizon specific network connections in the volatility of consumption growth. The jump intensities are stochastic and follow Hawkes processes with mean reverting dynamics of the form:
\begin{eqnarray}\label{eq:intensities_dynamics}
d\ell_{\iota,j,t} = \alpha_{\iota,j}\left(\ell_{\iota,j,\infty}-\ell_{\iota,j,t} \right)dt + \sum^{N}_{k=1}b_{\iota,j,k}d\mathcal{N}_{\iota,k,t}
\end{eqnarray}
This means that a horizon specific jumps in asset $k$ causes an increase in the horizon specific jump intensity of asset $j$ such that $\ell_{\iota,j}$ jumps by $b_{\iota,j,k}$ before decaying back towards the level $\ell_{\iota,j,\infty}$ at speed $\alpha_{\iota,j}$. If the increase in $\ell_{\iota,j}$ leads to a jump in asset $j$, and there is a non-zero $b_{\iota,n,j}$, the horizon specific shock passes on to asset $n$. In this manner the shocks can be propagated throughout the entire network, which also permits the initial shock to reach asset $k$ itself.

We have $N$ risky assets in the economy that whose dynamics are geometric Brownian motions with stochastic volatility. Formally the $k$-th asset has the following dynamics:
\begin{align}
\frac{dp_{k,t}}{p_{k,t}} &= \mu_{p_{k}}dt + \sqrt{v_{p_{k,S}}}dW_{S} + \sqrt{v_{p_{k,L}}}dW_{L}\\
d v_{p_{k},\iota} &= \kappa_{p_{k},\iota}\left(\bar{v}_{p_{k},\iota} - v_{p_{k},\iota}\right)dt + \sigma_{p_{k},\iota}\sqrt{v_{p_{k},\iota}}dW_{\xi} + Q_{\iota}d\mathcal{N}_{\iota,k,t}, \quad \iota=\{S,L\},\: \xi=\{S_1,L_1\}
\end{align}
which is similar to \cite{christoffersen2009shape} where the variance of the stock return is the sum of the two stochastic volatility components. Note that $\mathbb{C}\text{orr}\Big(dW_{S},dW_{S_1}\Big)=\rho_{W_{S},W_{S_1}}dt$ and $\mathbb{C}\text{orr}\Big(dW_{L},dW_{L_1}\Big)=\rho_{W_{L},W_{L_1}}dt, \: \mathbb{C}\text{orr}\Big(dW_{S},dW_{L_1}\Big)=\mathbb{C}\text{orr}\Big(dW_{L},dW_{S_1}\Big)=0 $. We add discontinuities to each variance process that also enter horizon specific stochastic volatility processes for each consumption dividend stream. $Q_{\iota}$, are the jump sizes of compound Poisson processes $Q_{L}>Q_{S}>0$ and $\mathcal{N}_{S,k,t},\:\mathcal{N}_{L,k,t}$ are mutually independent Poisson processes. Note their intensity parameters are as in (\ref{eq:intensities_dynamics}).

The economy also contains a risk-less bond that follows an ordinary differential equation $\frac{dB}{B} = r_{f} dt$ with $r_{f}$ being the instantaneous risk-free rate. Applying It\^ o’s lemma to consumption $c_t=s_{t}D_{S}+(1-s_{t})D_{L}$ we obtain consumption dynamics
\begin{align}\label{eq:consumption_dynamics0}
\frac{dc_t}{c_t} &= [s_{t}\mu_S  + (1 - s_{t})\mu_L ] dt + s_{t}\sqrt{v_{S,t}} dZ_S + (1 - s_{t})\sqrt{v_{L,t}}dZ_L.
\end{align}

Using the results in Appendix \ref{app:economy} providing first two moments of consumption growth and further details of the economy, the following proposition formalizes how network connections influence stochastic discount factor.

\begin{prop}[Stochastic Discount Factor Innovations in a Network Economy]
  \label{prop:2}
Consider an endowment economy with a short run and long run cash flows for consumption as in Equations \ref{eq:consumption01} -- \ref{eq:consumption02} and consumption dynamics follow Equation \ref{eq:consumption_dynamics0}. Then, the time $t$ expected innovation in the stochastic discount factor, $\mathbb{E}_{t}\left[\frac{d \Lambda_{t}}{\Lambda_{t}} \right]$ are
  \begin{multline}\label{eq:SDF}
\mathbb{E}_{t}\left[\frac{d \Lambda_{t}}{\Lambda_{t}} \right] = -\delta dt - \gamma_{t} \Big[s_{t}\mu_{S} + (1-s_{t})\mu_{L}\Big]dt \\ + \frac{1}{2} \eta_{t} \Big[s_{t}^{2}v_{S,t} + (1-s_{t})^{2}v_{L,t}+s_{t}(1-s_{t})\sqrt{v_{S,t}}\sqrt{v_{L,t}}\rho_{S,L}\Big]dt
\end{multline} 
with $\delta>0$ being impatience, $\gamma_{t}\equiv -\frac{u''(c_{t})c_{t}}{u'(c_{t})} >0$ is the coefficient of risk aversion, and $\eta_{t}\equiv \frac{u'''(c_{t})(c_{t})^{2}}{u'(c_{t})} >0$ is precautionary saving.  
\end{prop}
\begin{proof}
  See Appendix \ref{app:proofs1}.
\end{proof}

We can see that the stochastic discount factor depends on horizon specific network connections through the respective stochastic volatility processes. Note that these expressions hold in general for any utility function\footnote{Note also that it is not the purpose of this work to provide a specific, either closed-form or numerical, solution to this model. However, one may do so by specifying, for example, recursive preferences and adopt the methods in \cite{eraker2008equilibrium} to obtain a numerical solution. Setting our economy up in a similar framework to the example in \cite{eraker2008equilibrium} would yield risk premia for: horizon specific components of the volatility of consumption growth, and dynamic horizon specific directional network risk. We leave the solution and calibration of this model for future research.}. We use the above to outline an economy that permits network connections to enter the expression for the stochastic discount factor. In turn, this implies that exposure to horizon specific network risk will be priced in equilibrium. 

As we outline earlier, the literature concerning volatility risk shows that stocks loading positively on volatility act as inter-temporal hedging devices against future uncertainty. In this economy, network connections form among stock return volatilities which means that investors will demand compensation for dynamic horizon specific network risk arising through volatility connections of individual assets; over and above the premia for being exposed to horizon specific consumption growth volatility risk. 

It is clear from (\ref{eq:SDF}) that the stochastic discount factor loads on horizon specific network connections through precautionary savings. Therefore, stocks loading positively on horizon specific network risk earn lower returns because they act as inter-temporal hedging devices against adverse changes in investment opportunities. Importantly in this economy, we show that this is a source of risk that is different to volatility risk. Empirically, Sections \ref{main_results} and \ref{extensions} confirm that dynamic horizon specific network risk constitutes a source of risk different to measures of volatility risk.

\section{Measurement of Dynamic Networks}\label{empirical_measures}

Our objective here is to understand how shocks with different persistence propagate across a network of assets. Knowing how a shock to stock $j$ transmits to stock $k$ will define weighted and directed network at a given period of time and at a given horizon. In turn, one may use this information to build aggregate system-wide measures of network connectedness as well as disaggregate measures stemming from individual connections that will characterize various types of risks. In contrast to the network literature in Finance \citep{elliott2014financial,glasserman2016contagion,herskovic2018networks}, we focus on network dynamics over time and across horizons. In doing so, we estimate a large dimensional system and our methods are readily available for studying big data.

Algebraically, the adjacency matrix captures all of the information about the network, and any sensible measure must be related to it. A typical metric used by the wide network literature that provides the user with information about the relative importance or influence of nodes is network centrality. For our purposes, we want to measure node degrees that capture the number of links to other nodes. The distribution shape of the node degrees is a network-wide property that closely relates to network behavior. As for the connectedness of the network, the location of the degree distribution is key, and hence, the mean of the degree distribution emerges as a benchmark measure of overall network connectedness.

A network working with causal linkages can be characterized well through variance decompositions from a vector auto-regression approximation model \citep{diebold2014}. Variance decompositions provide useful information about how much of the future variance of variable $j$ is due to shocks in variable $k$. Thus, the variance decomposition matrix defines the network adjacency matrix and is intimately related to network node degrees, mean degrees, and connectedness measures \citep{diebold2014}. Currently studies examine, almost exclusively, static networks mimicking time dynamics with estimation from an approximating window\footnote{\cite{geraci2018measuring} estimate multiple pairwise time varying parameter models in an attempt to characterize a network of financial stocks using autoregressive coefficients.}. We employ a locally stationary TVP VAR that allows us to estimate the adjacency matrix for a network, or market, of stocks at each point in time. We  decompose this into horizon specific components that allow us to disentangle short-term and long-term network connections\footnote{Dimensionality is a problem using large scale TVP VARs. \cite{chan2020reducing} propose methods to estimate TVP VARs with 15 variables, and  \cite{kapetanios2019large} and \cite{petrova2019quasi} define large TVP VARs as 78 and 87 variables respectively.}.

Dynamic horizon specific networks that we define by time-varying variance decompositions are more sophisticated than classical network structures. In a typical network, the adjacency matrix contains zero and one entries, depending on the node being linked or not, respectively. In the above notion, one interprets variance decompositions as weighted links showing the strength of the connections. In addition, the links are directed, meaning that the $j$ to $k$ link is not necessarily the same as the $k$ to $j$ link, and hence, the adjacency matrix is not symmetric. Therefore we can define weighted, directed versions of network connectedness statistics readily that include degrees, degree distributions, distances and diameters. 

These measures are key to our analysis since risk stems directly from asymmetries within the network. Our adjacency matrix is dynamic and our measures are dynamic horizon-specific in-degrees and out-degrees. We also define the mean degree of the network which characterizes total connectedness of the network at a given time and horizon. This is a useful summary statistic for the network that studies use to analyse stylized facts of financial data \citep{billio2012econometric,carvalho2013great}.

To construct dynamic network measures, we interpret the TVP VAR approximating model as a dynamic network that provides information about how much of future uncertainty of variable $j$ is due to shocks in variable $k$. A natural way to describe horizon specific dynamics (i.e. short-term, and long-term) of the network connections is to consider the spectral representation of the approximating model. \cite{stiassny1996spectral} introduces the notion of a spectral representation in a relatively restrictive setting, while \cite{barunik2018measuring} define horizon specific connectedness measures for a simple VAR that we further generalize to a locally stationary processes.

Consider a doubly indexed $N$-variate time series $(\bX_{t,T})_{1\le t \le T,T \in \IN}$ with components $\bX_{t,T}=(\bX_{t,T}^1,\ldots,\bX_{t,T}^N)^{\top}$ that describe all assets in an economy. Here $t$ refers to a discrete time index and $T$ is an additional index indicating the sharpness of the local approximation of the time series $(\bX_{t,T})_{1\le t \le T,T \in \IN}$ by a stationary one. We rescale time using the conditions in \cite{dahlhaus1996kullback} such that the continuous parameter $u \approx t/T$ is a local approximation of the weakly stationary time-series. 

We assume assets to follow a locally stationary TVP-VAR of lag order $p$ as
\begin{equation}\label{eq:VAR1}
\bX_{t,T}=\bPhi_{1}(t/T)\bX_{t-1,T}+\ldots+\bPhi_{p}(t/T)\bX_{t-p,T} + \bepsilon_{t,T},
\end{equation}
where $\bepsilon_{t,T}=\bSigma^{-1/2}(t/T)\bbeta_{t,T}$ with $\bbeta_{t,T}\sim NID(0,\boldsymbol{I}_M)$ and $\bPhi(t/T)=(\bPhi_{1}(t/T),\ldots,\bPhi_{p}(t/T))^{\top}$ are the time varying autoregressive coefficients. In a neighborhood of a fixed time point $u_0=t_0/T$, we approximate the process $\bX_{t,T}$ by a stationary process $\widetilde{\bX}_t(u_0)$ as
\begin{equation}\label{eq:VAR2}
\widetilde{\bX}_t(u_0)=\bPhi_1(u_0)\widetilde{\bX}_{t-1}(u_0)+\ldots+\bPhi_p(u_0)\widetilde{\bX}_{t-p}(u_0) + \bepsilon_t,
\end{equation}
with $t\in \IZ$ and under suitable regularity conditions which justifies the notation ``locally stationary process''. Importantly, the process has time varying Vector Moving Average VMA($\infty$) representation \citep{dahlhaus2009empirical,roueff2016prediction}
\begin{equation}
\bX_{t,T} = \sum_{h=-\infty}^{\infty} \bPsi_{t,T}(h)\bepsilon_{t-h}
\end{equation}
where parameter vector $\bPsi_{t,T}(h) \approx\bPsi(t/T,h)$ is a bounded stochastic process\footnote{Since $\bPsi_{t,T}(h)$ contains an infinite number of lags, we approximate the the moving average coefficients at $h=1,\ldots,H$ horizons.}. The connectedness measures rely on variance decompositions, which are transformations of the information in $\bPsi_{t,T}(h)$ that permit the measurement of the contribution of shocks to the system. Since a shock to a variable in the model does not necessarily appear alone, an identification scheme is crucial in calculating variance decompositions. We adapt the generalized identification scheme in \cite{pesaran1998generalized} to locally stationary processes. 

We transform the local impulse responses in the system to local impulse transfer functions using Fourier transforms $\bPsi(u)e^{-i\omega} = \sum_h e^{-i\omega h} \bPsi(u,h)$\footnote{Note that $i=\sqrt{-1}$.}. The following proposition establishes a dynamic representation of the variance decomposition of shocks from asset $j$ to asset $k$. It is central to the development of the dynamic network measures since it constitutes a dynamic horizon specific adjacency matrix.

\begin{prop}[Dynamic Horizon Specific Networks]\footnote{Note to notation: $[\boldsymbol A]_{j,k}$ denotes the $j$th row and $k$th column of matrix $\boldsymbol A$ denoted in bold. $[\boldsymbol A]_{j,\cdot}$ denotes the full $j$th row; this is similar for the columns. A $\sum A$, where $A$ is a matrix that denotes the sum of all elements of the matrix $A$.} \label{prop:3}
Suppose $\mathbf{\bX_{t,T}}$ is a weakly locally stationary process with $\sigma_{kk}^{-1} \displaystyle\sum_{h = 0}^{\infty}\left|\Big[\bPsi(u,h) \bSigma(u) \Big]_{j,k}\right| < +\infty, \forall j,k.$
	Then the $(j,k)$th element of the dynamic horizon specific adjacency matrix \btheta(u,d) at a rescaled time $u=t_0/T$ and horizon $d = (a,b): a,b \in (-\pi, \pi), a < b$ is defined as
	\begin{equation} \label{eq:dynamicadjmatrix}
	\Big[ \btheta(u,d) \Big]_{j,k} = \frac{\sigma_{kk}^{-1} \displaystyle \int_{a}^{b} \left| \bigg[ \bPsi(u) e^{-i \omega} \bSigma(u) \bigg]_{j,k} \right|^2 d \omega}{ \displaystyle \int_{-\pi}^{\pi} \Bigg[ \Big\{\bPsi(u) e^{-i \omega} \Big\}\bSigma(u) \Big\{ \bPsi(u) e^{+i \omega } \Big\}^{\top}  \Bigg]_{j,j} d \omega}\end{equation}
\end{prop} 
\begin{proof}
	See Appendix \ref{app:proofs2}.
\end{proof}

It is important to note that $\Big[ \btheta(u,d) \Big]_{j,k}$ is a natural dissagregation of traditional variance decompositions to a time-varying and horizon specific adjacency matrix. This is because the portion of the local error variance of the $j$th variable at horizon $d$ due to shocks in the $k$th variable is scaled by the total variance of the $j$th variable. Note that the quantity in proposition (\ref{prop:3}) is the squared modulus of weighted complex numbers, thus producing a real quantity.

This relationship is an identity which means the integral is a linear operator, summing over disjoint intervals covering the entire range $(-\pi, \pi)$ recovers the time domain counterpart of the local variance decomposition. The following remark formalizes.

\begin{remark}[Aggregation of Dynamic Network]
	\label{rem:recomposition}
	Denote by $d_s$ an interval on the real line from the set of intervals $D$ that form a partition of the interval $(-\pi, \pi)$, such that $\cap_{d_s \in D} d_s = \emptyset, $ and $\cup_{d_s \in D} d_s = (-\pi, \pi)$. Due to the linearity of integral and the construction of $d_s$, we have 
	$$
	\Big[ \btheta(u,\infty) \Big]_{j,k} = \sum_{d_s \in D} \Big[ \btheta(u,d_s) \Big]_{j,k}.$$
\end{remark}

Remark (\ref{rem:recomposition}) is important as it establishes the aggregation of horizon specific network connectedness measures to its time domain counterpart. Thus, short-term ($d=S$) and long-term ($d=L$) time-varying network characteristics always sum up to an aggregate time domain counterpart; this makes them directly interpretable. As the rows of the dynamic adjacency matrix do not necessarily sum to one, we normalize the element in each by the corresponding row sum $\Big[ \widetilde \btheta(u,d) \Big]_{j,k} = \Big[ \btheta(u,d) \Big]_{j,k}\Bigg/ \sum_{k=1}^N\Big[  \btheta(u,\infty) \Big]_{j,k}$. 

It is important to note that proposition \ref{prop:3} defines the dynamic horizon specific network completely. Naturally, our adjacency matrix is filled with weighted links showing strengths of the connections. The links are directional, meaning that the $j$ to $k$ link is not necessarily the same as the $k$ to $j$ link. Therefore the adjacency matrix is asymmetric which creates undiversifiable risk. Using our notion above, the adjacency matrix evolves dynamically in time and is horizon specific.

To characterize network risk, we define total dynamic network connectedness measures at a given horizon as the ratio of the off-diagonal elements to the sum of the entire matrix

\begin{equation}
\mC(u,d) = 100\times\displaystyle \sum_{\substack{j,k=1\\ j\ne k}}^N \Big[\widetilde \btheta(u,d)\Big]_{j,k}\Bigg/\displaystyle \sum_{j,k=1}^N \Big[\widetilde \btheta(u,\infty)\Big]_{j,k}
\end{equation}
This measures the contribution of forecast error variance attributable to all shocks in the system, minus the contribution of own shocks. Similar to the aggregate network connectedness measure that infers the system-wide risk, we define measures that will reveal when an individual asset is a transmitter or a receiver of shocks. We use these measures to proxy dynamic horizon specific network risk. The dynamic directional connectedness that measures how much of each asset's $j$ variance is due to shocks in other assets $j\ne k$ in the economy is given by
\begin{equation}
\mC_{j\leftarrow\bullet}(u,d) = 100\times\displaystyle \sum_{\substack{k=1\\ k\ne j}}^N \Big[\widetilde \btheta(u,d)\Big]_{j,k}\Bigg/\displaystyle \sum_{j,k=1}^N \Big[\widetilde \btheta(u,\infty)\Big]_{j,k},
\end{equation}
defining the so-called \textsc{from} connectedness. Note one can precisely interpret this quantity as dynamic from-degrees (or out-degrees in the network literature) that associates with the nodes of the weighted directed network we represent by the dynamic variance decomposition matrix. Likewise, the contribution of asset $j$ to variances in other variables is 
\begin{equation}
\mC_{j\rightarrow \bullet}(u,d) = 100\times\displaystyle \sum_{\substack{k=1\\ k\ne j}}^N \Big[\widetilde \btheta(u,d)\Big]_{k,j}\Bigg/\displaystyle \sum_{j,j=1}^N \Big[\widetilde \btheta(u,\infty)\Big]_{k,j}
\end{equation}
and is the so-called \textsc{to} connectedness. Again, one precisely interprets this as dynamic to-degrees (or in-degrees in the network literature) that associates with the nodes of the weighted directed network that we represent by the variance decompositions matrix. These two measures show how other assets contribute to the risk of asset $j$, and how asset $j$ contributes to the riskiness of others, respectively, in a time-varying fashion at horizon $d$. Notably one can simply add these measures across all horizons to obtain aggregate time-varying measures.

Finally, to obtain the time-varying coefficient estimates, and the time-varying covariance matrices at a fixed time point $u=t_{0}/T$, $\bPhi_{1}(u),...,\bPhi_{p}(u)$ $\bSigma(u)$, we estimate the approximating model in (\ref{eq:VAR2}) using Quasi-Bayesian Local-Likelihood (QBLL) methods \citep{petrova2019quasi}.  

Specifically, we use a kernel weighting function that provides larger weights to observations that surround the period whose coefficient and covariance matrices are of interest. Using conjugate priors, the (quasi) posterior distribution of the parameters of the model are available analytically. This alleviates the need to use a Markov Chain Monte Carlo (MCMC) simulation algorithm and permits the use of parallel computing. Note also that in using (quasi) Bayesian estimation  methods, we obtain a distribution of parameters that we use to construct network measures that provide confidence bands for inference. We detail the estimation algorithm in Appendix \ref{app:estimate}. % In addition, we provide a computationally efficient packages \texttt{DynamicNets.jl} and \texttt{DynamicNets} that allows one to obtain our measures on data the researcher desires in \textsf{JULIA} and \textsf{MATLAB} available from \url{https://github.com/barunik/DynamicNets.jl} and \url{https://github.com/mte00/DynamicNets}.\footnote{Unlike traditional TVP VARs time-variation evolves in a non-parametric manner thus making no assumption on the laws of motion within the model. Typically, the model of \cite{primiceri2005time}, and many extensions, assume parameters evolve as random walks or autoregressive processes.}

We provide some details on estimation here. First, the variance decompositions of the forecast errors from the VMA($\infty$) representation require a truncation of the infinite horizon with a $H$ horizon approximation. As $H\rightarrow \infty$ the error disappears \citep{lutkepohl2005new}. We note here that $H$ serves as an approximating factor and has no interpretation in the time-domain. We obtain horizon specific measures using Fourier transforms and set our truncation horizon $H$=100; results are qualitatively similar for $H\in \{50,100,200\}$. Second in computing our measures, we diagonalize the covariance matrix because our objective is to focus on the causal affects of network connections. The $\bPsi(u,h)$ matrix embeds the causal nature of network linkages, and the covariance matrix $\bSigma(u)$ contains contemporaneous covariances within the off-diagonal elements. In diagonalizing the covariance matrix we remove the contemporaneous effects and focus solely on causation.

\section{Network Dynamics of S\&P500 Constituents}\label{data_fact_constr}

%\footnote{Institutions connect directly through counterparty risk, contractual obligations or other general business relationships. High-frequency analysis of such networks requires a high-frequency balance sheet and other generally unavailable information. In contrast, market-based measures are available in high frequencies that reflect the decisions of many agents assessing risks and therefore contain information regarding network linkages. The pure market-based approach we use in contrast to other network techniques allows us to monitor  network risk on a daily basis at the cost of minimal assumptions.}

Our objective is to explore the pricing implications of dynamic horizon specific directional network risk for stock returns. The economy we outline in Section \ref{economy} proposes a stochastic discount factor that loads on network connections among stock return volatilities. Therefore we build a dynamic network among stock return volatilities. In doing so, we use high frequency data for all stocks listed on the S\&P500 from July 5, 2005 to August 31, 2018 from Tick Data. Specifically, we compute daily returns $\text{R}_{t} = \sum^{\text{D}}_{i=1}(p_{t,i}-p_{t,i-1})$, and realized volatility $\text{RV}_{t} = \sqrt{\sum^{\text{D}}_{i=1}(p_{t,i}-p_{t,i-1})^2}$ for each stock on day $t$, where $\text{D}$ denotes total number of intraday observations. The $i$ subscripts denote intraday observations which we observe at 5 minutes intervals; and $p_{t,i}$ is the intraday price of the asset\footnote{To obtain our network measures, we estimate the TVP VAR model as in (\ref{eq:VAR1}) on $N$=496 stocks with $p$=2 lags on our $T$=3278 days of data. We estimate our horizon specific dynamic network measures on a 48-core server. For every $t\in\{1,2,\dots, T\}$, we generate 500 simulations of the (quasi) posterior distribution which results in a total estimation time of 10 days.}. We define short-term as the 1-day to 1-week horizon and long-term as horizons greater than 1-week. 

Figure \ref{RETS} plots our horizon specific dynamic total connectedness measures from July 5, 2005 to August 31, 2018. Overall, there are substantial differences in the levels of horizon specific connectedness throughout our estimation sample. In general, long-term connections are muted during periods of economic/financial tranquillity. However, it is clear that long term connections surge during periods of economic recession or key stock market events. For example, long-term total connectedness begins to rise in 2006 and continues to do throughout the 2007-2009 recession. Adding to this, we can see long-term connectedness rising during 2010-2012. This may be attributable to fears of contagion of the European sovereign debt crisis, the 2010 flash crash, and when the S\&P500 entered a bear market in 2011; albeit short-lived. We can also see during mid to late 2015 short-term and long-term connectedness rising which is consistent with the stock market sell-off starting in August 2015; this can also be linked to fears of contagion of the Chinese stock market crash in late 2015.

\begin{figure}[!hp]
\centering
\scalebox{0.9}{\includegraphics{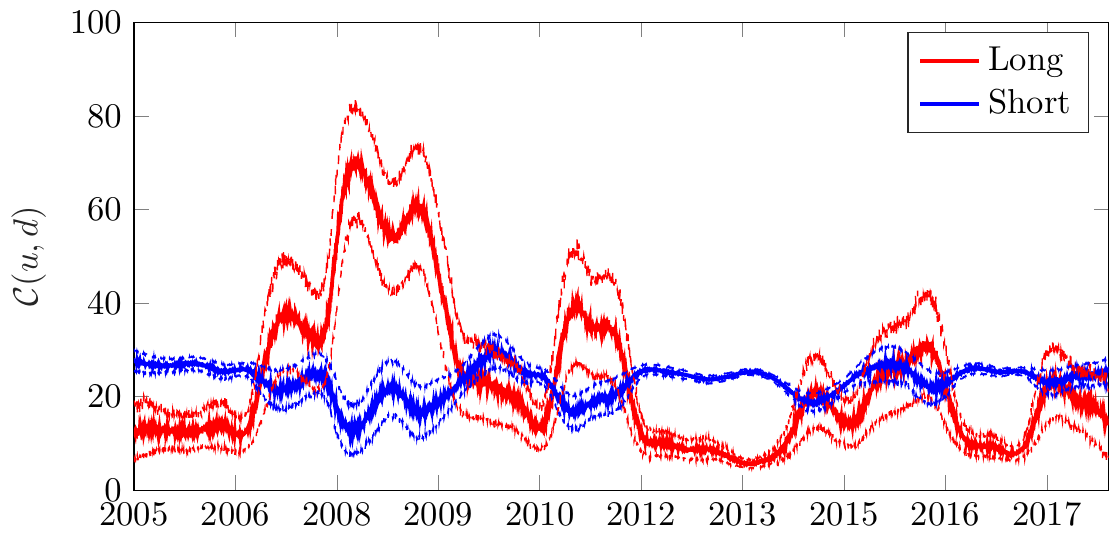}}
		\caption{\textbf{Horizon specific dynamic total network connectedness for S\&P500 constituents} \\ \small{This figure plots the quasi posterior median and 1-standard deviation percentiles of horizon specific dynamic total network connectedness, $\mC(u,d),\: d \in \{\text{S},\text{L}\} $ from July 5, 2005 to August 31, 2018. S refers to the short-term which we define as 1 day to 1 week; and L refers to long-term which we define as horizons $>$ 1 week. The spectrum with which the horizons stem from are linked to the frequency with which the data is observed.}}
      \label{RETS}
\end{figure}

For illustrative purposes, Figure \ref{NDC_EG}, reports net-directional connectedness of: Apple, Netflix, and Google. Net-directional connectedness is the difference between \textsc{to} connectedness and \textsc{from} connectedness $\mathcal{C}(u,d)_{j \rightarrow \bullet}-\mathcal{C}(u,d)_{j\leftarrow \bullet} $. Net-directional connectedness tells us, on average, how a stock on the market contributes to the network. When $\mathcal{C}(u,d)_{j \rightarrow \bullet}-\mathcal{C}(u,d)_{j\leftarrow \bullet}>0\: (<0)$ this tells us that the stock transmits (receives) shocks to (from) the network at horizon band $d$ and time period $u$ after accounting for how receptive (transmissive) the stock is to the network. Overall, we see differences among horizon specific net-directional connectedness measures for these stocks. Apple, is a long-term transmitter during the build up to the 2008 recession before becoming a receiver of shocks during 2008-2009. Netflix is a short-term transmitter to the network during the bear market of 2011 and fears of contagion of the European sovereign debt crisis. Finally, Google is a strong short-term receiver at the beginning of our sample before becoming a long-term receiver during the 2008 recession.

Overall, it is clear that our measures provide useful descriptions on the evolution of connections among financial assets over horizons and time. We can see that total connections, particularly over the long-term, rise substantially during key events. Intuitively, this links well with the observation that financial asset return volatilities exhibit clustering \citep{harris1991stock}. Our measures show that connections among assets intensify during these periods and decompose this into horizon specific bands\footnote{This methodology also permits one to look at disaggregated measures of directional network connections. These are also shown to vary across assets, horizons and time. Notably one could go further and examine pairwise connections between each asset on the market.}. This links well with \cite{diebold2014} and the subsequent literature emerging that uses their measures to track systemic risk.  

\begin{figure}[!hp]
\centering
\scalebox{0.9}{\includegraphics{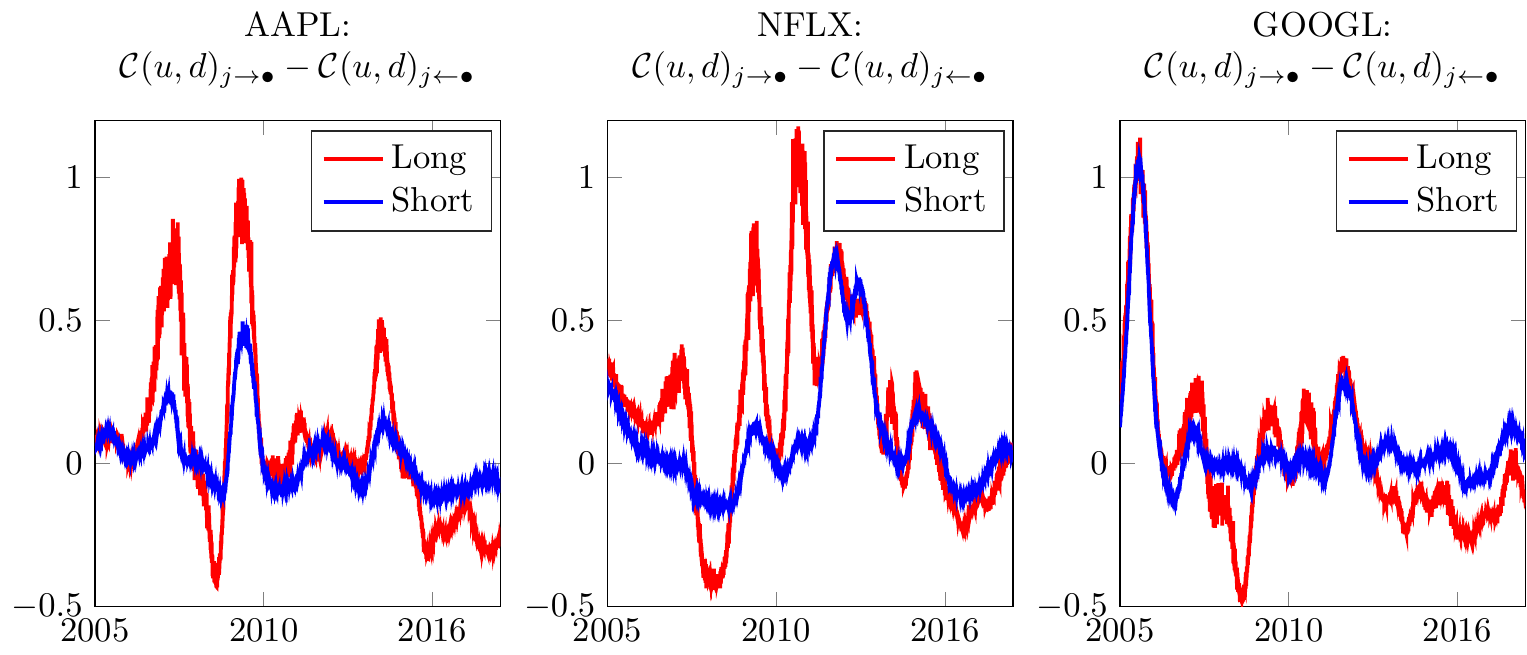}}
		\caption{\textbf{Horizon specific dynamic net-directional connectedness of Apple, Netflix and Google} \\ \small{The top panel of this figure plots the quasi posterior median of the horizon specific dynamic net-directional connectedness of: Apple, AAPL; Netflix, NFLX; and Google, GOOGL. Net-directional connectedness is computed as the difference between \textsc{to} connectedness and \textsc{from} connectedness as $\mathcal{C}(u,d)_{j \rightarrow \bullet}-\mathcal{C}(u,d)_{j\leftarrow \bullet} $  from July 5, 2005 to August 31, 2018. S refers to the short-term which we define as 1 day to 1 week; and L refers to long-term which we define as horizons $>$ 1 week. The spectrum with which the horizons stem from are linked to the frequency with which the data is observed.}}
      \label{NDC_EG}
\end{figure}

\subsection{How to measure dynamic horizon specific network risk?}
\cite{fama1993common} and the vast ensuing literature use data on stock characteristics to create factors that approximate the stochastic discount factor. This search for potential factor candidates results in hundreds to choose from \citep{harvey2016and,mclean2016does}. Typically, the approach conducts portfolio sorts and creates long-short portfolios to approximate risk factors. Despite the success of potential asset pricing factors, many bear no theoretical grounding\footnote{There is an emerging literature proposing econometric techniques to select only those factors that are meaningful (see e.g. \cite{feng2020taming}).}. 

The economy we outline in Section \ref{lit} provides a theoretical framework where the stochastic discount factor loads on dynamic horizon specific network connections among asset return volatilities that are time dependent. Our expression for innovations in the stochastic discount factor shows that dynamic network risk influences returns through the precautionary savings channel which implies negative risk prices. Although network connections are made through state variables and are thus unobservable, we are able to approximate dynamic horizon specific network connections using the measures we introduce earlier. By definition, these measures allow for asymmetric network connections among asset return volatilities that vary over time and are horizon specific.

Note that shock transmission and reception together create dynamic network risk. Thus, sources of systematic risk from the dynamic network are in the directional connections among asset return volatilities. Another key aspect is that assets (nodes) on the market all differ in size; something that can bear influence on shock propagation \citep{elliott2014financial,glasserman2016contagion}. \cite{herskovic2018networks} shows that network concentration, the dominance of few large sectors, is an important network-wide property one should account for; Figure \ref{multilayer_networks} illustrates the importance of network concentration in the case of individual stocks.

In the context of our study, we interpret network concentration as a function of firm size and combine this with shock transmission and reception capacity. Economically speaking, shocks to the volatilities' of key firms on the market, that are more central to the network from a directional perspective, are more likely to contribute to rises in dynamic network risk; relative to those with peripheral directional connections\footnote{We are not dismissing the possibility that small stocks with strong directional connections can create risks investors demand compensation for, or that the \textquoteleft snowball\textquoteright\: effect in \cite{elliott2014financial} for small stocks cannot lead to systematic shocks; which by definition our proxies for horizon specific directional network risk can capture.}.  

With the above in mind, we proceed following the factor literature by conducting portfolio sorts stemming from horizon specific directional network connections. Specifically, we conduct double sorts with daily rebalancing for S\&P500 assets based on their relative size and strength of net directional connections using horizon specific, $\mathcal{C}(u,d)_{j \rightarrow \bullet}-\mathcal{C}(u,d)_{j\leftarrow \bullet}$, measures. This allows us to account for both shock propagation and reception capacity, and also to extract information from the entire adjacency matrix. 

Tables \ref{tab:sort1} reports the average annual returns of value weighted quintile portfolios from our double sorting procedure. Panels A and B report double sorts on net directional connections over the short-term and long-term respectively. We define short-term as 1 day--1 week and long-term as horizons as 1 week--$\infty$ (i.e. $>$ 1 week)\footnote{Note that definitions of horizons stems from the frequency with which data is observed in constructing the network. While we choose the bands that define horizons in a subjective manner we believe these assumptions are reasonable. More generally speaking one could use more bands that span the spectrum.}. Panel C reports average annual returns of value weighted quintile portfolios sorted on aggregate net directional connectedness. \textsc{from} denotes portfolios using assets that are most receptive to shocks from other assets in the network. \textsc{to} denotes portfolios using assets that are the strongest shock propagators to other assets in the network. The final column reports quintile portfolios of an average of the \textsc{from} and \textsc{to} portfolios, which summarizes horizon specific directional network connections. % Hence, in a \cite{diebold2014} sense, \textsc{from} portfolios consist of net shock receivers and \textsc{to} portfolios consist of assets that are net shock transmitters.

In general, portfolio returns are monotonically increasing with size. We also see that \textsc{from} portfolios earn higher returns than \textsc{to} portfolios. However, we do not observe a monotonic relationship as we move from the \textsc{from} portfolios to the \textsc{to} portfolios; nor should we expect to. As we outline above, shock transmission and reception may both create network risks in a market \citep{branger2018equilibrium}. Therefore one should expect to see portfolios using assets in the tails of the cross-sectional distribution of net directional connections earning relatively lower returns than those using asset in middle quintiles; and overall this is the case. We also report the average annual returns of small minus big portfolios at each quintile and for the equally weighted average of the \textsc{from} and \textsc{to} portfolios. Note all portfolio returns are negative and statistically significant at conventional levels.  

% Table generated by Excel2LaTeX from sheet 'XX'
\begin{table}[!hp]
  \centering
  \caption{\textbf{Average annual returns for quintile portfolios on size and horizon specific directional network connections}\\
  \small{Notes: This table reports average annual returns from portfolio double sorts. Specifically portfolio sorts are separated by size and their horizon specific directional network connectedness measures using daily rebalancing from July 5, 2005 to August 31, 2018. Panel A sorts on size and $d$=short-term net-directional connections which we define as 1 day--1 week; Panel B sorts on size and $d$=long-term net-directional connections which we define as 1 week--$\infty$;  Panel C sorts on size and uses $d$=aggregate net-directional connections that sums over short-term and long-term frequency bands.}
  }
    \begin{tabular}{rrrrrrr}
    \toprule
    \midrule
    \midrule
    \multicolumn{1}{l}{\textbf{A:}} & \multicolumn{6}{c}{$d$=\textbf{Short}} \\
          &       &       &       &       &       & \multicolumn{1}{l}{\textsc{from}\textbf{/2+}} \\
          & \multicolumn{1}{l}{\textbf{1} \textsc{from}} & \textbf{2} & \textbf{3} & \textbf{4} & \textbf{5} \textsc{to} & \textsc{to}\textbf{/2} \\
    \midrule
    \textbf{1 Small} & -1.28\% & -4.87\% & -7.75\% & -5.91\% & -22.32\% & -11.80\% \\
    \textbf{2} & 6.29\% & 4.75\% & 4.93\% & 0.81\% & -3.67\% & 1.31\% \\
    \textbf{3} & 12.68\% & 5.73\% & 4.78\% & 5.16\% & 3.69\% & 8.18\% \\
    \textbf{4} & 10.51\% & 9.59\% & 7.71\% & 7.24\% & 5.60\% & 8.05\% \\
    \textbf{5 Big} & 9.79\% & 10.82\% & 8.96\% & 9.82\% & 5.22\% & 7.51\% \\
    \textbf{Small-Big} & -11.07\% & -15.69\% & -16.71\% & -15.72\% & -27.54\% & -19.31\% \\
    \textbf{$t$-stat} & -13.11 & -14.58 & -13.78 & -12.28 & -10.08 & -12.83 \\
    \midrule
    \multicolumn{1}{l}{\textbf{B:}} & \multicolumn{6}{c}{$d$=\textbf{Long}} \\
          &       &       &       &       &       & \multicolumn{1}{l}{\textsc{from}\textbf{/2+}} \\
          & \multicolumn{1}{l}{\textbf{1} \textsc{from}} & \textbf{2} & \textbf{3} & \textbf{4} & \textbf{5} \textsc{to} & \textsc{to}\textbf{/2} \\
    \midrule
    \textbf{1 Small} & -5.05\% & -5.54\% & -5.18\% & -6.25\% & -20.09\% & -12.57\% \\
    \textbf{2} & -1.29\% & 4.95\% & 6.35\% & 4.85\% & -1.81\% & -1.55\% \\
    \textbf{3} & 5.08\% & 7.77\% & 6.32\% & 8.32\% & 4.62\% & 4.85\% \\
    \textbf{4} & 5.49\% & 7.54\% & 8.56\% & 10.24\% & 8.98\% & 7.23\% \\
    \textbf{5 Big} & 6.59\% & 11.46\% & 11.23\% & 12.32\% & 4.92\% & 5.75\% \\
    \textbf{Small-Big} & -11.64\% & -16.99\% & -16.41\% & -18.57\% & -25.01\% & -18.32\% \\
    \textbf{$t$-stat} & -12.61 & -13.36 & -18.71 & -13.06 & -8.63 & -11.73 \\  
    \midrule
    \midrule
    \multicolumn{1}{l}{\textbf{C:}} & \multicolumn{6}{c}{$d$=\textbf{Aggregate}} \\
          &       &       &       &       &       & \multicolumn{1}{l}{\textsc{from}\textbf{/2+}} \\
          & \multicolumn{1}{l}{\textbf{1} \textsc{from}} & \textbf{2} & \textbf{3} & \textbf{4} & \textbf{5} \textsc{to} & \textsc{to}\textbf{/2} \\
    \midrule
    \textbf{1 Small} & -4.33\% & -4.73\% & -5.18\% & -5.96\% & -22.07\% & -13.20\% \\
    \textbf{2} & 0.55\% & 4.28\% & 6.90\% & 5.07\% & -3.69\% & -1.57\% \\
    \textbf{3} & 8.14\% & 6.82\% & 4.95\% & 8.00\% & 4.20\% & 6.17\% \\
    \textbf{4} & 7.81\% & 6.64\% & 8.41\% & 9.65\% & 8.17\% & 7.99\% \\
    \textbf{5 Big} & 7.32\% & 11.73\% & 9.73\% & 11.40\% & 5.49\% & 6.40\% \\
    \textbf{Small-Big} & -11.65\% & -16.46\% & -14.91\% & -17.36\% & -27.56\% & -19.61\% \\
    \textbf{$t$-stat} & -14.74 & -13.52 & -14.97 & -14.16 & -8.78 & -11.28 \\
    \midrule
    \midrule
    \bottomrule
    \end{tabular}%
  \label{tab:sort1}%
\end{table}%

\subsubsection{Factor creation}

Using the results in Table \ref{tab:sort1}, we construct factors in a manner that summarize dynamic horizon specific directional network risk among asset return volatilities whilst controlling for network concentration (i.e size of nodes). We define short and long horizons as 1 day--1 week, and as horizons $>$ 1 week respectively. We also sort on aggregate, directional network connections which sums network connections over horizons. Specifically, each day we sort S\&P500 stocks above and below the day's median price level and define these as small and big stocks. Then, conditional on size, we sort on horizon specific net-directional connectedness and create value-weighted \textsc{to} and \textsc{from} portfolios. These portfolios admit stocks above (below) the 70$^{\text{th}}$ (30$^{\text{th}}$) percentile of each respective day's horizon specific net-directional connectedness distribution. We characterize these portfolios as \textsc{to} and \textsc{from} portfolios. We then take an average of the \textsc{to} and \textsc{from} portfolios and then a long--short position in the respective small and big portfolios. 

Formally, the day $t$ directional network risk factor $\text{NET}(d)_{t}$, at horizon $d=\{S,\:L,\:A\}$, where $S$=short-term which corresponds to horizons of 1-day to 1-week, $L$=long-term which corresponds to horizons 1-week to $\infty$, and $A$ is an aggregate that sums over all horizons (i.e. 1-day to $\infty$), is given by

\begin{equation}\label{eq:factor}
\text{NET}(d)_{t} = \frac{\left(\textsc{from}^{\text{small}}(d)_{t}+\textsc{to}^{\text{small}}(d)_{t}\right)}{2} - \frac{\left(\textsc{from}^{\text{big}}(d)_{t}+\textsc{to}^{\text{big}}(d)_{t}\right)}{2}
\end{equation}

where small and big refer to S\&P500 constituents falling below and above the median daily prices of all constituents respectively. Our dynamic horizon specific network risk factors implicitly account for network concentration through firm size. Therefore our proxies for dynamic horizon specific directional network risk use network-wide properties that summarize information using the entire adjacency matrix \citep{herskovic2018networks}.

% Table generated by Excel2LaTeX from sheet 'VW_final'
\begin{table}[!hp]
  \centering
  \caption{\textbf{Descriptive statistics of network risk factors}\\
  \small{Notes: This table reports descriptive statistics for daily horizon specific and aggregate network risk factors from July 5, 2005 to August 31, 2018. Short corresponds to the short-term network risk factor which captures directional network risk from 1-day to 1 week. Long is the long-term network risk factor that captures directional network risk from horizons 1-week to $\infty$. Aggregate, sums over horizons and thus captures aggregate directional network risk at all horizons. The top half of the table reports the annualized expected returns and standard deviations, along with sample skewness and kurtosis. The bottom half of this table reports sample correlations on and below the main diagonal, and sample covariances above the main diagonal.}  
  }
    \begin{tabular}{lrrr}
    \toprule
    \midrule
    \midrule
          & \multicolumn{1}{c}{\textbf{Short}} & \multicolumn{1}{c}{\textbf{Long}} & \multicolumn{1}{c}{\textbf{Aggregate}} \\
    \midrule
    \textbf{Annualized Expected Return} & -7.34\% & -7.64\% & -7.51\% \\
    \textbf{Annualized Standard Deviation} & 5.94\%  & 6.29\% & 6.39\% \\
    \textbf{Sample Skewness} & -0.36 & -0.40 & -0.22 \\
    \textbf{Sample Kurtosis} & 11.45  & 13.06 & 13.80 \\
    \midrule
    \midrule
   \multicolumn{4}{l}{\textbf{Sample Correlations/Covariances}}  \\
          & \multicolumn{1}{c}{\textbf{Short}} & \multicolumn{1}{c}{\textbf{Long}} & \multicolumn{1}{c}{\textbf{Aggregate}} \\
          \midrule
    \multicolumn{1}{r}{\textbf{Short}} & 1.00  & 0.14 & 0.14 \\
    \multicolumn{1}{r}{\textbf{Long}}  & 0.93  & 1.00     & 0.16 \\
    \multicolumn{1}{r}{\textbf{Aggregate}} & 0.95  & 0.98  & 1.00 \\
    \midrule
    \midrule
    \bottomrule
    \end{tabular}%
  \label{tab:descriptive}%
\end{table}%

Table \ref{tab:descriptive} shows descriptive statistics for short-term, long-term, and aggregate directional network risk factors. The top half of the table reports annualized expected returns and standard deviations, as well as sample skewness and kurtosis. The bottom half reports sample correlations and covariances. Several observations emerge from Table \ref{tab:descriptive}. First, all factors earn considerable negative returns, are volatile, skewed and leptokurtic. These negative returns are consistent with: economic theory relating to volatility risk from an ICAPM perspective \citep{campbell2018intertemporal}; the model and empirics within \cite{herskovic2018networks} that documents negative returns of long-short portfolios of high and low concentration (i.e. size relative to the whole network) sensitive assets in a production based asset pricing model; and the model we outline in Section \ref{economy}. The channel is that firms loading positively on dynamic horizon specific directional network risk earn lower returns as they act as inter-temporal hedging devices against consumption declines.

Second, by definition dynamic horizon specific network risk factors are highly correlated. This is because the horizon specific net directional connections are a decomposition of aggregate net directional connections. Therefore in order to appropriately interpret asset sensitivities to horizon specific network risk that includes both short-term and long-term directional network risk factors, we orthogonalize by using the residuals from a regression of the short-term network risk factor on the long-term network risk factor. Note that the annualized expected return and standard deviation are -0.01\% and 2.25\% respectively. The sample skewness and kurtosis of the orthogonalized factor are given by -0.15 and 7.87 respectively. The sample correlation between the orthogonalized short-term directional network risk factor and the aggregate directional network risk factor is 0.11.

\begin{figure}[!hp]
\centering
\scalebox{1.00}{\includegraphics{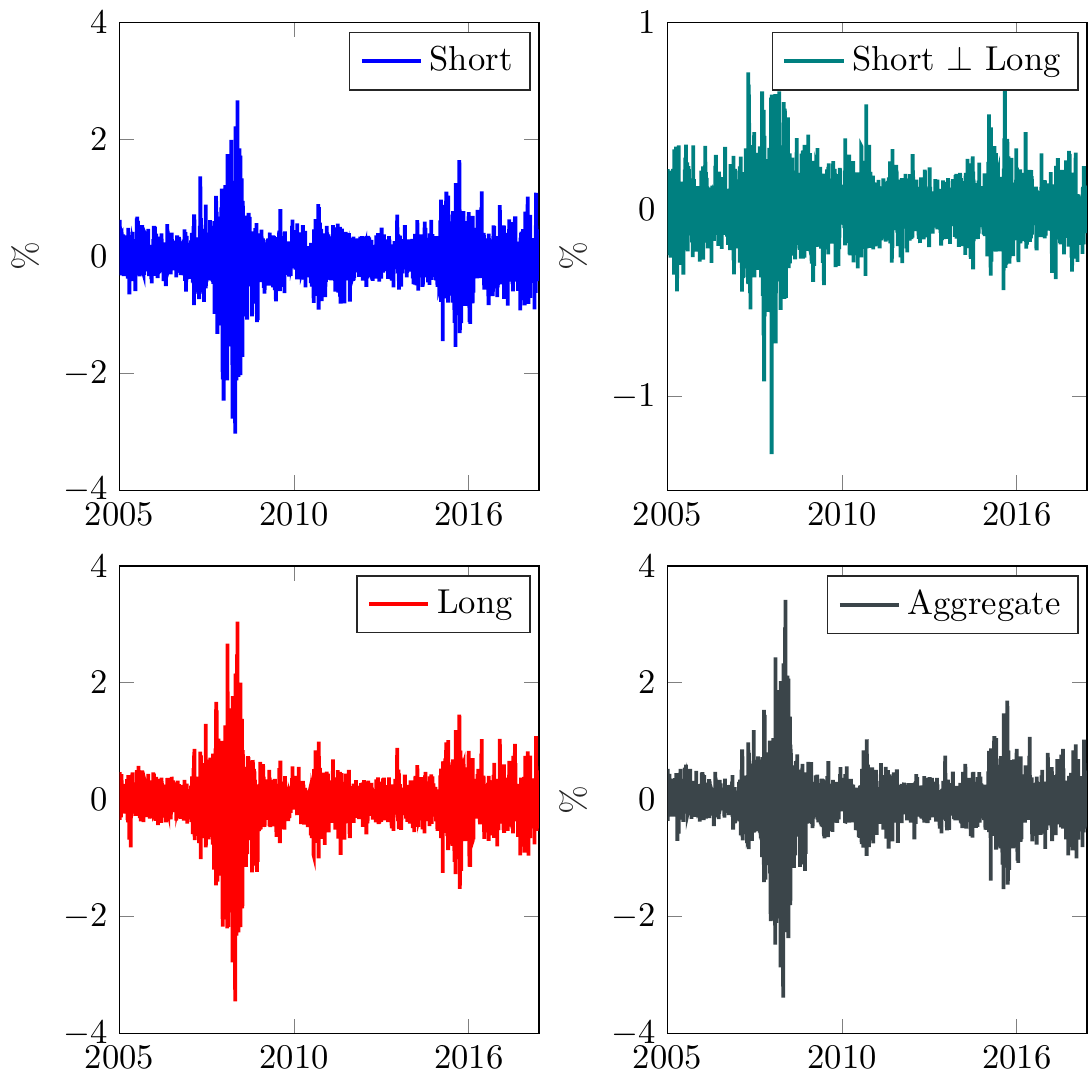}}
		\caption{\textbf{Horizon specific and aggregate directional network risk factors}\\ \small{The top left and right hand side panels of this figure plot the short-term directional and orthogonalized short-term directional network (Short $\perp$ Long) risk factor daily returns respectively. The bottom left and right hand side panels report the long-term directional network risk factor daily returns and the aggregate directional network risk factor daily returns respectively. The sample period is from July 5, 2005 to August 31, 2018. Short-term is defined as 1 day--1 week and long-term is defined as horizons $>$ 1 week. The bottom panel of this figure plots the aggregate directional network risk factor daily returns from July 5, 2005 to August 31, 2018. These returns stem from sorts on unconditional net directional connectedness (i.e. sums net directional connectedness over horizons).}}
      \label{FACTOR_RETS}
\end{figure}

Figure \ref{FACTOR_RETS} plots the daily returns of our network risk factors. We report: the short-term directional network risk factor in the top left panel; the orthogonalized daily return for the short-term directional network risk factor in the top right panel; the long-term directional network risk factor in the bottom left panel; and the aggregate directional network risk factor in the bottom panel. We can see that surges in the volatility of factor returns during the 2007--2009 recession, 2011--2012 particularly for long-term factor returns, and during 2015. As we outline above, these periods contain notable events where directional network risk heightens, and is thus mirrored by larger price fluctuations of our factors.

\newpage

\section{Dynamic Horizon Specific Network Risk Pricing}\label{main_results}
This section contains our main results on the pricing of dynamic horizon specific network risk in the cross-section of S\&P500 returns. We examine the contemporaneous link between factor loadings and returns. Our tests employ individual stocks as our base assets. We do this for two reasons. First, our focus is on directional network connections among asset return volatilities and the cross-sectional pricing implications of (horizon specific) directional network risk. Therefore, tracking network connections for a market requires this level of granularity. Second, \cite{ang2020using} show that creating portfolios ignores the fact that stocks within portfolios have different betas which can lead to larger standard errors in cross-sectional risk premia estimates. We first present Fama-MacBeth regressions before moving on to consider portfolio sorts.

\subsection{Fama-MacBeth Regressions}

In our Fama-MacBeth analysis, we run two-stage regressions of daily individual S\&P500 excess returns on horizon specific and aggregate network risk whilst controlling for an array of factors existing in the literature. In particular, we control for the \cite{fama2015five} five-factors; as well as five additional factors. The additional characteristics we control for are: changes in the VIX; momentum; conditional skewness \citep{harvey2000conditional}; conditional kurtosis \citep{dittmar2002nonlinear}; and illiquidity \citep{amihud2002illiquidity}. The VIX is from CRSP and the momentum factor is from Ken French's data library. Conditional skewness is the sample counterpart of 
\begin{equation*}
\text{CSKEW} = \frac{\mathbb{E}\left[(R_{i,t}-\mu_{i,t}) \cdot (\text{MKT}_{t} - \mu_{\text{MKT},t})^2 \right] }{\sqrt{\text{var}(R_{i,t})} \cdot \text{var}(\text{MKT}_{t})}
\end{equation*}
with $\mu_{i,t},\: \mu_{\text{MKT},t}$ being the average excess return on stock $i$ and the market respectively. We define conditional kurtosis, CKURT, analogously. We construct factors that sort stocks according to their conditional skewness and kurtosis and construct long-short portfolios using daily rebalancing\footnote{We also include firm specific values of CSKEW and CKURT in the time-series regressions to obtain betas; results are robust to this specification and our dynamic horizon specific directional network risk proxies remain statistically significant and economically meaningful.}. Illiquidity is Amihud's illiquidity measure which is the ratio of absolute return to trading volume. We compute this for each S\&P500 constituent and take a cross-sectional average each day as a proxy for market illiquidity.

Tables \ref{tab:FMB_horz} and \ref{tab:Rolling_FMB_horz} present cross-sectional pricing results using daily returns of S\&P500 constituents using full sample estimates and daily rolling regressions with a 3-year window. Columns 1 and 2 in each table show results from the CAPM and the \cite{fama2015five} 5-factor model. The remaining columns report results from alternative specifications incorporating dynamic horizon specific network risk. Note that models including both short-term and long-term directional network risk use the short-term directional network risk factor that is orthogonal to long-term directional network risk.

From Table \ref{tab:FMB_horz}, we can see that the market prices of risk are all negative and statistically significant. Note that the estimates for market prices and statistical significance of dynamic horizon specific directional network risk are quantitatively similar across specifications 3, 4, 5, and 6; with the latter controlling for five factors over and above the \cite{fama2015five} model. When we allow for orthogonalized short-term directional network risk and long-term directional network risk in columns 7 and 8, we observe statistically significant negative market prices of risk\footnote{By construction, the estimates of market prices of orthogonal short-term directional network risk are smaller in absolute value relative to specifications containing only short-term directional network risk.}. 

Turning to Table \ref{tab:Rolling_FMB_horz}, a similar story emerges, with the exception of orthogonal short-term directional network risk becoming statistically insignificant in columns 7 and 8. We also observe that accounting for additional factors beyond \cite{fama2015five} do not dampen the significance or magnitude of the market prices of dynamic horizon specific network risk. 

We now examine market prices of dynamic aggregate network risk. Table \ref{tab:FMB_T_FS_ROLL} reports full sample and rolling regression results in Panels A and B respectively. Again, we find a significantly negative market price of network risk. The estimates change marginally when accounting for our additional battery of factors. 

Overall, models containing only the market risk premium or the \cite{fama2015five} five-factor model result in negative prices for the market risk premium. This suggests model misspecification. Adding to this, we can see that the intercepts across all specifications including directional network risk are statistically insignificant. Thereby implying that these specifications adequately explain cross-sectional return variation. 

Several implications emerge from this analysis. First, horizon specific network risk is priced in the cross-section of stock returns and the market price of risk is negative. Second, this result is robust across multiple specifications and holds using full-sample and rolling window estimates. We note here that the rolling specifications in columns 7 and 8 of Table \ref{tab:Rolling_FMB_horz} using short-term directional network risk orthogonal to long-term directional risk is statistically insignificant and the mean prices of risk very close to zero. Third, the effect of dynamic horizon specific network risk is economically meaningful. Focussing on columns 5 and 6 of Table \ref{tab:Rolling_FMB_horz} to obtain an idea of the economic significance of these risk premia, a one-standard-deviation increase across stocks in $\beta^{\text{NET}(S)}$, $\beta^{\text{NET}(L)}$ implies an annualized fall in returns of 6.76\% and 7.66\% respectively\footnote{The standard deviation of $\beta^{\text{NET}(S)}$=0.81, $\beta^{\text{NET}(L)}$=0.79. Taking the market risk premia of -0.034 and -0.04 we have: i) -0.034 $\times$ 0.81 =-6.76\% and -7.66\% respectively. To put this into perspective, a two-standard deviation increase across stocks in $\beta^{\text{MOM}}$, with the standard deviation  results in an annualized increase in returns of 9.16\% and 8.40\% respectively. Note, similar values are implied from alternative specifications.}.

\newpage 

%\begin{landscape}
% Table generated by Excel2LaTeX from sheet 'Full_Sample'
\begin{table}[!hp]
  \centering
  \caption{\textbf{Full sample Fama-MacBeth regressions: Dynamic horizon specific network risk}
\small{Notes: This table reports Fama-MacBeth regressions for daily S\&P500 stock returns from July 5, 2005 to August 31, 2018. We use Newey West Standard errors with 12-lags, $t$-ratios in square brackets below coefficient estimates are adjusted following \cite{shanken1992estimation}. NET($S$) is the short-term directional network risk factor and NET($L$) is the long-term directional network risk factor. For models including both NET($S$) and NET($L$), we use the short-term directional network risk factor that is orthogonal to the long-term directional network risk factor, NET($S^{\perp}$). VIX is the daily change in the VIX index; MOM is the momentum factor; CSKEW and CKURT are conditional skewness and conditional kurtosis factors respectively; ILLIQ is market illiquidity.}   
   }
    %\adjustbox{max height=\textwidth ,max width=9.70in, keepaspectratio}{
     \adjustbox{max height=9.70in, max width=\textwidth , keepaspectratio}{
    \begin{tabular}{lrrrrrrrr}
    \toprule
    \midrule
    \midrule
%    \textbf{Model:} & \multicolumn{1}{l}{CAPM} & \multicolumn{1}{l}{FF} & \multicolumn{1}{l}{FF} & \multicolumn{1}{l}{FF} & \multicolumn{1}{l}{FF + NET($S$)} & \multicolumn{1}{l}{FF + NET($L$)} & \multicolumn{1}{l}{FF + NET($S^{\perp}$)} & \multicolumn{1}{l}{FF + NET($S^{\perp}$)} \\
%          &       &  & \multicolumn{1}{l}{+ NET($S$)} & \multicolumn{1}{l}{+ NET($L$)}  & %\multicolumn{1}{l}{+ Extra} & \multicolumn{1}{l}{+ Extra} & \multicolumn{1}{l}{+ NET($L$) } & \multicolumn{1}{l}{+ NET($L$) + Extra } \\
              \textbf{Model:} & \multicolumn{1}{c}{1} & \multicolumn{1}{c}{2} & \multicolumn{1}{c}{3} & \multicolumn{1}{c}{4} & \multicolumn{1}{c}{5} & \multicolumn{1}{c}{6} & \multicolumn{1}{c}{7} & \multicolumn{1}{c}{8} \\
    \midrule
    NET($S$)/NET($S^{\perp}$) &       &       & -0.077 &       & -0.063 &       & -0.022 & -0.016 \\
          &       &       & [-6.50] &       & [-6.21] &       & [-2.42] & [-1.92] \\
    NET($L$) &       &       &       & -0.079 &       & -0.064 & -0.064 & -0.056 \\
          &       &       &       & [-6.34] &       & [-5.75] & [-4.77] & [-4.47] \\
    MKT   & -0.094 & -0.083 & 0.053 & 0.043 & 0.065 & 0.065 & 0.052 & 0.061 \\
          & [-2.58] & [-2.09] & [1.58] & [1.30] & [1.73] & [1.76] & [1.56] & [1.68] \\
    SMB   &       & -0.02 & -0.006 & 0.000     & 0.006 & 0.009 & -0.009 & 0.004 \\
          &       & [-1.18] & [-0.35] & [0.01] & [0.42] & [0.58] & [-0.59] & [0.30] \\
    HML   &       & -0.042 & -0.021 & -0.01 & -0.014 & -0.010 & -0.027 & -0.018 \\
          &       & [-2.77] & [-1.33] & [-0.66] & [-0.92] & [-0.66] & [-1.88] & [-1.27] \\
    RMW   &       & 0.047 & 0.026 & 0.025 & 0.011 & 0.011 & 0.027 & 0.012 \\
          &       & [3.00] & [1.67] & [1.64] & [0.78] & [0.80] & [1.81] & [0.95] \\
    CMA   &       & -0.036 & -0.025 & -0.025 & -0.011 & -0.012 & -0.026 & -0.011 \\
          &       & [-3.00] & [-2.06] & [-2.06] & [-1.10] & [-1.24] & [-2.17] & [-1.08] \\
    VIX   &       &       &       &       & -0.186 & -0.182 &       & -0.164 \\
          &       &       &       &       & [-0.45] & [-0.44] &       & [-0.40] \\
    MOM   &       &       &       &       & 0.146 & 0.150 &       & 0.151 \\
          &       &       &       &       & [4.06] & [4.16] &       & [4.19] \\
    CSKEW &       &       &       &       & 0.028 & 0.038 &       & 0.024 \\
          &       &       &       &       & [1.14] & [1.51] &       & [0.97] \\
    CKURT &       &       &       &       & 0.005 & 0.002 &       & 0.007 \\
          &       &       &       &       & [0.21] & [0.11] &       & [0.30] \\
    ILLIQ &       &       &       &       & 0.008 & 0.007 &       & 0.008 \\
          &       &       &       &       & [3.05] & [2.93] &       & [3.13] \\
    Intercept & 0.08  & 0.068 & 0.01  & 0.02  & -0.005 & -0.004 & 0.007 & -0.006 \\
          & [5.05] & [3.80] & [0.65] & [1.29] & [-0.31] & [-0.26] & [0.48] & [-0.36] \\
%    \midrule
%    $\bar{R}^2$ & 0.17  & 0.23  & 0.41  & 0.39  & 0.59  & 0.59  & 0.41  & 0.59 \\
    \midrule
    \midrule
    \bottomrule
    \end{tabular}%
    }
  \label{tab:FMB_horz}%
\end{table}%
%\end{landscape}

\newpage

%\begin{landscape}
% Table generated by Excel2LaTeX from sheet 'Rolling'
\begin{table}[!hp]
  \centering
  \caption{\textbf{Rolling Fama-MacBeth regressions: Dynamic horizon specific network risk}
\small{Notes: This table reports daily rolling Fama-MacBeth regressions for daily S\&P500 stock returns from July 5, 2005 to August 31, 2018 using a 3-year window. We use Newey West Standard errors with 24-lags, $t$-ratios in square brackets below coefficient estimates are adjusted following \cite{shanken1992estimation}. NET($S$) is the short-term directional network risk factor and NET($L$) is the long-term directional network risk factor. For models including both NET($S$) and NET($L$), we use the short-term directional network risk factor that is orthogonal to the long-term directional network risk factor, NET($S^{\perp}$). VIX is the daily change in the VIX index; MOM is the momentum factor; CSKEW and CKURT are conditional skewness and conditional kurtosis factors respectively; ILLIQ is market illiquidity.}   
   }
    %\adjustbox{max height=\textwidth ,max width=9.70in, keepaspectratio}{
     \adjustbox{max height=9.70in, max width=\textwidth , keepaspectratio}{
    \begin{tabular}{lrrrrrrrr}
    \toprule
    \midrule
    \midrule
%    \textbf{Model:} & \multicolumn{1}{l}{CAPM} & \multicolumn{1}{l}{FF} & \multicolumn{1}{l}{FF} & \multicolumn{1}{l}{FF} & \multicolumn{1}{l}{FF + NET($S$)} & \multicolumn{1}{l}{FF + NET($L$)} & \multicolumn{1}{l}{FF + NET($S^{\perp}$)} & \multicolumn{1}{l}{FF + NET($S^{\perp}$)} \\
%          &       &  & \multicolumn{1}{l}{+ NET($S$)} & \multicolumn{1}{l}{+ NET($L$)}  & %\multicolumn{1}{l}{+ Extra} & \multicolumn{1}{l}{+ Extra} & \multicolumn{1}{l}{+ NET($L$) } & \multicolumn{1}{l}{+ NET($L$) + Extra } \\
              \textbf{Model:} & \multicolumn{1}{c}{1} & \multicolumn{1}{c}{2} & \multicolumn{1}{c}{3} & \multicolumn{1}{c}{4} & \multicolumn{1}{c}{5} & \multicolumn{1}{c}{6} & \multicolumn{1}{c}{7} & \multicolumn{1}{c}{8} \\
    \midrule
    NET($S$)/NET($S^{\perp}$) &       &       & -0.043 &       & -0.034 &       & 0.007 & 0.008 \\
          &       &       & [-2.40] &       & [-2.24] &       & [0.77] & [0.85] \\
    NET($L$) &       &       &       & -0.048 &       & -0.040 & -0.052 & -0.044 \\
          &       &       &       & [-2.53] &       & [-2.41] & [-2.67] & [-2.56] \\
    MKT   & -0.103 & -0.015 & 0.022 & 0.027 & 0.015 & 0.021 & 0.028 & 0.020 \\
          & [-1.22] & [-0.10] & [0.39] & [0.44] & [0.26] & [0.34] & [0.47] & [0.34] \\
    SMB   &       & -0.021 & 0.005 & 0.008 & 0.004 & 0.006 & 0.009 & 0.007 \\
          &       & [-0.85] & [0.15] & [0.28] & [0.13] & [0.20] & [0.30] & [0.24] \\
    HML   &       & -0.049 & -0.025 & -0.021 & -0.019 & -0.017 & -0.022 & -0.016 \\
          &       & [-1.55] & [-0.77] & [-0.70] & [-0.70] & [-0.62] & [-0.73] & [-0.62] \\
    RMW   &       & 0.042 & 0.023 & 0.022 & 0.009 & 0.008 & 0.024 & 0.010 \\
          &       & [1.39] & [0.70] & [0.71] & [0.34] & [0.30] & [0.82] & [0.41] \\
    CMA   &       & -0.0007 & 0.01  & 0.01  & 0.001 & 0.002 & 0.008 & -0.001 \\
          &       & [0.11] & [0.59] & [0.60] & [0.05] & [0.06] & [0.48] & [-0.07] \\
    VIX   &       &       &       &       & 0.003 & -0.022 &       & -0.017 \\
          &       &       &       &       & [-0.01] & [-0.06] &       & [-0.05] \\
    MOM   &       &       &       &       & 0.087 & 0.080 &       & 0.080 \\
          &       &       &       &       & [2.06] & [1.96] &       & [1.99] \\
    CSKEW &       &       &       &       & 0.031 & 0.030 &       & 0.032 \\
          &       &       &       &       & [0.47] & [0.42] &       & [0.46] \\
    CKURT &       &       &       &       & 0.006 & 0.005 &       & 0.007 \\
          &       &       &       &       & [0.69] & [0.66] &       & [0.71] \\
    ILLIQ &       &       &       &       & 0.007 & 0.007 &       & 0.007 \\
          &       &       &       &       & [1.12] & [1.09] &       & [1.06] \\
    Intercept & 0.08  & 0.027 & 0.022 & 0.023 & 0.030 & 0.030 & 0.023 & 0.032 \\
          & [2.43] & [0.94] & [0.79] & [0.83] & [1.25] & [1.25] & [0.85] & [1.35] \\
%    \midrule
%    $\bar{R}^2$ & 0.11  & 0.26  & 0.34  & 0.35  & 0.47  & 0.48  & 0.36  & 0.48 \\
    \midrule
    \midrule
    \bottomrule
    \end{tabular}%
    }
  \label{tab:Rolling_FMB_horz}%
\end{table}%
%\end{landscape}

\newpage

\begin{table}[!hp]
  \centering
  \caption{\textbf{Fama-MacBeth regressions: Aggregate directional network risk}
\small{Notes: This table reports Full Sample Fama-MacBeth regressions in Panel A, and Rolling Fama-MacBeth regressions in Panel B, for daily S\&P500 stock returns. The full sample regressions span from July 5, 2005 to August 31, 2018 and use Newey West Standard errors with 12-lags. Rolling regressions use a 3-year window add 1 day at a time and use Newey West Standard errors with 24-lags.  $t$-ratios in square brackets below coefficient estimates are adjusted following \cite{shanken1992estimation}. NET($A$) is the aggregate directional network risk factor that sums over short-term and long-term frequency bands. VIX is the daily change in the VIX index; MOM is the momentum factor; CSKEW and CKURT are conditional skewness and conditional kurtosis factors respectively; ILLIQ is market illiquidity.} 
} 
    \begin{tabular}{lrrrrr}
    \toprule
    \midrule
    \midrule
          & \multicolumn{2}{c}{\textbf{A: Full Sample}} & & \multicolumn{2}{c}{\textbf{B: Rolling}} \\
\cmidrule(lr){2-3}\cmidrule(lr){5-6} 
%         & \multicolumn{1}{l}{FF + NET($S^{\perp}$)} & \multicolumn{1}{l}{FF + NET($S^{\perp}$)} &  \multicolumn{1}{l}{FF + NET($S^{\perp}$)} & \multicolumn{1}{l}{FF + NET($S^{\perp}$)} \\
%          & \multicolumn{1}{l}{+ NET($L$)} & \multicolumn{1}{l}{+ NET($L$) + Extra} & \multicolumn{1}{l}{+ NET($L$)} & \multicolumn{1}{l}{+ NET($L$) + Extra} \\
                    & \multicolumn{1}{c}{1} & \multicolumn{1}{c}{2} & & \multicolumn{1}{c}{3} & \multicolumn{1}{c}{4} \\
    \midrule
    NET($A$) & -0.079 & -0.065 & & -0.047 & -0.037 \\
          & [-6.36] & [-5.94] & & [-2.44] & [-2.32] \\
    MKT & 0.050 & 0.066 & & 0.026 & 0.02 \\
          & [1.48] & [1.75] & & [0.43] & [0.33] \\
    SMB & -0.002 & 0.008 & & 0.006  & 0.005 \\
          & [-0.14] & [0.49] & & [0.22] & [0.18] \\
    HML & -0.013 & -0.01 & & -0.023 &  -0.018 \\
          & [-0.81] & [-0.68] & & [-0.75] & [-0.67] \\
    RMW & 0.025 & 0.01  & & 0.022 & 0.008 \\
          & [1.60] & [0.73] & & [0.70] & [0.29] \\
    CMA & -0.026 & -0.012  & & 0.009 & 0.001 \\
          & [-2.16] & [-1.23] & & [0.53] & [0.02] \\
    VIX &       & -0.173 &     &  & -0.022 \\
          &       & [-0.42] &   &    & [-0.07] \\
    MOM &       & 0.145 &    &   & 0.084 \\
          &       & [3.99] &   &    & [2.03] \\
    CSKEW &       & 0.034 &    &   & 0.03 \\
          &       & [1.37] &   &    & [0.45] \\
    CKURT &       & 0.004 &    &   & 0.005 \\
          &       & [0.16] &   &    & [0.67] \\
    ILLIQ &       & 0.007 &    &   & 0.007 \\
          &       & [2.97] &   &    & [1.12] \\
    Intercept & 0.016 & -0.003 & & 0.023 & 0.029 \\
          & [1.10] & [-0.20] &  & [0.83] & [1.23] \\
%    \midrule 
%    $\bar{R}^2$ & 0.40  & 0.59  & & 0.27  & 0.47 \\
    \midrule
    \midrule
    \bottomrule
    \end{tabular}%
  \label{tab:FMB_T_FS_ROLL}%
\end{table}%

\newpage
\subsection{Portfolio Sorts}\label{port_sorts}

We now turn to portfolio sorts. Specifically, we estimate betas for S\&P500 constituents and sort directly on factor loading estimates. Our procedure involves estimating rolling regressions using daily data, a 3-year rolling window, and moving forward through the sample on a day-by-day basis. For each of the $i$ S\&P500 constituents we estimate the following regressions

\begin{equation}
R_{i,t} = \beta_{0,i} + \beta^{\text{MKT}}_{i} \cdot \text{MKT}_{t} + \beta^{\text{NET}(d)}_{i} \cdot \text{NET}(d)_{t} + \varepsilon_{i,t}, \quad d\in \{S, L, A\}
\end{equation}
where $R_{i,t}$ is the excess return of the $i$th S\&P500 constituent on day $t$, MKT_{t} is the excess return on the market portfolio, from Ken French's data library, on day $t$. NET$(d)_{t}$ is the horizon specific network risk factor over the: short-term, $S$ which we define as 1-day to 1-week; long-term, $L$ which we define as horizons $>$ 1-week; and aggregate, $A$ which incorporates network connections over all horizons. We include only the NET$(d)_{t}$ factor we are sorting on, and only control for the market risk premium to obtain loadings to reduce the noise in estimation \citep{ang2006downside,cremers2015aggregate}\footnote{We also conduct the same exercise controlling for the \cite{fama2015five} five-factors, as well as the additional 5 factors we use throughout this paper. Results are qualitatively similar to those we report and are available on request.}. 

From July 10, 2009 we sort stocks into quintiles from the estimates of $\beta^{\text{NET}(d)}_{i}$on that day. We rebalance portfolios monthly during each year and compute annual returns by summing daily returns throughout each year\footnote{Note also that we use Newey West standard errors in constructing $t$-statistics for all rolling regressions with 24 lags. We also start in July 2009 as we use 1-year of data to construct conditional skewness and conditional kurtosis.}. Annualized risk-adjusted returns of quintile portfolios are the alpha of the \cite{fama2015five} five-factor model. We also report portfolio betas with respect to dynamic horizon specific network risk, \cite{fama2015five} five-factors, and the additional battery of factors we use in the Fama-MacBeth regressions; namely the VIX, momentum, conditional skewness, conditional kurtosis, and market illiquidity. We report results for value-weighted portfolios in Table \ref{tab:Port_sorts_VW_HOR}, and results for equal-weighted portfolios in Table \ref{tab:Port_sorts_EW_HOR}. 

Panel A of Table \ref{tab:Port_sorts_VW_HOR} shows average annual returns, Fama-French five-factor alphas and portfolio betas for value-weighted quintile portfolios that we sort on $\beta^{\text{NET}(S)}$; as well as a hedge portfolio that takes a long position in the portfolio containing assets with the highest 20\% loadings, and a short position in the portfolio containing assets with the lowest 20\% loadings. Meanwhile Panel B reports the same characteristics, but for value-weighted quintile portfolios that we sort on $\beta^{\text{NET}(L)}$. Table \ref{tab:Port_sorts_EW_HOR} shows comparable results for equal-weighted portfolios. For completeness, Table \ref{tab:Port_sorts_TOT} shows analogous results for quintile portfolio sorts on aggregate directional network risk betas. Panels A and B show value-weighted and equal-weighted characteristics respectively.

We are able to draw several conclusions from Tables \ref{tab:Port_sorts_VW_HOR}, \ref{tab:Port_sorts_EW_HOR}, and \ref{tab:Port_sorts_TOT}. First, it is clear that stocks whose returns load more positively on horizon specific network risk earn lower returns. This result is robust across average annual returns, and risk adjusted returns. We observe an almost monotonically decreasing relationship in both performance measures. 

Second, there is a monotonically increasing relationship for portfolio betas with respect to dynamic horizon specific network risk. The dynamic network risk betas associated to the hedge portfolios all exceed unity and are statistically significant at 1\% levels. This pattern exists after controlling for a battery of factors proposed in the existing literature. Ultimately, these portfolio sensitivities further justify the pricing of dynamic horizon specific network risk we observe from Fama-MacBeth regressions.

Third, we see economically important, and statistically significant raw and risk adjusted returns stemming from all hedge portfolios in Tables \ref{tab:Port_sorts_VW_HOR}, \ref{tab:Port_sorts_EW_HOR}, and \ref{tab:Port_sorts_TOT}. Specifically, the risk adjusted annual returns for value-weighted hedge portfolios against short-term and long-term dynamic network risk are -3.43\% and -4.05\% respectively. Furthermore, the equal-weighted portfolios imply annual market prices of short-term and long-term dynamic network risk of -7.14\% (-10.27\%/1.439) and -7.76\% (-10.78\%/1.389); of which are close to the annualized mean of each proxy for network risk we report in Table \ref{tab:descriptive}. The same holds true for results in Panel B of Table \ref{tab:Port_sorts_TOT}\footnote{A possible reason for the larger returns in value weighted portfolios may stem from the fact that our proxies for network risk control for concentration, and arguably importance, of assets within the network. From a theoretical standpoint, one could conjecture that those value weighted portfolios containing the largest stocks on the S\&P500 are relatively insulated from shock propagation, or that the feedback loop of own shocks is minimal.}.

\begin{landscape}
% Table generated by Excel2LaTeX from sheet 'P_sortshor_CAPM_NET'
\begin{table}[!hp]
  \centering
  \caption{\textbf{Contemporaneous characteristics of value-weighted portfolios}\\
\small{Notes: We create value-weighted portfolios by sorting stocks into quintiles based on daily realized horizon specific network risk betas; short-term in Panel A and long-term in Panel B. Betas are from daily 3-year rolling regressions and the sample spans July 10, 2009 to August 31, 2018. Portfolios are rebalanced monthly. The portfolio characteristics are: average annual returns, $R_{p}$; annualized \cite{fama2015five} five-factor alphas, Ann. $\alpha^{\text{FF}}$; and average annual betas with respect to dynamic horizon specific network risk, the \cite{fama2015five} five-factors, the VIX, Momentum, conditional skewness, conditional kurtosis and market illiquidity. The latter are obtained from rolling regressions using a 1-year window that add one-month at a time.}  
  }
    \begin{tabular}{lrrrrrrrrrrrrr}
    \toprule
    \midrule
    \midrule
    \multicolumn{14}{l}{\textbf{A: Sorts on NET($S$)}} \\
    Portfolio & \multicolumn{1}{l}{$R_{p}$} & \multicolumn{1}{l}{Ann. $\alpha^{\text{FF}}$} & \multicolumn{1}{l}{$\beta^{\text{NET}(S)}$} & \multicolumn{1}{l}{$\beta^{\text{MKT}}$} & \multicolumn{1}{l}{$\beta^{\text{SMB}}$} & \multicolumn{1}{l}{$\beta^{\text{HML}}$} & \multicolumn{1}{l}{$\beta^{\text{RMW}}$} & \multicolumn{1}{l}{$\beta^{\text{CMA}}$} & \multicolumn{1}{l}{$\beta^{\text{VIX}}$} & \multicolumn{1}{l}{$\beta^{\text{MOM}}$} & \multicolumn{1}{l}{$\beta^{\text{CSKEW}}$} & \multicolumn{1}{l}{$\beta^{\text{CKURT}}$} & \multicolumn{1}{l}{$\beta^{\text{ILLIQ}}$} \\
    \midrule
    1 Low $\beta^{\text{NET}(S)}$ & 7.71\% & -0.01\% & 0.085 & 0.568 & -0.004 & -0.274 & 0.021 & 0.082 & 0.003 & 0.073 & -0.054 & -0.179 & -0.008 \\
    2     & 10.86\% & 2.84\% & 0.381 & 0.560 & 0.042 & -0.244 & 0.130 & 0.134 & 0.001 & 0.088 & -0.028 & -0.060 & 0.392 \\
    3     & 9.38\% & 0.54\% & 0.499 & 0.579 & 0.126 & -0.231 & 0.076 & 0.171 & 0.003 & 0.085 & -0.089 & -0.022 & 1.501 \\
    4     & 7.51\% & -1.87\% & 0.785 & 0.608 & 0.189 & -0.090 & 0.037 & 0.096 & 0.005 & 0.073 & -0.114 & 0.066 & 2.218 \\
    5 High $\beta^{\text{NET}(S)}$ & 5.31\% & -3.44\% & 1.316 & 0.634 & 0.201 & -0.061 & -0.186 & 0.227 & 0.006 & 0.099 & -0.285 & -0.085 & 3.915 \\
    5--1 & -2.40\% & -3.43\% & 1.231 & 0.066 & 0.205 & 0.213 & -0.207 & 0.145 & 0.003 & 0.026 & -0.231 & 0.094 & 3.922 \\
    $t$-stat & -2.06 & -4.26 & 9.69  & 1.01  & 3.44  & 2.73  & -1.80 & 1.46  & 0.52  & -0.09 & -1.76 & 0.83  & 0.71 \\
    \midrule
    \midrule
    \multicolumn{14}{l}{\textbf{B: Sorts on NET($L$)}} \\
    Portfolio & \multicolumn{1}{l}{$R_{p}$} & \multicolumn{1}{l}{Ann. $\alpha^{\text{FF}}$} & \multicolumn{1}{l}{$\beta^{\text{NET}(L)}$} & \multicolumn{1}{l}{$\beta^{\text{MKT}}$} & \multicolumn{1}{l}{$\beta^{\text{SMB}}$} & \multicolumn{1}{l}{$\beta^{\text{HML}}$} & \multicolumn{1}{l}{$\beta^{\text{RMW}}$} & \multicolumn{1}{l}{$\beta^{\text{CMA}}$} & \multicolumn{1}{l}{$\beta^{\text{VIX}}$} & \multicolumn{1}{l}{$\beta^{\text{MOM}}$} & \multicolumn{1}{l}{$\beta^{\text{CSKEW}}$} & \multicolumn{1}{l}{$\beta^{\text{CKURT}}$} & \multicolumn{1}{l}{$\beta^{\text{ILLIQ}}$} \\
    \midrule
    1 Low $\beta^{\text{NET}(L)}$ & 7.57\% & -0.09\% & 0.077 & 0.569 & -0.011 & -0.266 & 0.026 & 0.092 & 0.003 & 0.069 & -0.057 & -0.183 & 0.150 \\
    2     & 10.46\% & 2.42\% & 0.365 & 0.556 & 0.043 & -0.226 & 0.126 & 0.142 & 0.001 & 0.106 & -0.015 & -0.057 & 0.242 \\
    3     & 9.88\% & 1.10\% & 0.543 & 0.578 & 0.123 & -0.224 & 0.095 & 0.159 & 0.003 & 0.077 & -0.066 & 0.012 & 2.053 \\
    4     & 8.18\% & -1.29\% & 0.810 & 0.603 & 0.181 & -0.081 & 0.083 & 0.116 & 0.004 & 0.077 & -0.092 & 0.059 & 1.700 \\
    5 High $\beta^{\text{NET}(L)}$ & 4.61\% & -4.13\% & 1.282 & 0.641 & 0.203 & -0.030 & -0.124 & 0.204 & 0.006 & 0.061 & -0.288 & -0.096 & 3.359 \\
    5--1 & -2.96\% & -4.05\% & 1.205 & 0.072 & 0.214 & 0.235 & -0.150 & 0.111 & 0.003 & -0.009 & -0.231 & 0.088 & 3.209 \\
    $t$-stat & -2.55 & -5.30 & 9.62  & 1.05  & 3.27  & 2.88  & -1.32 & 1.13  & 0.47  & -0.50 & -1.50 & 0.70  & 0.38 \\
    \midrule
    \midrule
    \bottomrule
    \end{tabular}%
  \label{tab:Port_sorts_VW_HOR}%
\end{table}%

\newpage

% Table generated by Excel2LaTeX from sheet 'P_sortshor_CAPM_NET'
\begin{table}[!hp]
  \centering
  \caption{\textbf{Contemporaneous characteristics of equal-weighted portfolios}\\
\small{Notes: We create equal-weighted portfolios by sorting stocks into quintiles based on daily realized horizon specific network risk betas; short-term in Panel A and long-term in Panel B. Betas are from daily 3-year rolling regressions and the sample spans July 10, 2009 to August 31, 2018. Portfolios are rebalanced monthly. The portfolio characteristics are: average annual returns, $R_{p}$; annualized \cite{fama2015five} five-factor alphas, Ann. $\alpha^{\text{FF}}$; and average annual betas with respect to dynamic horizon specific network risk, the \cite{fama2015five} five-factors, the VIX, Momentum, conditional skewness, conditional kurtosis and market illiquidity. The latter are obtained from rolling regressions using a 1-year window that add one-month at a time.}  
  }
    \begin{tabular}{lrrrrrrrrrrrrr}
    \toprule
    \midrule
    \midrule
    \multicolumn{14}{l}{\textbf{A: Sorts on NET($S$)}} \\
    Portfolio & \multicolumn{1}{l}{$R_{p}$} & \multicolumn{1}{l}{Ann. $\alpha^{\text{FF}}$} & \multicolumn{1}{l}{$\beta^{\text{NET}(S)}$} &  \multicolumn{1}{l}{$\beta^{\text{MKT}}$} & \multicolumn{1}{l}{$\beta^{\text{SMB}}$} & \multicolumn{1}{l}{$\beta^{\text{HML}}$} & \multicolumn{1}{l}{$\beta^{\text{RMW}}$} & \multicolumn{1}{l}{$\beta^{\text{CMA}}$} & \multicolumn{1}{l}{$\beta^{\text{VIX}}$} & \multicolumn{1}{l}{$\beta^{\text{MOM}}$} & \multicolumn{1}{l}{$\beta^{\text{CSKEW}}$} & \multicolumn{1}{l}{$\beta^{\text{CKURT}}$} & \multicolumn{1}{l}{$\beta^{\text{ILLIQ}}$} \\
    \midrule
    1 Low $\beta^{\text{NET}(S)}$ & 6.05\% & -1.38\% & 0.192 & 0.554 & -0.023 & -0.285 & 0.043 & 0.142 & 0.003 & 0.058 & -0.026 & -0.154 & -0.353 \\
    2     & 9.20\% & 1.37\% & 0.457 & 0.557 & 0.035 & -0.251 & 0.133 & 0.150 & 0.001 & 0.081 & -0.016 & -0.077 & 0.116 \\
    3     & 7.76\% & -0.80\% & 0.618 & 0.569 & 0.113 & -0.249 & 0.080 & 0.175 & 0.003 & 0.080 & -0.055 & -0.018 & 0.769 \\
    4     & 4.94\% & -4.21\% & 0.914 & 0.607 & 0.203 & -0.113 & 0.040 & 0.085 & 0.005 & 0.070 & -0.123 & 0.021 & 1.457 \\
    5 High $\beta^{\text{NET}(S)}$ & -4.21\% & -11.98\% & 1.631 & 0.651 & 0.255 & -0.125 & -0.220 & 0.277 & 0.007 & 0.041 & -0.303 & -0.149 & 4.131 \\
    5--1 & -10.27\% & -10.75\% & 1.439 & 0.097 & 0.279 & 0.160 & -0.263 & 0.134 & 0.004 & -0.017 & -0.277 & 0.005 & 4.484 \\
    $t$-stat & -7.45 & -11.80 & 10.01 & 1.25  & 4.18  & 1.83  & -1.91 & 1.07  & 0.53  & -0.68 & -2.04 & -0.02 & 0.65 \\
    \midrule
    \midrule
    \multicolumn{14}{l}{\textbf{B: Sorts on NET($L$)}} \\
    Portfolio & \multicolumn{1}{l}{$R_{p}$} & \multicolumn{1}{l}{Ann. $\alpha^{\text{FF}}$} & \multicolumn{1}{l}{$\beta^{\text{NET}(L)}$} & \multicolumn{1}{l}{$\beta^{\text{MKT}}$} & \multicolumn{1}{l}{$\beta^{\text{SMB}}$} & \multicolumn{1}{l}{$\beta^{\text{HML}}$} & \multicolumn{1}{l}{$\beta^{\text{RMW}}$} & \multicolumn{1}{l}{$\beta^{\text{CMA}}$} & \multicolumn{1}{l}{$\beta^{\text{VIX}}$} & \multicolumn{1}{l}{$\beta^{\text{MOM}}$} & \multicolumn{1}{l}{$\beta^{\text{CSKEW}}$} & \multicolumn{1}{l}{$\beta^{\text{CKURT}}$} & \multicolumn{1}{l}{$\beta^{\text{ILLIQ}}$} \\
    \midrule
    1 Low $\beta^{\text{NET}(L)}$ & 6.04\% & -1.34\% & 0.205 & 0.555 & -0.035 & -0.278 & 0.055 & 0.152 & 0.003 & 0.056 & -0.019 & -0.151 & -0.414 \\
    2     & 8.86\% & 1.01\% & 0.449 & 0.554 & 0.036 & -0.236 & 0.131 & 0.152 & 0.001 & 0.097 & -0.012 & -0.076 & 0.120 \\
    3     & 8.10\% & -0.41\% & 0.638 & 0.566 & 0.106 & -0.240 & 0.098 & 0.164 & 0.002 & 0.071 & -0.038 & 0.006 & 1.087 \\
    4     & 5.48\% & -3.79\% & 0.939 & 0.600 & 0.197 & -0.105 & 0.076 & 0.111 & 0.004 & 0.075 & -0.090 & 0.035 & 0.996 \\
    5 High $\beta^{\text{NET}(L)}$ & -4.74\% & -12.44\% & 1.594 & 0.656 & 0.248 & -0.078 & -0.162 & 0.257 & 0.006 & 0.011 & -0.301 & -0.154 & 3.409 \\
    5--1 & -10.78\% & -11.25\% & 1.389 & 0.100 & 0.283 & 0.200 & -0.218 & 0.105 & 0.004 & -0.045 & -0.282 & -0.003 & 3.824 \\
    $t$-stat & -7.73 & -12.59 & 9.68  & 1.28  & 3.91  & 2.21  & -1.58 & 0.81  & 0.51  & -0.91 & -2.00 & -0.11 & 0.41 \\
    \midrule
    \midrule
    \bottomrule
    \end{tabular}%
  \label{tab:Port_sorts_EW_HOR}%
\end{table}%

\newpage

\begin{table}[!hp]
  \centering
  \caption{\textbf{Contemporaneous characteristics of value- and equal-weighted portfolios}\\
\small{Notes: We create value-weighted portfolios by sorting stocks into quintiles based on daily realized aggregate network risk betas; value-weighted portfolios are in Panel A and equal-weighted portfolios are in Panel B. Betas are from daily 3-year rolling regressions and the sample spans July 10, 2009 to August 31, 2018. Portfolios are rebalanced monthly. The portfolio characteristics are: average annual returns, $R_{p}$; annualized \cite{fama2015five} five-factor alphas, Ann. $\alpha^{\text{FF}}$; and average annual betas with respect to aggregate directional network risk, the \cite{fama2015five} five-factors, the VIX, Momentum, conditional skewness, conditional kurtosis and market illiquidity. The latter are obtained from rolling regressions using a 1-year window that add one-month at a time.}  
  }
    \begin{tabular}{lrrrrrrrrrrrrr}
    \toprule
    \midrule
    \midrule
    \multicolumn{14}{l}{\textbf{A: Value Weighted Portfolios: Sorts on NET($A$)}} \\
    Portfolio & \multicolumn{1}{l}{$R_{p}$} & \multicolumn{1}{l}{Ann. $\alpha^{\text{FF}}$} & \multicolumn{1}{l}{$\beta^{\text{NET}(A)}$} &  \multicolumn{1}{l}{$\beta^{\text{MKT}}$} & \multicolumn{1}{l}{$\beta^{\text{SMB}}$} & \multicolumn{1}{l}{$\beta^{\text{HML}}$} & \multicolumn{1}{l}{$\beta^{\text{RMW}}$} & \multicolumn{1}{l}{$\beta^{\text{CMA}}$} & \multicolumn{1}{l}{$\beta^{\text{VIX}}$} & \multicolumn{1}{l}{$\beta^{\text{MOM}}$} & \multicolumn{1}{l}{$\beta^{\text{CSKEW}}$} & \multicolumn{1}{l}{$\beta^{\text{CKURT}}$} & \multicolumn{1}{l}{$\beta^{\text{ILLIQ}}$} \\
    \midrule
    1 Low $\beta^{\text{NET}(A)}$ & 7.55\% & -0.13\% & 0.114 & 0.565 & -0.013 & -0.272 & 0.019 & 0.080 & 0.003 & 0.078 & -0.045 & -0.183 & 0.134 \\
    2     & 11.12\% & 3.10\% & 0.372 & 0.550 & 0.044 & -0.237 & 0.124 & 0.148 & 0.001 & 0.105 & -0.015 & -0.041 & 0.135 \\
    3     & 9.43\% & 0.54\% & 0.566 & 0.585 & 0.129 & -0.222 & 0.099 & 0.127 & 0.003 & 0.085 & -0.070 & 0.014 & 1.869 \\
    4     & 7.81\% & -1.644\% & 0.832 & 0.597 & 0.194 & -0.082 & 0.067 & 0.085 & 0.004 & 0.086 & -0.083 & 0.090 & 1.698 \\
    5 High $\beta^{\text{NET}(A)}$ & 4.85\% & -3.82\% & 1.312 & 0.635 & 0.205 & -0.040 & -0.130 & 0.175 & 0.006 & 0.075 & -0.256 & -0.057 & 3.608 \\
    5--1 & -2.71\% & -3.69\% & 1.198 & 0.070 & 0.218 & 0.232 & -0.149 & 0.095 & 0.003 & -0.003 & -0.211 & 0.125 & 3.474 \\
    $t$-stat & -2.30 & -4.97 & 9.76  & 0.92  & 3.47  & 2.81  & -1.28 & 1.05  & 0.53  & -0.49 & -1.38 & 0.98  & 0.53 \\
    \midrule
    \midrule
        \multicolumn{14}{l}{\textbf{B: Equal Weighted Portfolios: Sorts on NET($A$)}} \\
    Portfolio & \multicolumn{1}{l}{$R_{p}$} & \multicolumn{1}{l}{Ann. $\alpha^{\text{FF}}$} & \multicolumn{1}{l}{$\beta^{\text{NET}(A)}$} &  \multicolumn{1}{l}{$\beta^{\text{MKT}}$} & \multicolumn{1}{l}{$\beta^{\text{SMB}}$} & \multicolumn{1}{l}{$\beta^{\text{HML}}$} & \multicolumn{1}{l}{$\beta^{\text{RMW}}$} & \multicolumn{1}{l}{$\beta^{\text{CMA}}$} & \multicolumn{1}{l}{$\beta^{\text{VIX}}$} & \multicolumn{1}{l}{$\beta^{\text{MOM}}$} & \multicolumn{1}{l}{$\beta^{\text{CSKEW}}$} & \multicolumn{1}{l}{$\beta^{\text{CKURT}}$} & \multicolumn{1}{l}{$\beta^{\text{ILLIQ}}$} \\
    \midrule
    1 Low $\beta^{\text{NET}(A)}$ & 5.88\% & -1.46\% & 0.230 & 0.550 & -0.036 & -0.276 & 0.052 & 0.145 & 0.002 & 0.062 & -0.003 & -0.154 & -0.486 \\
    2     & 9.64\% & 1.75\% & 0.454 & 0.548 & 0.037 & -0.249 & 0.129 & 0.153 & 0.001 & 0.098 & -0.008 & -0.056 & -0.071 \\
    3     & 7.81\% & -0.76\% & 0.660 & 0.570 & 0.116 & -0.234 & 0.098 & 0.133 & 0.003 & 0.080 & -0.040 & 0.012 & 1.005 \\
    4     & 4.81\% & -4.41\% & 0.961 & 0.598 & 0.207 & -0.112 & 0.058 & 0.088 & 0.004 & 0.082 & -0.086 & 0.060 & 0.970 \\
    5 High $\beta^{\text{NET}(A)}$ & -4.40\% & -12.09\% & 1.640 & 0.647 & 0.256 & -0.090 & -0.171 & 0.215 & 0.007 & 0.030 & -0.264 & -0.105 & 3.562 \\
    5--1 & -10.28\% & -10.78\% & 1.410 & 0.097 & 0.292 & 0.186 & -0.223 & 0.070 & 0.004 & -0.032 & -0.261 & 0.050 & 4.048 \\
    $t$-stat & -7.24 & -11.88 & 10.11 & 1.10  & 4.17  & 2.04  & -1.56 & 0.63  & 0.58  & -0.81 & -1.76 & 0.22  & 0.54 \\
    \midrule
    \midrule
    \bottomrule
        \end{tabular}%
  \label{tab:Port_sorts_TOT}%
\end{table}%

\end{landscape}

% NEED TO SUMMARISE HERE AND DRAW ON NETWORK LITERATURE AND ALSO THE VOLATILITY LITERATURE WANT TO SYNTHESIZE THIS WITH INTRO ETC AND INTERTEMPORAL HEDGING STORY THAT WE ARE RELYING ON TO EXPLAIN THE RESULTS.
In summary, we provide evidence on the pricing of dynamic horizon specific network risk stemming from connections among asset return volatilities. Fama-MacBeth regressions and portfolio sorts confirm that stocks with a higher sensitivity to horizon specific network risk earn lower returns. Our results suggest that accounting for only short-term or long-term directional network risk result in similar economic magnitudes and statistical significance. We also demonstrate this result holds for aggregate dynamic network risk.  Consistent with economic theory and empirical evidence \citep{herskovic2016common,cremers2015aggregate,campbell2018intertemporal}, as well as our outline of an economy that generates risk pricing of directional volatility connections, stocks loading positively on horizon specific network risk earn lower returns as investors seek to hedge against changes in investment opportunities.

\section{Extensions}\label{extensions}
Here we outline and describe a variety of extensions to our main results. We first account for additional risk proxies that link with volatility and higher moments of stock returns. We then define an alternative dynamic horizon specific network risk and examine the robustness of our results for Fama-MacBeth regressions and portfolio sorts. Finally, we investigate whether we are able to implement an ex-ante strategy to construct hedge portfolios; thus investigating the ability to predict future network risk.

\subsection{Is network risk the same as volatility risk?}
We investigate whether the contemporaneous relations among dynamic horizon specific network risk and returns remains when controlling for additional risk proxies. In particular, we consider: the variance risk premium \citep{bollerslev2009expected}; tail risk \citep{kelly2014tail}; idiosyncratic volatility \citep{ang2006downside}; and idiosyncratic skewness \citep{boyer2010expected}.

We define the variance risk premium, VRP, as the difference between the VIX and S\&P500 realized variance; the latter we construct from our own realized variance measures where weights relate to each stock's daily market capitalisation. We proxy tail risk, SKIND, using daily percent change in the CBOE's skewness index. Idiosyncratic volatility, Idio. Vol, and idiosyncratic skewness, Idio. Skew, are relative to the \cite{fama2015five} five-factor model. Specifically, they are the respective sample standard deviation and skewness of the residuals from daily rolling regressions using a 1-year window. 

% Table generated by Excel2LaTeX from sheet 'Full_Sample'
\begin{table}[!hp]
  \centering
  \caption{\textbf{Fama-MacBeth regressions: Additional factors}
\small{Notes: This table reports Full Sample Fama-MacBeth regressions in Panel A, and Rolling Fama-MacBeth regressions in Panel B, for daily S\&P500 stock returns. The full sample regressions span from July 5, 2005 to August 31, 2018 and use Newey West Standard errors with 12-lags. Rolling regressions use a 3-year window and add 1 day at a time and use Newey West Standard errors with 24-lags.  $t$-ratios in square brackets below coefficient estimates are adjusted following \cite{shanken1992estimation}. NET($S$) (NET($S^\perp$)) is the (orthogonalized) short-term directional network risk factor and NET($L$) is the long-term directional network risk factor. NET($A$) is the aggregate directional network risk factor that sums over short-term and long-term frequency bands. VIX is the daily change in the VIX index; MOM is the momentum factor; CSKEW and CKURT are conditional skewness and conditional kurtosis factors respectively; VRP is the market variance risk premium; SKIND is the daily change in the CBOE's skewness index; Idio. Vol is idiosyncratic volatility; and Idio Skew is idiosyncratic skewness. We refrain from reporting coefficients associated to the \cite{fama2015five} five-factor model; however each model controls for these.} 
} 
    \begin{tabular}{lrrrrrrrr}
    \toprule
    \midrule
    \midrule
          & \multicolumn{4}{c}{\textbf{A: Full Sample}} & \multicolumn{4}{c}{\textbf{B: Rolling}} \\
\cmidrule(lr){2-5} \cmidrule(lr){6-9}
          &   \multicolumn{1}{c}{1} &  \multicolumn{1}{c}{2}   &   \multicolumn{1}{c}{3}   &   \multicolumn{1}{c}{4}   &   \multicolumn{1}{c}{5}   &   \multicolumn{1}{c}{6}   &   \multicolumn{1}{c}{7}   &  \multicolumn{1}{c}{8} \\
    \midrule
NET($S$)/NET($S^{\perp}$) & -0.063 &       & -0.016 &       & -0.036 &       & 0.007 &  \\
          & [-6.03] &       & [-2.10] &       & [-2.30] &       & [0.80] &  \\
NET($L$) &       & -0.062 & -0.054 &       &       & -0.041 & -0.046 &  \\
          &       & [-5.47] & [-4.33] &       &       & [-2.47] & [-2.62] &  \\
NET($A$) &       &       &       & -0.065 &       &       &       & -0.04 \\
          &       &       &       & [-5.75] &       &       &       & [-2.39] \\
VIX & -0.267 & -0.267 & -0.244 & -0.267 & 0.036 & 0.007 & 0.007 & 0.008 \\
          & [-0.63] & [-0.64] & [-0.59] & [-0.65] & [0.06] & [-0.01] & [-0.00] & [-0.01] \\
MOM & 0.124 & 0.128 & 0.128 & 0.124 & 0.092 & 0.085 & 0.085 & 0.088 \\
          & [3.30] & [3.38] & [3.39] & [3.26] & [2.15] & [2.03] & [2.08] & [2.09] \\
CSKEW & 0.034 & 0.042 & 0.029 & 0.038 & 0.039 & 0.038 & 0.039 & 0.038 \\
          & [1.25] & [1.54] & [1.10] & [1.42] & [0.53] & [0.49] & [0.51] & [0.51] \\
CKURT & 0.006 & 0.005 & 0.007 & 0.006 & 0.001 & 0.001 & 0.003 & 0.001 \\
          & [0.27] & [0.23] & [0.34] & [0.27] & [0.58] & [0.56] & [0.62] & [0.56] \\
VRP & 0.013 & 0.012 & 0.013 & 0.013 & 0.008 & 0.008 & 0.008 & 0.008 \\
          & [2.42] & [2.31] & [2.51] & [2.35] & [0.93] & [0.92] & [0.88] & [0.94] \\
SKIND & 0.462 & 0.444 & 0.463 & 0.451 & 0.146 & 0.146 & 0.14  & 0.151 \\
          & [2.01] & [1.94] & [2.01] & [1.96] & [1.03] & [1.04] & [1.01] & [1.07] \\
Idio. Vol & 0.185 & 0.183 & 0.184 & 0.183 & 0.031 & 0.031 & 0.031 & 0.031 \\
          & [5.20] & [5.17] & [5.17] & [5.17] & [1.32] & [1.31] & [1.31] & [1.30] \\
Idio. Skew & -0.061 & -0.065 & -0.064 & -0.064 & -0.016 & -0.015 & -0.015 & -0.015 \\
          & [-2.76] & [-2.86] & [-2.84] & [-2.81] & [-1.32] & [-1.27] & [-1.26] & [-1.28] \\
Intercept & -0.009 & -0.007 & -0.009 & -0.007 & 0.032 & 0.032 & 0.033 & 0.031 \\
          & [-0.52] & [-0.45] & [-0.52] & [-0.42] & [1.31] & [1.31] & [1.37] & [1.29] \\
    \midrule
    \midrule
    \bottomrule
    \end{tabular}%
  \label{tab:FMB_ADDIT}%
\end{table}%

Table \ref{tab:FMB_ADDIT} presents Fama-MacBeth regressions that account for these additional four factors; as well as the \cite{fama2015five} five-factors, the daily percent change in the VIX index, and the MOM, CSKEW and CKURT factors respectively. We remove ILLIQ due to the significant positive correlation with VRP. Panel A reports full sample estimates and Panel B shows results from rolling regressions. We refrain from reporting risk prices of the \cite{fama2015five} five-factors for brevity; however note they are similar to our baseline results.

% First considering full sample estimates in Panel A, we can see that horizon specific directional network risk is statistically significant across all models. Note also that the estimates of risk prices are qunatitatively similar to those in Tables \ref{tab:FMB_horz} and \ref{tab:FMB_T_FS_ROLL}. Note that VRP, SKIND, Idio. Vol, and Idio. Skew are all statistically significant with the $t$-statistics for idiosyncratic volatility risk prices all exceeding 5.  Turning to Panel B, a similar story emerges. The prices of horizon specific directional network risk from rolling regressions are similar to our baseline results from a statistical and economic perspective. The statistical significance of our additional factors disappears in rolling regressions. 

Overall, Table \ref{tab:FMB_ADDIT} suggests that the pricing of dynamic horizon specific network risk is remarkably robust. This is clear from the quantitatively similar risk price estimates for both full-sample and rolling specifications. Apart from the risk price estimate of orthogonalized short-term directional network risk in model 7, all \cite{shanken1992estimation} $t$-statistics remain significant; which is also consistent with our baseline analysis. The main takeaway point from Table \ref{tab:FMB_ADDIT} is that the reward for bearing horizon specific, and indeed aggregate, directional network risk is stable and always negative, as well as statistically and economically significant.

\subsection{Alternative proxies of dynamic horizon specific network risk}

We explore whether the contemporaneous relation between dynamic horizon specific network risk is robust to an alternative definition. This alternative specification takes horizon specific directional connections and size, but re-formulates our original definition. Specifcally, we first sort all S\&P500 constituents with respect to size as before. However, instead of taking an equal weighted average of \textsc{to} and \textsc{from} portfolios that are constructed from assets in the top and bottom 30\% percentiles of the daily distribution of net-directional connections, we now take the difference between \textsc{to} and \textsc{from} portfolios. Formally, the day $t$ network risk factor $\text{NET}'(d)_{t}$, at horizon $d=\{\text{Short},\:\text{Long},\: \text{Aggregate}\}$ is given by

\begin{equation}\label{eq:factor2}
\text{NET}'(d)_{t} = \left(\textsc{to}^{\text{small}}(d)_{t}-\textsc{from}^{\text{small}}(d)_{t}\right) - \left(\textsc{to}^{\text{big}}(d)_{t}-\textsc{from}^{\text{big}}(d)_{t}\right)
\end{equation}

We can see from Table \ref{tab:sort1} that in general \textsc{to} portfolios earn lower returns than \textsc{from} portfolios and that overall portfolios of smaller stocks earn lower returns than those larger stocks. As with our original definition, the dynamic horizon specific network risk factor summarizes all information regarding directional connections. Instead of imposing equal importance on shock transmission and reception after controlling for size, these dynamic network risk proxies place a higher weight on shock transmission after accounting for overall shock reception of the network.

Table \ref{tab:FMB_ALT} shows full sample and rolling regression results for Fama-MacBeth regressions using our alternative specifications of horizon specific network risk in Panels A and B respectively. For the sake of brevity, we report results using short-term and long-term directional network risk in the same regression, where short-term directional network risk is orthogonal to long-term (i.e. NET$'$($S^{\perp}$) and NET$'$($L$)), and aggregate directional network risk, NET$'$(T); results using short-term and long-term directional network risk in isolation to one another are available on request. As we can see, estimates for market risk prices from full sample and rolling regressions are similar to those in Tables \ref{tab:FMB_horz} and \ref{tab:FMB_T_FS_ROLL}. Note that the statistical significance drops slightly for rolling regressions and is significant at 10\% levels when controlling for only \cite{fama2015five} five-factors.

Moving on to portfolio sorts, we conduct the exact same exercise as in Section \ref{port_sorts} for our alternative factor specification. Table \ref{tab:Port_sorts_alt} shows contemporaneous characteristics for quintile portfolios, and a long-short hedge portfolio that sort on aggregate directional network risk, $\beta^{\text{NET}'(T)}$; again results for sorts on short-term and long-term directional network risk betas are available on request. Overall, the same message appears here relative to Tables \ref{tab:Port_sorts_VW_HOR}, \ref{tab:Port_sorts_EW_HOR}, and \ref{tab:Port_sorts_TOT}. Specifically, we observe a monotonic increase in portfolio sensitivities to aggregate directional network risk after controlling for our battery of other factors. Note also that the betas of the hedge portfolios are statistically significant. 

There are two main differences with our main results that may stem from the definition of our alternative network risk factors. First, the implied annual market price of risk from value-weighted portfolios is consistent with the average annual return of the risk factor itself; namely the annual market price of risk is -13.35\% (-4.34\%/0.325) and the average annual return of this alternative factor is -15.91\%. Second, we do not observe a monotonic link between average raw returns and risk adjusted returns. What we do still see are statistically significant and economically meaningful returns for the hedge portfolios. This, combined with the monotonic pattern for directional network risk betas suggests it is important to control for these risk factors; which is further justified by the Fama-MacBeth regressions in Table \ref{tab:FMB_ALT}\footnote{We also conduct our analysis for another alternative specification that ignores relative sizes of assets within the network. This simply takes the difference between \textsc{to} and \textsc{from} portfolios. These results conform to those we present here, albeit slightly weaker in terms of economic and statistical significance. These are available on request.}.

% Table generated by Excel2LaTeX from sheet 'Full_Sample'
\begin{table}[!hp]
  \centering
  \caption{\textbf{Fama-MacBeth regressions: Alternative directional network risk proxies}
\small{Notes: This table reports Full Sample Fama-MacBeth regressions in Panel A, and Rolling Fama-MacBeth regressions in Panel B, for daily S\&P500 stock returns. The full sample regressions span from July 5, 2005 to August 31, 2018 and use Newey West Standard errors with 12-lags. Rolling regressions use a 3-year window and add 1 day at a time and use Newey West Standard errors with 24-lags.  $t$-ratios in square brackets below coefficient estimates are adjusted following \cite{shanken1992estimation}. NET$'$($S^\perp$) is the orthogonalized short-term directional network risk factor and NET$'$($L$) is the long-term directional network risk factor. NET$'$($A$) is the aggregate directional network risk factor that sums over short-term and long-term frequency bands. VIX is the daily change in the VIX index; MOM is the momentum factor; CSKEW and CKURT are conditional skewness and conditional kurtosis factors respectively; ILLIQ is market illiquidity.} 
} 
    \begin{tabular}{lrrrrrrrr}
    \toprule
    \midrule
    \midrule
          & \multicolumn{4}{c}{\textbf{A: Full Sample}} & \multicolumn{4}{c}{\textbf{B: Rolling}} \\
\cmidrule(lr){2-5} \cmidrule(lr){6-9}
          &   \multicolumn{1}{c}{1} &  \multicolumn{1}{c}{2}   &   \multicolumn{1}{c}{3}   &   \multicolumn{1}{c}{4}   &   \multicolumn{1}{c}{5}   &   \multicolumn{1}{c}{6}   &   \multicolumn{1}{c}{7}   &  \multicolumn{1}{c}{8} \\
    \midrule
    NET$'$($S^{\perp}$) & -0.016 & -0.017 &       &       & 0.009 & 0.005 &       &  \\
          & [-1.20] & [-1.22] &       &       & [-0.31] & [0.02] &       &  \\
    NET$'$($L$) & -0.119 & -0.101 &       &       & -0.069 & -0.047 &       &  \\
          & [-5.37] & [-4.62] &       &       & [-1.88] & [-1.45] &       &  \\
    NET$'$($A$) &       &       & -0.124 & -0.104 &       &       & -0.078 & -0.048 \\
          &       &       & [-5.41] & [-4.79] &       &       & [-2.24] & [-1.52] \\
    MKT   & -0.038 & 0.05  & -0.040 & 0.05  & 0.002 & 0.007 & -0.007 & 0.003 \\
          & [-1.03] & [1.39] & [-1.05] & [1.37] & [0.03] & [0.13] & [-0.08] & [0.07] \\
    SMB   & -0.008 & 0.004 & -0.009 & 0.003 & -0.010 & 0.000 & -0.013 & -0.003 \\
          & [-0.48] & [0.27] & [-0.55] & [0.18] & [-0.47] & [-0.09] & [-0.55] & [-0.18] \\
    HML   & -0.009 & -0.009 & -0.010 & -0.01 & -0.027 & -0.024 & -0.026 & -0.024 \\
          & [-0.60] & [-0.64] & [-0.64] & [-0.71] & [-1.05] & [-0.98] & [-1.04] & [-1.00] \\
    RMW   & 0.035 & 0.015 & 0.035 & 0.015 & 0.032 & 0.016 & 0.033 & 0.014 \\
          & [2.30] & [1.10] & [2.27] & [1.08] & [1.19] & [0.68] & [1.24] & [0.60] \\
    CMA   & -0.031 & -0.016 & -0.030 & -0.015 & 0.002 & -0.005 & 0.006 & -0.002 \\
          & [-2.61] & [-1.61] & [-2.53] & [-1.56] & [0.18] & [-0.30] & [0.40] & [-0.14] \\
    VIX   &       & -0.037 &       & -0.044 &       & 0.004 &       & 0.05 \\
          &       & [-0.09] &       & [-0.11] &       & [-0.01] &       & [0.11] \\
    MOM   &       & 0.153 &       & 0.156 &       & 0.094 &       & 0.094 \\
          &       & [4.40] &       & [4.46] &       & [2.24] &       & [2.24] \\
    CSKEW  &       & 0.047 &       & 0.049 &       & 0.029 &       & 0.035 \\
          &       & [1.86] &       & [1.94] &       & [0.47] &       & [0.57] \\
    CKURT  &       & -0.005 &       & -0.004 &       & 0.012 &       & 0.009 \\
          &       & [-0.21] &       & [-0.17] &       & [0.94] &       & [0.83] \\
    ILLIQ &       & 0.007 &       & 0.007 &       & 0.007 &       & 0.006 \\
          &       & [2.77] &       & [2.74] &       & [1.15] &       & [1.12] \\
    Intercept & 0.053 & -0.002 & 0.053 & -0.003 & 0.023 & 0.028 & 0.024 & 0.029 \\
          & [3.12] & [-0.11] & [3.08] & [-0.17] & [0.87] & [1.21] & [0.91] & [1.24] \\
%\midrule
%    $\bar{R}^{2}$ & 0.31  & 0.58  & 0.31  & 0.58  & 0.33  & 0.48  & 0.32  & 0.47 \\
    \midrule
    \midrule
    \bottomrule
    \end{tabular}%
  \label{tab:FMB_ALT}%
\end{table}%

\begin{landscape}    
% Table generated by Excel2LaTeX from sheet 'Portfolio_Sorts_Total'
\begin{table}[!hp]
  \centering
  \caption{\textbf{Contemporaneous characteristics of value- and equal-weighted portfolios}\\
\small{Notes: We create value-weighted portfolios by sorting stocks into quintiles based on daily realized aggregate network risk betas; value-weighted portfolios are in Panel A and equal-weighted portfolios are in Panel B. Betas are from daily 3-year rolling regressions and the sample spans July 10, 2009 to August 31, 2018. Portfolios are rebalanced monthly. The portfolio characteristics are: average annual returns, $R_{p}$; annualized \cite{fama2015five} five-factor alphas, Ann. $\alpha^{\text{FF}}$; and average annual betas with respect to dynamic horizon specific network risk, the \cite{fama2015five} five-factors, the VIX, Momentum, conditional skewness, conditional kurtosis and market illiquidity. The latter are obtained from rolling regressions using a 1-year window that add one-month at a time.}  
  }
    \begin{tabular}{lrrrrrrrrrrrrr}
    \toprule
    \multicolumn{14}{l}{\textbf{A: Value Weighted Portfolios: Sorts on NET($A$)}} \\
    Portfolio & \multicolumn{1}{l}{$R_{p}$} & \multicolumn{1}{l}{Ann. $\alpha^{\text{FF}}$} & \multicolumn{1}{l}{$\beta^{\text{NET}'(A)}$} &  \multicolumn{1}{l}{$\beta^{\text{MKT}}$} & \multicolumn{1}{l}{$\beta^{\text{SMB}}$} & \multicolumn{1}{l}{$\beta^{\text{HML}}$} & \multicolumn{1}{l}{$\beta^{\text{RMW}}$} & \multicolumn{1}{l}{$\beta^{\text{CMA}}$} & \multicolumn{1}{l}{$\beta^{\text{VIX}}$} & \multicolumn{1}{l}{$\beta^{\text{MOM}}$} & \multicolumn{1}{l}{$\beta^{\text{CSKEW}}$} & \multicolumn{1}{l}{$\beta^{\text{CKURT}}$} & \multicolumn{1}{l}{$\beta^{\text{ILLIQ}}$} \\
    \midrule
    1 Low $\beta^{\text{NET}'(A)}$ & 8.94\%  & -0.07\% & -0.241 & 0.599 & 0.176 & -0.146 & 0.002 & 0.064 & 0.004 & -0.010 & -0.235 & -0.168 & 2.243 \\
    2     & 8.58\%  & -0.13\% & -0.137 & 0.594 & 0.123 & -0.181 & 0.073 & 0.190 & 0.004 & 0.005 & -0.132 & -0.033 & 2.606 \\
    3     & 9.45\%  & 0.97\%  & -0.090 & 0.587 & 0.111 & -0.170 & 0.110 & 0.216 & 0.003 & -0.003 & -0.086 & -0.038 & 0.363 \\
    4     & 9.75\%  & 1.75\%  & -0.059 & 0.582 & 0.118 & -0.222 & 0.064 & 0.235 & 0.003 & 0.021 & -0.010 & -0.078 & -1.396 \\
    5 High $\beta^{\text{NET}'(A)}$ & 4.61\%  & -2.88\% & 0.084 & 0.629 & 0.198 & -0.124 & -0.166 & 0.293 & 0.004 & -0.040 & -0.184 & -0.248 & -0.837 \\
    5--1 & -4.34\% & -2.81\% & 0.325 & 0.030 & 0.022 & 0.022 & -0.168 & 0.230 & 0.000 & -0.030 & 0.051 & -0.081 & -3.080 \\
    $t$-stat & -5.54 & -3.61 & 5.34  & 0.89  & 0.81  & 0.79  & -1.59 & 2.21  & 0.09  & -0.50 & 0.48  & -0.41 & -1.01 \\
    \midrule
    \midrule
    \multicolumn{14}{l}{\textbf{B: Equal Weighted Portfolios: Sorts on NET($A$)}} \\
    Portfolio & \multicolumn{1}{l}{$R_{p}$} & \multicolumn{1}{l}{Ann. $\alpha^{\text{FF}}$} & \multicolumn{1}{l}{$\beta^{\text{NET}'(A)}$} &  \multicolumn{1}{l}{$\beta^{\text{MKT}}$} & \multicolumn{1}{l}{$\beta^{\text{SMB}}$} & \multicolumn{1}{l}{$\beta^{\text{HML}}$} & \multicolumn{1}{l}{$\beta^{\text{RMW}}$} & \multicolumn{1}{l}{$\beta^{\text{CMA}}$} & \multicolumn{1}{l}{$\beta^{\text{VIX}}$} & \multicolumn{1}{l}{$\beta^{\text{MOM}}$} & \multicolumn{1}{l}{$\beta^{\text{CSKEW}}$} & \multicolumn{1}{l}{$\beta^{\text{CKURT}}$} & \multicolumn{1}{l}{$\beta^{\text{ILLIQ}}$} \\
    \midrule
    1 Low $\beta^{\text{NET}'(A)}$ & 6.25\%  & -3.13\% & -0.196 & 0.616 & 0.238 & -0.031 & -0.043 & 0.070 & 0.006 & -0.033 & -0.247 & -0.078 & 2.003 \\
    2     & 7.56\%  & -1.28\% & -0.101 & 0.604 & 0.159 & -0.167 & 0.075 & 0.230 & 0.004 & -0.026 & -0.139 & -0.051 & 1.366 \\
    3     & 7.03\%  & -1.53\% & -0.065 & 0.593 & 0.143 & -0.156 & 0.116 & 0.242 & 0.003 & -0.028 & -0.096 & -0.068 & 0.099 \\
    4     & 7.06\%  & -0.87\% & -0.032 & 0.604 & 0.167 & -0.193 & 0.071 & 0.255 & 0.004 & -0.027 & -0.070 & -0.132 & -1.172 \\
    5 High $\beta^{\text{NET}(A)}$ & -4.15\% & -10.42\% & 0.193 & 0.674 & 0.278 & -0.063 & -0.217 & 0.360 & 0.006 & -0.144 & -0.282 & -0.325 & -0.135 \\
    5--1 & -10.40\% & -7.52\% & 0.389 & 0.058 & 0.041 & -0.032 & -0.175 & 0.291 & 0.000 & -0.111 & -0.035 & -0.247 & -2.138 \\
    $t$-stat & -10.69 & -12.28 & 6.31  & 1.27  & 0.98  & -0.33 & -1.41 & 2.43  & 0.11  & -2.25 & -0.38 & -2.15 & -0.72 \\
    \midrule
    \midrule
    \bottomrule
    \end{tabular}%
  \label{tab:Port_sorts_alt}%
\end{table}%    
  \end{landscape}

\subsection{Predicting future network risk}

% Table generated by Excel2LaTeX from sheet 'Portfolio_sorts_OOS'
\begin{table}[!hp]
  \centering
  \caption{\textbf{Average annual portfolio returns sorted by past horizon specific and aggregate directional network risk betas}\\
  \small{Notes: This table presents average annual returns of portfolios sorted on past estimates of horizon specific and aggregate directional network risk betas. Specifically, on day $t$ we sort stocks according to their network risk betas on day $t-1$. We rebalance quintile portfolios every month. The table presents average annual returns of value-weighted and equal-weighted portfolios. We also present the annualized \cite{fama2015five} five-factor alpha from rolling regressions using a 1-year window and moving forward 1-month at a time, Ann. $\alpha^{\text{FF}}$ of the 5--1 portfolios. $t$-statistics associated with the annualized \cite{fama2015five} alpha use Newey West standard errors with 24 lags and are adjusted according to \cite{shanken1992estimation}.}
  }
    \begin{tabular}{lcccccc}
    \toprule
    \midrule
\midrule
    \multicolumn{7}{c}{\textbf{Average Annual Returns}} \\
    \midrule
          & \multicolumn{2}{c}{$\beta^{\text{NET}(S)}$} & \multicolumn{2}{c}{$\beta^{\text{NET}(L)}$} & \multicolumn{2}{c}{$\beta^{\text{NET}(A)}$} \\
                    \cmidrule(lr){2-3} \cmidrule(lr){4-5} \cmidrule(lr){6-7} 
          & \multicolumn{1}{c}{\textbf{VW}} & \multicolumn{1}{c}{\textbf{EW}} & \multicolumn{1}{l}{\textbf{VW}} & \multicolumn{1}{l}{\textbf{EW}} & \multicolumn{1}{l}{\textbf{VW}} & \multicolumn{1}{l}{\textbf{EW}} \\
    \midrule
    1 low $\beta$ & 7.71\%  & 6.11\%  & 7.53\%  & 6.02\%  & 7.52\%  & 5.88\% \\
    2     & 10.94\% & 9.26\%  & 10.57\% & 8.96\%  & 11.25\% & 9.74\% \\
    3     & 9.62\%  & 8.01\%  & 10.04\% & 8.27\% & 9.59\%  & 7.99\% \\
    4     & 7.58\%  & 4.93\%  & 8.35\%  & 5.61\%  & 7.90\%  & 4.93\% \\
    5 high $\beta$ & 5.44\%  & -3.97\% & 4.78\%  & -4.52\% & 5.04\%  & -4.20\% \\
    5--1   & -2.27\% & -10.08\% & -2.76\% & -10.54\% & -2.47\% & -10.08\% \\
    $t$-stat & -1.93 & -7.23 & -2.33 & -7.44 & -2.09 & -7.02 \\
\midrule
\midrule
     Ann. $\alpha^{\text{FF}}$ 5--1  & -3.41\% & -11.03\% & -4.00\% & -11.45\% & -3.66\% & -11.00\% \\
    $t$-stat & -2.11 & -7.62 & -3.23 & -8.32 & -3.05 & -7.88 \\
    \midrule
\midrule
    \bottomrule
    \end{tabular}%
  \label{tab:port_sort_OOS}%
\end{table}%

Our main focus for this paper is the contemporaneous relationship between dynamic horizon specific network risk and stock returns. Although we find strong evidence in favour of significantly priced directional network risk, contemporaneous analysis tells us nothing with regards the practical use of tracking these network linkages. Therefore, it is necessary to examine whether, in real-time, investors are able to construct hedge portfolios against such directional network risk. 

We now conduct portfolio sorts using information on past directional network risk sensitivities. Specifically, on day $t$ we sort stocks into quintile portfolios based on $\beta$ estimates on day $t-1$. We rebalance these portfolios monthly as in our baseline analysis throughout each year with the starting day of July 11, 2009 so we can use the beta estimates on July 10, 2009. Thus, we follow largely the same procedure as in our contemporaneous portfolio sorts. Table \ref{tab:port_sort_OOS} shows average annual returns for portfolios sorted on past $\beta$ estimates of short-term, long-term, and aggregate directional network risk sensitivities; as well as the risk-adjusted returns from the long-short hedge portfolios.

Across both value-weighted and equal-weighted portfolios, we find consistent results with our contemporaneous analysis. We can see there is a monotonically decreasing link in average returns from quintiles 2--5; with quintile 1 earning lower returns than quintile 2 in all specifications. Also note that the 5--1 hedge portfolios exhibit returns of a similar magnitude to our main results. A similar story emerges for the \cite{fama2015five} five-factor alphas. Note also that the 5--1 hedge portfolio raw and risk adjusted returns are in general statistically significant. 

\section{Conclusion}\label{conclusion}

In this paper we investigate the pricing of short-term and long-term dynamic network risk in the cross-section of stock returns. Stocks with a high sensitivity to dynamic network risk earn lower returns. We rationalize our finding with economy theory by outlining an economy that allows the stochastic discount factor to load on network connections among the idiosyncratic volatilities of stock returns through the precautionary savings channel.

Our results show that dynamic horizon specific network risk is priced in the cross-section of stock returns and is economically meaningful. In particular, a one-standard-deviation rise across stocks in short-term and long-term directional network risk factor loadings implies a fall in expected annual returns of 6.71\% and 7.66\% respectively. We also show that these patterns hold when considering aggregate dynamic network risk that accounts for all horizons. Our results are empirically robust to alternative specifications of dynamic network risk and controlling for a battery of factors over and above the five-factors of \cite{fama2015five}.

There are various important implications that emerge from our study. First, directional network connections in asset return volatility are priced in the cross-section that are economically significant. Second, these horizon specific connections, and indeed the pricing of this risk, varies substantially throughout time. Nevertheless, we show in real time that one can predict future dynamic horizon specific network risk and therefore investors can implement strategies to hedge against these exposures. Finally, our method of decomposing overall network connections among asset return volatilities into horizon specific components permits investors to examine any horizon of interest. This may be useful for investment decisions, diversification, and constructing hedge portfolios against this source of risk.

%{\footnotesize{
%\setlength{\bibsep}{3pt}
\bibliographystyle{chicago}
\bibliography{bibliography}
%}} 

\newpage
\appendix
\section*{Internet Appendix (Not for Publication)}

\section{A detailed outline of the asset pricing model} \label{app:economy}
\numberwithin{equation}{section} \makeatletter % "activate" the preparatory code, but for section-level headers only
\newcommand{\section@cntformat}{Appendix \thesection:\ }
\makeatother
This section outlines an economy that generates horizon specific connections in the volatility of asset returns and consumption growth. We provide this outline as an intuitive general framework that generates a mechanism for horizon specific network risk to price assets in equilibrium. Consider an extension of the two tree asset pricing model in \cite{cochrane2007two} that builds on the work of \cite{lucas1978asset}. Without loss of generality, our endowment economy has $h$ sources of aggregate risk that we interpret as horizon specific dividend streams from endowment cash flows. We adopt the assumptions in \cite{bansal2004risks} and \cite{backus2011disasters} and model returns from risky assets (including dividends) as claims on certain risk factors in the consumption process. To keep things tractable, we present below a model containing $\iota=S,L$ sources of aggregate risk where $S$ corresponds to short-term and $L$ corresponds to long-term risk. Note that accounting for more horizons in $\iota$ results in more tedious algebra. 

The representative investor has the following general utility over the stream of consumption
\begin{align}
U_t = \mathbb{E}_t \int_{0}^{\infty} e^{- \delta \tau} u(c_{t+\tau}) d \tau . 
\end{align}
Each endowment dividend stream follows a geometric Brownian motions with stochastic volatility; whose respective drift and diffusion parameters differ.
\begin{align}\label{eq:consumption}
\frac{dD_\iota}{D_\iota} &= \mu_\iota dt + \sqrt{v_{\iota,t}} d Z_\iota, \quad \iota =\{S,L\}\\
dv_{\iota,t}       &= \kappa_{v_{\iota}}\left(\bar{v}_{\iota} - v_{\iota,t} \right)dt + \sigma_{v_{\iota}}\sqrt{v_{\iota,t}}d Z_{v_{\iota}} + \sum^{N}_{j=1} K_{\iota,j}d\mathcal{N}_{\iota,j,t}
\end{align}

where $dZ_\iota$ are standard Brownian motions that are possibly correlated, $\mathbb{C}\text{orr}\Big(dZ_{S},dZ_{L}\Big)=\rho_{S,L}dt$. We interpret the endowment trees as long-term and short-term risk factors within the economy. The $L$ tree corresponds to the long-term component of consumption. This generates a persistent dividend stream and bears long-term risk in the economy. The $S$ tree generates a less persistent dividend stream bearing short-term risks in the economy; $\mu_L > \mu_{S} > 0$. The diffusion of each tree follows a mean-reverting square root process, similar to the process in \cite{heston1993closed}. $\bar{v}_{\iota}$ is long-run conditional variance and $\kappa_{v_{\iota}}$ captures the speed of mean reversion. We assume the speed of mean reversion is slower for the volatility process associated to the long-term past of consumption.  

To introduce horizon specific network connections, we add $N$ self and mutually exciting jumps $\mathcal{N}_{\iota,j,t}$, $\iota=\{S,L\}$ to the respective square root processes \citep{ait2015modeling}. These horizon specific self and mutually exciting jumps have constant jump sizes $K_{\iota,j}>0,\:K_{\iota,j} \neq K_{\iota,k}$. Their stochastic jump intensities follow Hawkes processes which by definition have the following $\mathcal{F}_{t}$ conditional mean jump rate per unit of time:
\begin{eqnarray*}
\mathbb{P}[N_{\iota,j,t+\Delta} - N_{\iota,j,t} = 0 | \mathcal{F}_{t}] &=& 1-\ell_{\iota,j,t}\Delta + o(\Delta)\\
\mathbb{P}[N_{\iota,j,t+\Delta} - N_{\iota,j,t} = 1 | \mathcal{F}_{t}] &=& \ell_{\iota,j,t}\Delta + o(\Delta)\\
\mathbb{P}[N_{\iota,j,t+\Delta} - N_{\iota,j,t} > 1 | \mathcal{F}_{t}] &=& o(\Delta)\\
\end{eqnarray*}
With dynamics
\begin{align*}
\ell_{\iota,j,t} = \ell_{\iota,j,\infty} + \sum^{N}_{k=1} \int^{t}_{-\infty} g_{\iota,j,k}(t-s)d\mathcal{N}_{\iota,k,s}, \quad k=1,\dots,M.
\end{align*} 
here $\ell_{\iota,j,\infty} \geq 0$ for all $j=1,\dots,N$, the real-valued functions $g_{\iota,j,k}(u)\geq 0$ for all $u\geq0$ and for all $j,k =1,\dots,N$. This ensures the intensity processes are non-negative with probability 1. Note here that $\mathbb{E}\left[d\mathcal{N}_{\iota,j,s}\right] = \ell_{\iota,j,s}ds$ and we have
\begin{align*}
\ell_{\iota,j} = \ell_{\iota,j,\infty} + \sum^{N}_{j=1}\ell_{\iota,j}\int^{t}_{-\infty}g_{\iota,j,k}(t-s)ds\\
           = \ell_{\iota,j,\infty} + \sum^{N}_{j=1}\left(\int^{\infty}_{0} g_{\iota,j,k}(u)du \right) \ell_{\iota,j}
\end{align*}
Under this general form, we see that each prior horizon specific jump raises the corresponding horizon specific intensities $\ell_{\iota,j}$\footnote{In vector form we have $\mathbb{L}_{\iota} = (I-\mathbf{G}_{\iota})^{-1}\mathbb{L}_{\iota,\infty}$ where $\mathbf{G}_{\iota}$ contains elements $\int^{\infty}_{0}g_{\iota,j,k}(u)du$ and the vectors $\mathbb{L}_{\iota},\:\mathbb{L}_{\iota,\infty}$ contain the corresponding $\ell_{\iota,j},\:\ell_{\iota,j,\infty}$ elements. This ensures stationarity of model.}. The distribution of the jump processes $\mathcal{N}_{j,s}$ are determined by that of the intensities. The compensated processes $\mathcal{N}_{\iota,j,t} -\int^{t}_{-\infty}\ell_{\iota,j,s}ds$ are local martingales. The way in which we impose shock propagation, and therefore the network structure is by imposing exponential decay on the $g_{\iota,j,k}(t-s)$ formally:
\begin{eqnarray*}
g_{\iota,j,k}(t-s) = b_{\iota,j,k}e^{-\alpha_{\iota,j}(t-s)}, \quad s<t,\: \iota=S,L,\: j,k=1,\dots,N
\end{eqnarray*}
with $\alpha_{\iota,j}>0,b_{\iota,j,k}>0$ for all $j,k=1,\dots,N$. Therefore using this functional form of $g_{\iota,j,k}(t-s)$ we have the following mean reverting dynamics
\begin{eqnarray}
d\ell_{\iota,j,t} = \alpha_{\iota,j}\left(\ell_{\iota,j,\infty}-\ell_{\iota,j,t} \right)dt + \sum^{N}_{k=1}b_{\iota,j,k}dN_{\iota,k,t}
\end{eqnarray}
This means that a horizon specific jumps in asset $k$ causes an increase in the horizon specific jump intensity of asset $j$ such that $\ell_{\iota,j}$ jumps by $b_{\iota,j,k}$ before decaying back towards the level $\ell_{\iota,j,\infty}$ at speed $\alpha_{\iota,j}$. If the increase in $\ell_{\iota,j}$ leads to a jump in asset $j$, and there is a non-zero $b_{\iota,n,j}$, the horizon specific shock passes on to asset $n$. In this manner the shocks can be propagated throughout the entire network, and also permits the initial shock to reach asset $k$ itself.

The expected value of the jump sizes are $\mathbb{E}[K_{\iota,j}]=\mu_{K_{\iota,j}}$ in general\footnote{If the $j$ jump sizes are constant for each horizon $\iota$, then the expected value is $\mathbb{E}[K_{\iota,j}]=\mu_{K_{\iota}}$. Note we could further simplify and assume that jump sizes are constant across horizons and assets then $\mathbb{E}[K_{\iota,j}]=\mu_{K}$ \citep{branger2018equilibrium}.}.

We have $N$ risky assets in the economy that whose dynamics are geometric brownian motions with stochastic volatility. Formally the $k$-th asset has the following dynamics:
\begin{align}
\frac{dp_{k,t}}{p_{k,t}} &= \mu_{p_{k}}dt + \sqrt{v_{p_{k,S}}}dW_{S} + \sqrt{v_{p_{k,L}}}dW_{L}\\
d v_{p_{k},\iota} &= \kappa_{p_{k},\iota}\left(\bar{v}_{p_{k},\iota} - v_{p_{k},\iota}\right)dt + \sigma_{p_{k},\iota}\sqrt{v_{p_{k},\iota}}dW_{\xi} + Q_{\iota}d\mathcal{N}_{\iota,k,t}, \quad \iota=\{S,L\},\: \xi=\{S_1,L_1\}
\end{align}
which is similar to \cite{christoffersen2009shape} where the variance of the stock return is the sum of the two stochastic volatility components. Note that $\mathbb{C}\text{orr}\Big(dW_{S},dW_{S_1}\Big)=\rho_{W_{S},W_{S_1}}dt$ and $\mathbb{C}\text{orr}\Big(dW_{L},dW_{L_1}\Big)=\rho_{W_{L},W_{L_1}}dt, \: \mathbb{C}\text{orr}\Big(dW_{S},dW_{L_1}\Big)=\mathbb{C}\text{orr}\Big(dW_{L},dW_{S_1}\Big)=0 $. We add discontinuities to each variance process that also enter horizon specific stochastic volatility processes for each consumption dividend stream. $Q_{\iota}$, $\iota=S,L$, are the jump sizes of compound Poisson processes $Q_{L}>Q_{S}>0$ and $\mathcal{N}_{S,k,t},\:\mathcal{N}_{L,k,t}$ are mutually independent Poisson processes. Note their intensity parameters are as in (\ref{eq:intensities_dynamics}).

The economy also contains a risk-less bond, B, that follows an ordinary differential equation:
\begin{align}
\frac{d\text{B}}{\text{B}} = r_{f} dt .
\end{align}
with $r_{f}$ being the instantaneous risk-free rate.

The share of dividend streams within the consumption process determine the state variable for the economy. Thus we have
\begin{align}
s_{t} = \frac{D_S}{D_S + D_L}
\end{align}
the relative size of the first dividend - i.e. short-term part of the consumption process to aggregate consumption. It is important to note here that since consumption volatility varies throughout time, the states, $s_{t},\: (1-s_{t})$ are also time-varying. The function generating consumption is $c_t=s_{t}D_{S}+(1-s_{t})D_{L}$. Applying It\^ o’s lemma, consumptions dynamics are
\begin{align}\label{eq:consumption_dynamics}
\frac{dc_t}{c_t} &= [s_{t}\mu_S + (1 - s_{t})\mu_L ] dt + s_{t}\sqrt{v_{S,t}} dZ_S + (1 - s_{t})\sqrt{v_{L,t}}dZ_L
\end{align}

\begin{prop}
  \label{prop:1}
  Consider a $\iota=\{S,L\}$ tree endowment economy where cash flows for consumption have jump-diffusion dynamics with $N$ self and mutually exciting jumps whose intensities depend on network connections among $N$ risky assets in the economy. Then the first two moments of consumption growth, $ \mathbb{E}_{t}\left[\frac{dc_{t}}{c_{t}} \right],\: \mathbb{E}_{t}\left[\Big(\frac{dc_{t}}{c_{t}}\Big)^{2} \right]$,  and the risk-free rate, $r_{f}$ are given by 
  \begin{eqnarray}
  \label{eq:c_mean_net}
\mathbb{E}_{t}\left[\frac{dc_{t}}{c_{t}} \right] &=& \Big[s_{t}\mu_{S} + (1-s_{t})\mu_{L}\Big]dt\\
\label{eq:c_var_net}
\mathbb{E}_{t}\left[\left(\frac{dc_{t}}{c_{t}}\right)^{2} \right] &=& \Big[s_{t}^{2}v_{S,t} + (1-s_{t})^{2}v_{L,t}+s_{t}(1-s_{t})\sqrt{v_{S,t}}\sqrt{v_{L,t}}\rho_{S,L}\Big]dt  
    \end{eqnarray}  
    \begin{equation}
    \label{eq:rf_net}
    \begin{split}
r_{f} &= \delta  + \gamma_{t}\Big[s_{t}\mu_{S} + (1-s_{t})\mu_{L}\Big] - \frac{1}{2}\eta_{t}\Big[s_{t}^{2}v_{S,t} + (1-s_{t})^{2}v_{L,t}+s_{t}(1-s_{t})\sqrt{v_{S,t}}\sqrt{v_{L,t}}\rho_{S,L}\Big]
\end{split}
    \end{equation} 
with $\delta>0$ being impatience, $\gamma_{t}\equiv -\frac{u''(c_{t})c_{t}}{u'(c_{t})} >0$ is the coefficient of risk aversion, and $\eta_{t}\equiv \frac{u'''(c_{t})(c_{t})^{2}}{u'(c_{t})} >0$ is precautionary saving.   
\end{prop}
\begin{proof}
  See Appendix \ref{app:proofs1}.
\end{proof}

Naturally, one would assume that jumps affecting the long-term consumption claim would be larger than those relating to the short-term consumption claim, $K_{L,j}>K_{S,j}$ with $\mathbb{E}[K_{L,j}]>\mathbb{E}[K_{S,j}]$. Adding to this, one might also assume that the horizon specific intensities would have different mean reversion speeds (i.e. $\alpha_{L,j}<\alpha_{S,j}$) implying that intensities exhibit a relatively slower return back to equilibrium in the long-run. This also follows the notion of jump clustering where the network is more receptive to shocks, thereby increasing intensities, and reinforcing the shock propagation mechanism \citep{ait2015modeling}. It is clear that a jump to either consumption claim reduces consumption growth. If we wish to examine how the self and mutually exciting jumps impact relative consumption shares we can apply It\^{o}'s lemma to $s_{t}$ which yields
\begin{equation}
\begin{split}
ds_{t} &= s_{t}(1-s_{t})\Big(\mu_{S} - \mu_{L} - s_{t}v_{S,t} + (1-s_{t})v_{L,t} + 2\Big(s_{t}-\frac{1}{2}\Big)\sqrt{v_{S,t}}\sqrt{v_{L,t}}\rho_{S,L}\Big)dt \\&+ s_{t}(1-s_{t})\Big[\sqrt{v_{S,t}}dZ_{S} - \sqrt{v_{L,t}}dZ_{L} \Big]
\end{split}
\end{equation}
The drift of state $s_{t}$ is dependent on: i) current states; ii) current conditional variance of each consumption process; and iii) the conditional covariance between the long-term and short-term parts of consumption. Importantly, horizon specific network connections enter the conditional variances and therefore influence consumption shares  indirectly. If there are no jumps, then we arrive back at the \cite{cochrane2007two} model with stochastic volatility. The drift of the long-term consumption claim exhibits what \cite{cochrane2007two} class as \textquotedblleft S-shaped mean reversion". In this case, if $s_{t}=(1-s_{t})=\frac{1}{2}$ dividend shares are most volatile. In our framework the states depend on current levels of variance, $v_{S,t},\:v_{L,t}$. If $s_{t}=\frac{1}{2}$, then the terms multiplied by $2\left(s_{t}-\frac{1}{2}\right)$ are zero. This means that covariance between long-term and short-term consumption claims disappear.  

\section{Proofs}\label{app:proofs1}
\subsection{Proofs: Asset pricing with horizon specific volatility connections}
All proofs for the asset pricing model stem from the first order conditions from the representative agent's utility maximisation problem. This yields an expression for the pricing kernel which we can then combine with the fundamental pricing equation for the $k$-th asset. After applying It\^{o}s lemma, we can derive an expression for the risk-free rate, along with the risk premium of $M$ risky assets in the economy.
\begin{eqnarray*}
p_{k,t}u'(c_{t}) &=& \mathbb{E}_t \int_{0}^{\infty} e^{- \delta \tau} u'(c_{t+\tau})\text{d}_{k,t+\tau} d \tau . \\ 
\Lambda_{t} &\equiv & e^{- \delta \tau} u'(c_{t+\tau}) .  \\
\mathbb{E}_{t}\left[d(p_{k,t}\Lambda_{t})\right] &=& 0 \\
\end{eqnarray*}

\begin{proof}[Proposition~\ref{prop:1}]
From the fundamental pricing equation we have:
\begin{eqnarray*}
r_{f}dt &=& - \mathbb{E}_{t}\left[ \frac{d\Lambda_{t}}{\Lambda_{t}} \right] \\
r_{f}dt &=& - \mathbb{E}_{t}\left[ -\delta dt + \frac{u''(c_{t})c_{t}}{u'(c_{t})}\frac{dc_{t}}{c_{t}} + \frac{1}{2}\frac{u'''(c_{t})c^{2}_{t}}{u'(c_{t})}\frac{(dc_{t})^{2}}{c^{2}_{t}} \right]\\
r_{f}dt &=& \delta dt - \frac{u''(c_{t})c_{t}}{u'(c_{t})}\mathbb{E}_{t}\left[\frac{dc_{t}}{c_{t}} \right] - \frac{1}{2}\frac{u'''(c_{t})c^{2}_{t}}{u'(c_{t})}\mathbb{E}_{t}\left[\frac{(dc_{t})^{2}}{c^{2}_{t}} \right]\\
r_{f}dt &=& \delta dt + \gamma_{t} \mathbb{E}_{t}\left[\frac{dc_{t}}{c_{t}} \right] - \frac{1}{2}\eta_{t} \mathbb{E}_{t}\left[\left(\frac{dc_{t}}{c_{t}}\right)^{2} \right]
\end{eqnarray*}
where $\delta>0$ is impatience, $\gamma_{t}\equiv -\frac{u''(c_{t})c_{t}}{u'(c_{t})} >0$ is the coefficient of risk aversion, and $\eta_{t}\equiv \frac{u'''(c_{t})(c_{t})^{2}}{u'(c_{t})} >0$ is precautionary saving. The moments of consumption are

\begin{eqnarray*}
\mathbb{E}_{t}\left[\frac{dc_{t}}{c_{t}} \right] &=& \left(s_{t}\mu_{S}+ (1-s_{t})\mu_{L}\right)dt
\end{eqnarray*}
\begin{eqnarray*}
\mathbb{E}_{t}\left[\left(\frac{dc_{t}}{c_{t}}\right)^{2} \right] &=& \Big[s_{t}^{2}v_{S,t} + (1-s_{t})^{2}v_{L,t}+s_{t}(1-s_{t})\sqrt{v_{S,t}}\sqrt{v_{L,t}}\rho_{S,L}\Big]dt	
\end{eqnarray*}
 Therefore the risk-free rate is given by
\begin{align*}
\begin{split}
r_{f}dt &= \delta dt + \gamma_{t}\left(s_{t}\mu_{S}+ (1-s_{t})\mu_{L}\right)dt - \frac{1}{2}\eta_{t}\Big[s_{t}^{2}v_{S,t} + (1-s_{t})^{2}v_{L,t}+s_{t}(1-s_{t})\sqrt{v_{S,t}}\sqrt{v_{L,t}}\rho_{S,L}\Big]dt	
\end{split}
\end{align*}
finally, dividing through by $dt$ delivers the risk-free rate. This completes the proof. 
\end{proof}

\begin{proof}[Proposition~\ref{prop:2}] From the fundamental pricing equation and \ref{prop:1}, we can immediately see
\begin{eqnarray*}
\mathbb{E}_{t}\left[ \frac{d\Lambda_{t}}{\Lambda_{t}} \right] &=& - r_{f}dt \\
&=& - \left\{\delta dt + \gamma_{t} \mathbb{E}_{t}\left[\frac{dc_{t}}{c_{t}} \right] - \frac{1}{2}\eta_{t} \mathbb{E}_{t}\left[\left(\frac{dc_{t}}{c_{t}}\right)^{2} \right]\right\}\\
&=& -\delta dt - \gamma_{t} \mathbb{E}_{t}\left[\frac{dc_{t}}{c_{t}} \right] + \frac{1}{2}\eta_{t} \mathbb{E}_{t}\left[\left(\frac{dc_{t}}{c_{t}}\right)^{2} \right]
\end{eqnarray*} 
inserting the expressions for $\mathbb{E}_{t}\left[\frac{dc_{t}}{c_{t}} \right]$ and $\mathbb{E}_{t}\left[\left(\frac{dc_{t}}{c_{t}}\right)^{2} \right]$ completes the proof.
\end{proof}

\subsection{Proofs: Measurement of dynamic horizon specific network risk for large dynamic networks}\label{app:proofs2}
\begin{proof}[Proposition~\ref{prop:3}]

Let us have the VMA($\infty$) representation of the locally stationary TVP VAR model  \citep{dahlhaus2009empirical,roueff2016prediction}
\begin{equation}
\bX_{t,T} = \sum_{h=-\infty}^{\infty} \bPsi_{t,T}(h)\bepsilon_{t-h}
\end{equation}
$\bPsi_{t,T}(h) \approx\bPsi(t/T,h)$ is a stochastic process satisfying $\sup_{\ell} ||\bPsi_t-\bPsi_{\ell}||^2 = O_p(h/t)$ for $1\le h \le t$ as $t\rightarrow \infty$, hence in a neighborhood of a fixed time point $u=t/T$ the process $\bX_{t,T}$ can be approximated by a stationary process $\widetilde{\bX}_t(u)$
\begin{equation}
\widetilde{\bX}_t(u) = \sum_{h=-\infty}^{\infty} \bPsi(u,h)\bepsilon_{t-h}
\end{equation}
with $\bepsilon$ being \textit{iid} process with $\mathbb{E}[\bepsilon_t]=0$, $\mathbb{E}[\bepsilon_s\bepsilon_t]=0$ for all $s\ne t$, and the local covariance matrix of the errors $\bSigma(u)$. Under suitable regularity conditions $|\bX_{t,T} - \widetilde{\bX}_t(u)| = O_p\big( |t/T-u|+1/T\big)$.

Since the errors are assumed to be serially uncorrelated, the total local covariance matrix of the forecast error conditional on the information at time $t-1$ is given by
\begin{equation}
	\bOmega(u,H) = \sum_{h=0}^{H} \bPsi(u,h) \bSigma(u) \bPsi^{\top}(u,h).
\end{equation}
Next, we consider the local covariance matrix of the forecast error conditional on knowledge of today's shock and future expected shocks to $k$-th variable. Starting from the conditional forecasting error,
\begin{equation}
\label{eq:LGFEVD1}
	\boldsymbol \xi^k(u,H) = \sum_{h=0}^{H} \bPsi(u,h) \Big[ \bepsilon_{t+H-h} - \mathbb{E}(\bepsilon_{t+H-h} | \bepsilon_{k,t+H-h})  \Big],
\end{equation} assuming normal distribution of $\bepsilon_t\sim N(0,\bSigma)$, we obtain\footnote{Note to notation: $[\boldsymbol A]_{j,k}$ denotes the $j$th row and $k$th column of matrix $\boldsymbol A$ denoted in bold. $[\boldsymbol A]_{j,\cdot}$ denotes the full $j$th row; this is similar for the columns. A $\sum A$, where $A$ is a matrix that denotes the sum of all elements of the matrix $A$.}  
\begin{equation}
\label{eq:LGFEVD2}
\mathbb{E}(\bepsilon_{t+H-h} | \bepsilon_{k,t+H-h}) = \sigma_{kk}^{-1} \Big[\bSigma(u) \Big]_{\cdot k} \bepsilon_{k,t+H-h}
\end{equation}
and substituting (\ref{eq:LGFEVD2}) to (\ref{eq:LGFEVD1}), we obtain
\begin{equation}
	\boldsymbol \xi^k(u,H) = \sum_{h=0}^{H} \bPsi(u,h) \Big[ \bepsilon_{t+H-h} - \sigma_{kk}^{-1} \Big[\bSigma(u) \Big]_{\cdot k} \bepsilon_{k,t+H-h}  \Big].
\end{equation}
Finally, the local forecast error covariance matrix is
\begin{equation}
	\bOmega^k(u,H) = \sum_{h=0}^{H} \bPsi(u,h) \bSigma(u) \bPsi^{\top}(u,h) - \sigma_{kk}^{-1} \sum_{h=0}^{H} \bPsi(u,h) \Big[\bSigma(u) \Big]_{\cdot k} \Big[\bSigma(u) \Big]_{\cdot k}^{\top} \bPsi^{\top}(u,h).
\end{equation} 
Then
\begin{equation}
	\Big[\boldsymbol \Delta(u,H)\Big]_{(j)k} = \Big[\bOmega(u,H) - \bOmega^k(u,H)\Big]_{j,j} = \sigma_{kk}^{-1} \sum_{h = 0}^{H} \Bigg( \Big[\bPsi(u,h) \bSigma(u)\Big]_{j,k} \Bigg)^2
\end{equation} is the unscaled local $H$-step ahead forecast error variance of the $j$-th component with respect to the innovation in the $k$-th component. Scaling the equation with $H$-step ahead forecast error variance with respect to the $j$th variable yields the desired time varying generalized forecast error variance decompositions (TVP GFEVD)
\begin{equation}
\label{eq:tvpgfevd}
	\Big[ \btheta(u,H) \Big]_{j,k} = \frac{\sigma_{kk}^{-1}\displaystyle\sum_{h = 0}^{H}\Bigg( \Big[\bPsi(u,h) \bSigma(u)\Big]_{j,k} \Bigg)^2}{\displaystyle \sum_{h=0}^{H} \Big[ \bPsi(u,h) \bSigma(u) \bPsi^{\top}(u,h) \Big]_{j,j}}
\end{equation}

Next, we derive the frequency representation of the quantity in (\ref{eq:tvpgfevd}) using the fact that unique time varying spectral density of $\bX_{t,T}$ at frequency $\omega$ which is locally the same as the spectral density of $\widetilde{\bX}_t(u)$ at $u=t/T$ can be defined as a Fourier transform of VMA($\infty$) filtered series over frequencies $\omega\in(-\pi,\pi)$ as
\begin{equation}
\boldsymbol S_{\bX}(u,\omega) = \sum_{h=-\infty}^{\infty}\mathbb{E}\Big[ \bX_{t+h}(u)\bX_{t}^{\top}(u)  \Big]e^{-i\omega h}=\Big\{\bPsi(u)e^{-i\omega}\Big\}\bSigma(u)\Big\{\bPsi(u)e^{+i\omega}\Big\}^{\top},
\end{equation}
where we consider a time varying frequency response function $\bPsi(u)e^{-i\omega} = \sum_h e^{-i\omega h} \bPsi(u,h)$ which can be obtained as a Fourier transform of the coefficients with $i=\sqrt{-1}$.

Letting $H\rightarrow \infty$, we have time varying generalized forecast error variance decompositions
\begin{equation}
\label{eq:tvpgfevdinfty}
	\Big[ \btheta(u,\infty) \Big]_{j,k} = \frac{\sigma_{kk}^{-1}\displaystyle\sum_{h = 0}^{\infty}\Bigg( \Big[\bPsi(u,h) \bSigma(u)\Big]_{j,k} \Bigg)^2}{\displaystyle \sum_{h=0}^{\infty} \Big[ \bPsi(u,h) \bSigma(u) \bPsi^{\top}(u,h) \Big]_{j,j}} = \frac{\mathcal{A}}{\mathcal{B}}.
\end{equation}

Starting with frequecy domain counterpart of the nominator $\mathcal{A}$,  we will use the standard integral
	\begin{equation}
		\frac{1}{2\pi} \int_{-\pi}^{\pi} e^{i \omega (r - v)} d \omega =
			\begin{cases}
				1 & \text{ for } r = v\\
				0 & \text{ for } r \neq v.
			\end{cases}
	\end{equation}

	Using the fact that $\sum_{h=0}^{\infty} \phi(h) \psi(h) = \frac{1}{2\pi} \int_{-\pi}^{\pi} \sum_{v=0}^{\infty} \sum_{r=0}^{\infty} \phi(r) \psi(v) e^{i \omega (r - v)} d \omega$, we can rewrite (\ref{eq:tvpgfevdinfty}) as

	\begin{align}
		\begin{split}
			\sigma_{kk}^{-1}\displaystyle\sum_{h = 0}^{\infty} & \Bigg( \Big[\bPsi(u,h) \bSigma(u)\Big]_{j,k} \Bigg)^2
			= \sigma_{kk}^{-1} \sum_{h=0}^{\infty} \left( \sum_{z=1}^{n} \Big[ \bPsi(u,h) \Big]_{j,z} \Big[ \bSigma(u) \Big]_{z,k} \right)^2 \\
			& = \sigma_{kk}^{-1} \frac{1}{2\pi} \int_{-\pi}^{\pi} \sum_{r = 0}^{\infty} \sum_{v = 0}^{\infty} \left( \sum_{x = 1}^{n} \Big[\bPsi(u,r) \Big]_{j,x} \Big[ \bSigma(u) \Big]_{x,k} \right)\left( \sum_{y = 1}^{n} \Big[ \bPsi(u,v) \Big]_{j,y} \Big[ \bSigma(u) \Big]_{y,k} \right) e^{i \omega (r-v)} d \omega\\
			& = \sigma_{kk}^{-1} \frac{1}{2\pi} \int_{-\pi}^{\pi} \sum_{r = 0}^{\infty} \sum_{v = 0}^{\infty} \left( \sum_{x = 1}^{n} \Big[ \bPsi(u,r) e^{i \omega r} \Big]_{j,x} \Big[ \bSigma(u) \Big]_{x,k} \right)\left( \sum_{y = 1}^{n} \Big[ \bPsi(u,v) e^{-i \omega v} \Big]_{j,y} \Big[ \bSigma(u) \Big]_{y,k} \right) d \omega\\
			& = \sigma_{kk}^{-1} \frac{1}{2\pi} \int_{-\pi}^{\pi} \left( \sum_{r = 0}^{\infty} \sum_{x = 1}^{n} \Big[ \bPsi(u,r) e^{i \omega r} \Big]_{j,x} \Big[ \bSigma(u) \Big]_{x,k} \right)\left( \sum_{v = 0}^{\infty} \sum_{y = 1}^{n} \Big[ \bPsi(u,v) e^{-i \omega v} \Big]_{j,y} \Big[ \bSigma(u) \Big]_{y,k} \right) d \omega \\
			& = \sigma_{kk}^{-1} \frac{1}{2\pi} \int_{-\pi}^{\pi} \left( \sum_{x = 1}^{n} \Big[ \bPsi(u) e^{i \omega} \Big]_{j,x} \Big[ \bSigma(u) \Big]_{x,k} \right)\left( \sum_{y = 1}^{n} \Big[ \bPsi(u) e^{-i \omega} \Big]_{j,y} \Big[ \bSigma(u) \Big]_{y,k} \right) d \omega \\
			& = \sigma_{kk}^{-1} \frac{1}{2\pi} \int_{-\pi}^{\pi} \left( \Big[ \bPsi(u) e^{ - i \omega} \bSigma(u) \Big]_{j,k} \right) \left( \Big( \bPsi(u) e^{ i \omega} \bSigma(u) \Big]_{j,k} \right) d \omega \\
			& = \sigma_{kk}^{-1} \frac{1}{2\pi} \int_{-\pi}^{\pi} \left| \bigg[ \bPsi(u) e^{-i \omega} \bSigma(u) \bigg]_{j,k} \right|^2 d \omega
		\end{split}
	\end{align}
	Hence we have established that 
	\begin{equation}
	\mathcal{A}=\sigma_{kk}^{-1}\displaystyle\sum_{h = 0}^{\infty} \Bigg( \Big[\bPsi(u,h) \bSigma(u)\Big]_{j,k} \Bigg)^2=\sigma_{kk}^{-1} \frac{1}{2\pi} \int_{-\pi}^{\pi} \left| \bigg[ \bPsi(u) e^{-i \omega} \bSigma(u) \bigg]_{j,k} \right|^2 d \omega
	\end{equation}
	from (\ref{eq:tvpgfevdinfty}), we use the local spectral representation of the VMA coefficients in the second step. The rest is a manipulation with the last step invoking the definition of modulus squared of a complex number to be defined as $|z|^2 = z z^*$. Note that we can use this simplification without loss of generality, because the $VMA(\infty)$ representation that is described by the coefficients $\boldsymbol \Psi(u,h)$ has a spectrum that is always symmetric.

	Next, we concentrate on $\mathcal{B}$ from (\ref{eq:tvpgfevdinfty}). Using similar steps and the positive semidefiniteness of the matrix $\bSigma(u)$ that ascertains that there exists $\bP(u)$ such that $\bSigma(u) = \bP(u) \bP^{\top}(u).$

	\begin{align}
		\begin{split}
			\sum_{h=0}^{\infty} \Big[ \bPsi(u,h) \bSigma(u) \bPsi^{\top}(u,h) \Big]
			& = \sum_{h = 0}^{\infty} \Big[ \bPsi(u,h) \bP(u) \Big] \Big[ \bPsi(u,h) \bP(u) \Big]^{\top} \\
			& = \frac{1}{2\pi} \int_{-\pi}^{\pi} \sum_{r = 0}^{\infty} \sum_{v = 0}^{\infty} \Big[ \bPsi(u,r) e^{i \omega r} \bP(u) \Big] \Big[ \bPsi(u,v) e^{-i \omega v} \bP(u) \Big]^{\top} d \omega \\
			& = \frac{1}{2\pi} \int_{-\pi}^{\pi} \sum_{r = 0}^{\infty} \Big[ \bPsi(u,r) e^{i \omega r} \bP(u) \Big] \sum_{v = 0}^{\infty} \Big[ \bPsi(u,v) e^{-i \omega v} \bP(u) \Big]^{\top} d \omega \\
			& = \frac{1}{2\pi} \int_{-\pi}^{\pi} \Big[ \bPsi(u) e^{i \omega} \bP(u) \Big] 
			\Big[ \bPsi(u) e^{-i \omega } \bP(u) \Big]^{\top} d \omega \\
			& = \frac{1}{2\pi} \displaystyle \int_{-\pi}^{\pi} \Bigg[ \Big\{\bPsi(u) e^{i \omega} \Big\}\bSigma(u) \Big\{ \bPsi(u) e^{-i \omega } \Big\}^{\top}  \Bigg] d \omega
		\end{split}
	\end{align}

That establishes the fact that  
	\begin{equation}
	\mathcal{B}=\displaystyle \sum_{h=0}^{\infty} \Big[ \bPsi(u,h) \bSigma(u) \bPsi^{\top}(u,h) \Big]_{j,j}
	=\frac{1}{2\pi} \displaystyle \int_{-\pi}^{\pi} \Bigg[ \Big\{\bPsi(u) e^{i \omega} \Big\}\bSigma(u) \Big\{ \bPsi(u) e^{-i \omega } \Big\}^{\top}  \Bigg]_{j,j} d \omega
	\end{equation}
	from (\ref{eq:tvpgfevdinfty}), and we have shown that 

\begin{equation}
	\Big[ \btheta(u,\infty) \Big]_{j,k} = \frac{\sigma_{kk}^{-1}\displaystyle\sum_{h = 0}^{\infty}\Bigg( \Big[\bPsi(u,h) \bSigma(u)\Big]_{j,k} \Bigg)^2}{\displaystyle \sum_{h=0}^{\infty} \Big[ \bPsi(u,h) \bSigma(u) \bPsi^{\top}(u,h) \Big]_{j,j}} = \frac{\sigma_{kk}^{-1} \displaystyle \int_{-\pi}^{\pi} \left| \bigg[ \bPsi(u) e^{-i \omega} \bSigma(u) \bigg]_{j,k} \right|^2 d \omega}{ \displaystyle \int_{-\pi}^{\pi} \Bigg[ \Big\{\bPsi(u) e^{i \omega} \Big\}\bSigma(u) \Big\{ \bPsi(u) e^{-i \omega } \Big\}^{\top}  \Bigg]_{j,j} d \omega}
\end{equation}

Finally, focusing on a frequency band $d=(a,b): a,b \in (- \pi, \pi), a < b$, we have 

\begin{equation}
	\Big[ \btheta(u,d) \Big]_{j,k} = \frac{\sigma_{kk}^{-1} \displaystyle \int_{a}^{b} \left| \bigg[ \bPsi(u) e^{-i \omega} \bSigma(u) \bigg]_{j,k} \right|^2 d \omega}{ \displaystyle \int_{-\pi}^{\pi} \Bigg[ \Big\{\bPsi(u) e^{i \omega} \Big\}\bSigma(u) \Big\{ \bPsi(u) e^{-i \omega } \Big\}^{\top}  \Bigg]_{j,j} d \omega}
\end{equation}

	This completes the proof.
\end{proof}

\section{Estimation of the time-varying parameter VAR model} \label{app:estimate}

To estimate our high dimensional systems, we follow the Quasi-Bayesian Local-Liklihood (QBLL) approach of \cite{petrova2019quasi}. let $\bX_{t}$ be an $N \times 1$ vector generated by a stable time-varying parameter (TVP) heteroskedastic VAR model with $p$ lags:

\begin{equation}\label{eq:VAR}
\bX_{t,T}=\bPhi_{1}(t/T)\bX_{t-1,T}+\ldots+\bPhi_{p}(t/T)\bX_{t-p,T} + \bepsilon_{t,T},
\end{equation}
where $\bepsilon_{t,T}=\bSigma^{-1/2}(t/T)\bbeta_{t,T}$ with $\bbeta_{t,T}\sim NID(0,\boldsymbol{I}_M)$ and $\bPhi(t/T)=(\bPhi_{1}(t/T),\ldots,\bPhi_{p}(t/T))^{\top}$ are the time varying autoregressive coefficients.
%\begin{eqnarray}\label{eq:VAR}
%y_{t} = \mathbf{B}_{0,t} + \sum^{L}_{p=1}\mathbf{B}_{p,t}y_{t-p} + \varepsilon_{t},\:\:\varepsilon_{t}=\mathbf{\Xi}^{-\frac{1}{2}}_{t}\kappa_{t},\:\: \kappa_{t}\backsim \text{NID}(0,\textbf{I}_{N})
%\end{eqnarray}
%where $\mathbf{B}_{0,t},\mathbf{B}_{p,t}$ contain the time-varying intercepts and autoregressive matrices, respectively. 
Note that all roots of the polynomial, $\chi(z)=\text{det}\left(\textbf{I}_{N}-\sum^{L}_{p=1}z^{p}\mathbf{B}_{p,t}\right)$, lie outside the unit circle, and $\bSigma^{-1}_{t}$ is a positive definite time-varying covariance matrix. Stacking the time-varying intercepts and autoregressive matrices in the vector $\phi_{t,T}$ with $\bar{\bX}'_{t} = \left(\text{\textbf{I}}_{N} \otimes x_{t}\right),\: x_{t}=\left(1,x'_{t-1},\dots,x'_{t-p}\right)$ and $\otimes$ denotes the Kronecker product, the model can be written as:
\begin{eqnarray}
\bX_{t,T} = \bar{\bX}'_{t,T}\phi_{t,T} + \bSigma^{-\frac{1}{2}}_{t/T}\bbeta_{t,T}
\end{eqnarray}
We obtain the time-varying parameters of the model by employing Quasi-Bayesian Local Likelihood (QBLL) methods. Estimation of (\ref{eq:VAR}) requires re-weighting the likelihood function. Essentially, the weighting function gives higher proportions to observations surrounding the time period whose parameter values are of interest. The local likelihood function at time period $k$ is given by:
\begin{eqnarray}
\text{L}_{k}\left(\bX|\theta_{k},\bSigma_{k},\bar{\bX} \right)&\propto& |\bSigma_{k}|^{\text{trace}(\mathbf{D}_{k})/2}\exp\left\{-\frac{1}{2}(\bX-\bar{\bX}'\phi_{k})'\left(\bSigma_{k}\otimes\mathbf{D}_{k}\right)(\bX-\bar{\bX}'\phi_{k})\right\}
\end{eqnarray}
The $\mathbf{D}_{k}$ is a diagonal matrix whose elements hold the weights:
\begin{eqnarray}
\mathbf{D}_{k} &=& \text{diag}(\varrho_{k1},\dots,\varrho_{kT})\\
\varrho_{kt} &=& \phi_{T,k}w_{kt}/\sum^{T}_{t=1}w_{kt}\\
\label{eq:weight}
w_{kt} &=& (1/\sqrt{2\pi})\exp((-1/2)((k-t)/H)^{2}),\quad\text{for}\: k,t\in\{1,\dots,T\}\\
\zeta_{Tk} &=& \left(\left(\sum^{T}_{t=1}w_{kt}\right)^{2}\right)^{-1}
\end{eqnarray}
where $\varrho_{kt}$ is a normalised kernel function. $w_{kt}$ uses a Normal kernel weighting function. $\zeta_{Tk}$ gives the rate of convergence and behaves like the bandwidth parameter $H$ in (\ref{eq:weight}), and it is the kernel function that provides greater weight to observations surrounding the parameter estimates at time $k$ relative to more distant observations.
 
Using a Normal-Wishart prior distribution for $\phi_{k}|\:\bSigma_{k}$ for $k\in\{1,\dots,T\}$:
\begin{eqnarray}
\phi_{k}|\bSigma_{k} \backsim \mathcal{N}\left(\phi_{0k},(\bSigma_{k} \otimes \mathbf{\Xi}_{0k})^{-1}\right)\\
\bSigma_{k} \backsim \mathcal{W}\left(\alpha_{0k},\mathbf{\Gamma}_{0k}\right)
\end{eqnarray}
where $\phi_{0k}$ is a vector of prior means, $\mathbf{\Xi}_{0k}$ is a positive definite matrix, $\alpha_{0k}$ is a scale parameter of the Wishart distribution ($\mathcal{W}$), and $\mathbf{\Gamma}_{0k}$ is a positive definite matrix. 

The prior and weighted likelihood function implies a Normal-Wishart quasi posterior distribution for $\phi_{k}|\:\bSigma_{k}$ for $k=\{1,\dots,T\}$. Formally let $\mathbf{A} = (\bar{x}'_{1},\dots,\bar{x}'_{T})'$ and $\mathbf{Y}=(x_{1},\dots,x_{T})'$ then:
\begin{eqnarray}
\phi_{k}|\bSigma_{k},\mathbf{A},\mathbf{Y} &\backsim & \mathcal{N}\left(\tilde{\theta}_{k},\left(\bSigma_{k}\otimes\mathbf{\tilde{\Xi}}_{k}\right)^{-1}\right)\\
\bSigma_{k} &\backsim & \mathcal{W}\left(\tilde{\alpha}_{k},\mathbf{\tilde{\Gamma}}^{-1}_{k} \right)
\end{eqnarray}
with quasi posterior parameters
\begin{eqnarray}
\tilde{\phi}_{k} &=& \left(\mathbf{I}_{N}\otimes \mathbf{\tilde{\Xi}}^{-1}_{k}\right)\left[\left(\mathbf{I}_{N}\otimes \mathbf{A}'\mathbf{D}_{k}\mathbf{A}\right)\hat{\phi}_{k}+ \left(\mathbf{I}_{N}\otimes \mathbf{\Xi}_{0k}\right)\phi_{0k} \right]\\
\mathbf{\tilde{\Xi}}_{k} &=& \mathbf{\tilde{\Xi}}_{0k} + \mathbf{A}'\mathbf{D}_{k}\mathbf{A}\\
\tilde{\alpha}_{k} &=& \alpha_{0k}+\sum^{T}_{t=1}\varrho_{kt}\\
\mathbf{\tilde{\Gamma}}_{k} &=& \mathbf{\Gamma}_{0k} + \mathbf{Y}'\mathbf{D}_{k}\mathbf{Y} + \mathbf{\Phi}_{0k}\mathbf{\Gamma}_{0k}\mathbf{\Phi}'_{0k} - \mathbf{\tilde{\Phi}}_{k}\mathbf{\tilde{\Gamma}}_{k}\mathbf{\tilde{\Phi}}'_{k}
\end{eqnarray}
where $\hat{\phi}_{k} = \left(\mathbf{I}_{N}\otimes \mathbf{A}'\mathbf{D}_{k}\mathbf{A}\right)^{-1}\left(\mathbf{I}_{N} \otimes \mathbf{A}'\mathbf{D}_{k}\right)y$ is the local likelihood estimator for $\phi_{k}$. The matrices $\mathbf{\Phi}_{0k},\:\mathbf{\tilde{\Phi}}_{k}$ are conformable matrices from the vector of prior means, $\phi_{0k}$, and a draw from the quasi posterior distribution, $\tilde{\phi}_{k}$, respectively.

The motivation for employing these methods are threefold. First, we are able to estimate large systems that conventional Bayesian estimation methods do not permit. This is typically because the state-space representation of an $N$-dimensional TVP VAR ($p$) requires an additional $N(3/2 + N(p+1/2))$ state equations for every additional variable. Conventional Markov Chain Monte Carlo (MCMC) methods fail to estimate larger models, which in general confine one to (usually) fewer than 6 variables in the system. Second, the standard approach is fully parametric and requires a law of motion. This can distort inference if the true law of motion is misspecified. Third, the methods used here permit direct estimation of the VAR's time-varying covariance matrix, which has an inverse-Wishart density and is symmetric positive definite at every point in time. 

In estimating the model, we use $p$=2 and a Minnesota Normal-Wishart prior with a shrinkage value $\varphi=0.05$ and centre the coefficient on the first lag of each variable to 0.1 in each respective equation. The prior for the Wishart parameters are set following \cite{kadiyala1997numerical}. For each point in time, we run 500 simulations of the model to generate the (quasi) posterior distribution of parameter estimates. Note we experiment with various lag lengths, $p=\{2,3,4,5\}$; shrinkage values, $\varphi=\{0.01, 0.25, 0.5\}$; and values to centre the coefficient on the first lag of each variable, $\{0, 0.05, 0.2, 0.5\}$. Network measures from these experiments are qualitatively similar. Notably, adding lags to the VAR  and increasing the persistence in the prior value of the first lagged dependent variable in each equation increases computation time.

%--------------

\end{document}